\documentclass[a4paper,fleqn]{cas-sc}

\usepackage[numbers,sort&compress]{natbib}
\usepackage{amsmath,amssymb}
\usepackage{booktabs}
\usepackage{array}
\usepackage{longtable}
\usepackage{amsthm}
\usepackage{float}
\usepackage[section]{placeins}

\newtheorem{lemma}{Lemma}
\newtheorem{proposition}{Proposition}

\usepackage{tikz}
\usetikzlibrary{positioning,calc,backgrounds}
\usepackage{pgfplots}
\pgfplotsset{compat=1.18}
\usepgfplotslibrary{groupplots}

\def\tsc#1{\csdef{#1}{\textsc{\lowercase{#1}}\xspace}}
\tsc{WGM}
\tsc{QE}

\begin{document}
\let\WriteBookmarks\relax
\def\floatpagefraction{1}
\def\textpagefraction{.001}

\shorttitle{Exact SAT Solving for 2DBMP}
\shortauthors{Pham Quang et al.}

\title [mode = title]{Exact SAT Solving for the Two-Dimensional Bandwidth Minimization Problem}

\author[]{Pham Quang Minh}[style=chinese,orcid=0009-0008-3591-3076]
\ead{24020239@vnu.edu.vn}

\author[]{Dao Xuan Nghia}[style=chinese,orcid=0009-0001-9628-636X]
\ead{25025081@vnu.edu.vn}

\author[]{To Van Khanh}[style=chinese,orcid=0009-0008-1907-7848]
\cormark[1]
\ead{khanhtv@vnu.edu.vn}

\affiliation{organization={Faculty of Information Technology, VNU University of Engineering and Technology},
            city={Hanoi},
            country={Vietnam}}

\cortext[1]{Corresponding author}

\begin{abstract}
The two-dimensional bandwidth minimization problem (2DBMP) seeks an injective embedding of a guest graph into a square grid that minimizes the maximum Manhattan distance over its edges. Heuristic methods can provide strong upper bounds, but these bounds do not by themselves certify optimality. We present an efficient exact SAT-based approach for 2DBMP that incrementally searches for the minimum feasible bandwidth and certifies optimality through satisfiability and unsatisfiability results. On the standard $\lceil\sqrt n\rceil\times\lceil\sqrt n\rceil$ host grid, under a 3600~s time limit, the proposed SAT approach certifies optimal bandwidths for 41 of 43 Regular instances and 42 of 93 Harwell--Boeing instances, 
achieving substantially broader optimality certification within the 3600 s time limit than a previous exact approach evaluated with a 72-hour time limit. In addition, it certifies three bandwidth values that improve all previously published comparison values considered in this study and establishes all three as optimal. We further evaluate the approach on alternative host geometries, namely $2\times\lceil n/2\rceil$ and $n\times n$ grids, to assess its effectiveness beyond the standard host. Overall, the results demonstrate that the proposed SAT approach provides an effective exact method for the small- and medium-sized benchmark instances considered in this study, with fewer than 400 vertices, while heuristic methods remain important for larger and more challenging instances.
\end{abstract}



\maketitle

\section{Introduction}

Graph layout is an important family of combinatorial optimization problems concerned with arranging the vertices of a graph on a target structure under different optimization objectives. This family of problems has broad practical relevance, with representative applications in VLSI design, telecommunication architectures, parallel processing systems, sparse matrix ordering, and matrix decomposition~\citep{Chung1988,LinLin2010,RodriguezGarcia2021}. Within this family, an important line of research focuses on minimizing the maximum host-graph distance induced between adjacent vertices. The two-dimensional bandwidth minimization problem (2DBMP) is a representative problem in this line because a graph is arranged on a two-dimensional grid and edge lengths are measured by the Manhattan distance.

The problem considered in this paper is 2DBMP on the standard square grid of side $h=\lceil\sqrt n\rceil$~\citep{RodriguezGarcia2021,Cavero2023}. Given an input graph $G=(V,E)$ with $n$ vertices, each vertex must be assigned to a distinct grid position. The objective is to minimize the maximum Manhattan distance between adjacent vertices after the graph is placed on the grid. In the classical one-dimensional bandwidth minimization problem, each vertex receives one label and the distance between adjacent vertices is the absolute difference between labels. In 2DBMP, each vertex is identified by two grid coordinates. Consequently, moving from a path host to a two-dimensional grid changes both the geometry of feasible embeddings and the structure of the distance constraints.

The literature on 2DBMP has progressed along complementary theoretical, exact, and heuristic lines~\citep{LinLin2010,RodriguezGarcia2021,Cavero2023,Zhou2026}. Theoretical studies established lower and upper bounds and exact values for selected graph classes under broader host-grid formulations~\citep{Lin1996,LinLin2010,LinLin2011}. On the heuristic side, recent methods have continued to improve the upper bounds reported for benchmark instances~\citep{TorresJimenez2025,Zhou2026}. These methods terminate according to algorithm-specific stopping criteria, such as time or iteration limits, and return the best solution found~\citep{Cavero2023,TorresJimenez2025,Zhou2026}. Because these heuristic searches do not provide certificates that exclude better embeddings, an improved best-known value is not necessarily a certified optimal value. 
Exact methods complement heuristic search by providing optimality certificates, but they generally require substantially greater computational effort in time and memory~\citep{RodriguezGarcia2021}. Both exact and heuristic approaches have been studied for the Manhattan-distance problem~\citep{RodriguezGarcia2021} and the related Chebyshev-distance variant~\citep{Cavero2026Linf}. Prior studies have also considered alternative host-grid dimensions, including $2\times\lceil n/2\rceil$ rectangular grids~\citep{Khandelwal2023SAGBMP} and $n\times n$ grids~\citep{LinLin2010,LinLin2011}. These formulations broaden the range of host-grid settings and motivate further study of alternative host geometries. On the widely used Regular and Harwell--Boeing benchmarks~\citep{RodriguezGarcia2021,Cavero2023,Zhou2026}, the exact results reported by~\citet{RodriguezGarcia2021} certified most Regular instances but left optimality unresolved for most Harwell--Boeing instances. Consequently, tightening bounds on unresolved instances and certifying additional optimal values remain important goals for 2DBMP.

Modern SAT solving has become an effective approach for a wide range of practical problems~\citep{Biere2021Handbook,Fichte2023SAT}. Its applications include exact methods for graph-related problems such as graph coloring, antibandwidth, and bandwidth coloring~\citep{Niewiadomski2019SAT,Fazekas2020Antibandwidth,Faber2024BandwidthColoring}. This makes SAT a natural framework for 2DBMP because vertex placements, occupancy requirements, and bandwidth limits can be represented by Boolean variables and clauses. Testing the resulting formula for successive values of $K$ converts the optimization problem into a sequence of related SAT decision problems.

This paper addresses the certification gap with a SAT-based exact approach for standard 2DBMP. The proposed SAT encoding combines components for position assignment, sequential occupancy, cutoff-aware distance, symmetry breaking through an anchor vertex, and bandwidth enforcement.
During the descending threshold search, the incremental procedure retains the fixed clauses and solver state while adding the bandwidth component $\Delta_K$ for each tested value of $K$. Experiments on the small and medium instances considered in this study, comprising 43 Regular and 93 Harwell--Boeing instances, show that the proposed method achieves competitive exact certification relative to the locally evaluated commercial solvers and extends the exact results available in the literature. Under the 3600~s time limit on the standard $\lceil\sqrt n\rceil\times\lceil\sqrt n\rceil$ host grid, it establishes three new certified optimal values for Harwell--Boeing instances. Each value is lower than every available published comparison bandwidth considered in this study and every feasible incumbent returned by Gurobi, CPLEX MIP, and CPLEX CP for the corresponding instance.

The remainder of the paper is organized as follows. Section~\ref{sec:background} defines 2DBMP, presents an integer linear programming (ILP) formulation, and reviews the directly relevant literature. Section~\ref{sec:sat} presents the SAT encoding and the incremental solving scheme. Section~\ref{sec:experiments} reports the computational experiments. Section~\ref{sec:conclusion} concludes the paper.

\section{Background}\label{sec:background}

\subsection{Problem Definition and ILP Formulation}
\label{sec:problem}
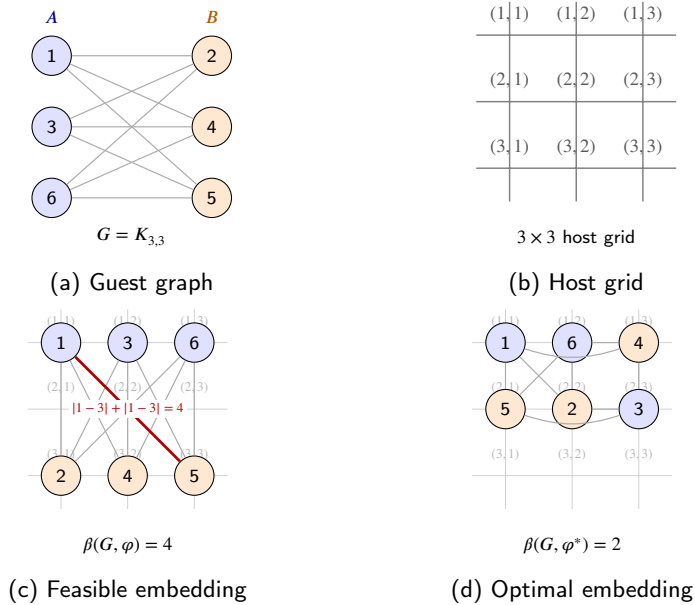
\begin{figure}[pos=htbp]
\centering
\small
\tikzset{
  vertex/.style={circle, draw=black, fill=white, minimum size=5.2mm, inner sep=0pt, font=\scriptsize},
  parta/.style={vertex, fill=blue!12},
  partb/.style={vertex, fill=orange!18},
  gridline/.style={gray!35, line width=0.3pt},
  emptygrid/.style={black!50, line width=0.55pt},
  gedge/.style={gray!65, line width=0.45pt},
  worst/.style={red!70!black, line width=1.05pt}
}
\begin{minipage}[t]{0.33\textwidth}
\centering
\begin{tikzpicture}[scale=0.94]
\node[parta] (g1) at (0,2) {1};
\node[parta] (g3) at (0,1) {3};
\node[parta] (g6) at (0,0) {6};
\node[partb] (g2) at (2.25,2) {2};
\node[partb] (g4) at (2.25,1) {4};
\node[partb] (g5) at (2.25,0) {5};
\foreach \a in {g1,g3,g6}{\foreach \b in {g2,g4,g5}{\draw[gedge] (\a) -- (\b);}}
\node[font=\scriptsize] at (1.12,-0.55) {$G=K_{3,3}$};
\node[font=\scriptsize, blue!55!black] at (0,2.55) {$A$};
\node[font=\scriptsize, orange!70!black] at (2.25,2.55) {$B$};
\end{tikzpicture}
\par\vspace{0.25em}(a) Guest graph
\end{minipage}\hspace{0.02\textwidth}
\begin{minipage}[t]{0.33\textwidth}
\centering
\begin{tikzpicture}[scale=0.88]
\foreach \r in {1,2,3}{
  \foreach \c in {1,2,3}{
    \node[font=\scriptsize, black!65] at (\c,-\r+0.32) {$(\r,\c)$};
  }
}
\draw[emptygrid] (0.5,-0.5) grid (3.5,-3.5);
\node[font=\scriptsize] at (2,-4.05) {$3\times3$ host grid};
\end{tikzpicture}
\par\vspace{0.25em}(b) Host grid
\end{minipage}

\par\vspace{0.35em}

\begin{minipage}[t]{0.33\textwidth}
\centering
\begin{tikzpicture}[scale=0.88]
\foreach \r in {1,2,3}{
  \foreach \c in {1,2,3}{
    \node[font=\tiny, gray!55] at (\c,-\r+0.32) {$(\r,\c)$};
  }
}
\draw[gridline] (0.5,-0.5) grid (3.5,-3.5);
\node[parta] (s1) at (1,-1) {1};
\node[parta] (s3) at (2,-1) {3};
\node[parta] (s6) at (3,-1) {6};
\node[partb] (s2) at (1,-3) {2};
\node[partb] (s4) at (2,-3) {4};
\node[partb] (s5) at (3,-3) {5};
\foreach \a/\b in {s3/s5,s3/s4,s3/s2,s6/s5,s6/s4,s6/s2,s1/s4,s1/s2}{\draw[gedge] (\a) -- (\b);}
\draw[worst] (s1) -- (s5);
\node[font=\tiny, red!70!black, fill=white, inner sep=1pt] at (2,-2) {$|1-3|+|1-3|=4$};
\node[font=\scriptsize] at (2,-4.05) {$\beta(G,\varphi)=4$};
\end{tikzpicture}
\par\vspace{0.25em}(c) Feasible embedding
\end{minipage}\hspace{0.02\textwidth}
\begin{minipage}[t]{0.33\textwidth}
\centering
\begin{tikzpicture}[scale=0.88]
\foreach \r in {1,2,3}{
  \foreach \c in {1,2,3}{
    \node[font=\tiny, gray!55] at (\c,-\r+0.32) {$(\r,\c)$};
  }
}
\draw[gridline] (0.5,-0.5) grid (3.5,-3.5);
\node[parta] (o1) at (1,-1) {1};
\node[parta] (o6) at (2,-1) {6};
\node[partb] (o4) at (3,-1) {4};
\node[partb] (o5) at (1,-2) {5};
\node[partb] (o2) at (2,-2) {2};
\node[parta] (o3) at (3,-2) {3};
\foreach \a/\b in {o3/o4,o3/o2,o6/o5,o6/o4,o6/o2,o1/o5,o1/o2}{\draw[gedge] (\a) -- (\b);}
\draw[gedge] (o1) to[bend right=18] (o4);
\draw[gedge] (o5) to[bend right=18] (o3);
\node[font=\scriptsize] at (2,-4.05) {$\beta(G,\varphi^*)=2$};
\end{tikzpicture}
\par\vspace{0.25em}(d) Optimal embedding
\end{minipage}
\caption{A 2DBMP example for $G=K_{3,3}$ on a $3\times3$ host grid. The feasible and optimal embeddings have bandwidths $4$ and $2$, respectively.}
\label{fig:embedding}
\end{figure}

The 2DBMP places the vertices of a guest graph at distinct positions on a two-dimensional host grid and minimizes the largest distance between the grid positions of adjacent guest vertices~\citep{LinLin2010}.
Here, the guest graph is a simple, undirected, and unweighted graph $G=(V,E)$ with $n\ge 2$. Its vertex set is $V=\{1,\ldots,n\}$, and $E$ denotes its nonempty edge set. The host grid consists of $h\times h$ positions arranged in $h$ rows and $h$ columns. In the square-root-grid setting
considered here, the grid side is $h=\lceil\sqrt n\rceil$, following recent 2DBMP studies~\citep{RodriguezGarcia2021,Cavero2023}.
Let $\varphi(v)=(x_v,y_v)$ denote the grid position assigned to vertex $v$. Equation~\eqref{eq:bandwidth} defines the bandwidth of an embedding as the maximum Manhattan distance over all guest edges. The 2DBMP seeks a feasible embedding $\varphi$ that minimizes $\beta(G,\varphi)$. The minimum bandwidth over all feasible embeddings of $G$ is denoted by $\beta^*(G)=\min_{\varphi}\beta(G,\varphi)$. Figure~\ref{fig:embedding} illustrates this setting using the guest graph $K_{3,3}$ and a $3\times3$ host grid. Panels (a) and (b) show the graph and the grid before vertex placement. Panels (c) and (d) compare a feasible embedding with bandwidth $4$ and an optimal embedding with bandwidth $2$. The highlighted edge in panel (c) has Manhattan distance $4$, which is the largest edge distance in this embedding.
\begin{equation}
\beta(G,\varphi)=\max_{(u,v)\in E}\left(|x_u-x_v|+|y_u-y_v|\right)
\label{eq:bandwidth}
\end{equation}

Equations~\eqref{eq:ilp}--\eqref{eq:ilp_domains} give an ILP formulation of the 2DBMP.
Let $H_h=\{(r,c):1\le r,c\le h\}$ denote the host-grid positions. For each
$v\in V$ and $(r,c)\in H_h$, let $a_{vrc}$ be a binary placement variable.
For each edge $e=(u,v)\in E$, the model uses four nonnegative integer auxiliary variables,
$d^{x,+}_{e},d^{x,-}_{e},d^{y,+}_{e},d^{y,-}_{e}$, to represent row and
column coordinate differences without absolute-value terms. It also uses one
integer bandwidth variable $\beta$.
\begin{align}
\min\quad & \beta \label{eq:ilp}\\
\text{s.t.}\quad
& \sum_{r=1}^{h}\sum_{c=1}^{h} a_{vrc}=1
&& \forall v\in V, \label{eq:assign_vertex}\\
& \sum_{v\in V}a_{vrc}\le 1
&& \forall (r,c)\in H_h, \label{eq:assign_cell}\\
& \left(\sum_{r=1}^{h}\sum_{c=1}^{h} r\,a_{urc}\right)
  -\left(\sum_{r=1}^{h}\sum_{c=1}^{h} r\,a_{vrc}\right)
  =d^{x,+}_{e}-d^{x,-}_{e}
&& \forall e=(u,v)\in E, \label{eq:dx_split}\\
& \left(\sum_{r=1}^{h}\sum_{c=1}^{h} c\,a_{urc}\right)
  -\left(\sum_{r=1}^{h}\sum_{c=1}^{h} c\,a_{vrc}\right)
  =d^{y,+}_{e}-d^{y,-}_{e}
&& \forall e=(u,v)\in E, \label{eq:dy_split}\\
& d^{x,+}_{e}+d^{x,-}_{e}
  +d^{y,+}_{e}+d^{y,-}_{e}\le \beta
&& \forall e\in E, \label{eq:beta_bound}\\
& a_{vrc}\in\{0,1\}
&& \forall v\in V,\ (r,c)\in H_h, \label{eq:binary_domain}\\
& d^{x,+}_{e},d^{x,-}_{e},
  d^{y,+}_{e},d^{y,-}_{e}\in\{0,\ldots,h-1\}
&& \forall e\in E, \label{eq:distance_domains}\\
& \beta\in\mathbb Z_{\ge0}. \label{eq:ilp_domains}
\end{align}
Constraints~\eqref{eq:assign_vertex} and~\eqref{eq:assign_cell} assign each
vertex to exactly one host-grid position and allow each position to contain
at most one vertex, so they enforce an injective placement.
For each vertex $v\in V$, let $x_v=\sum_{r,c} r\,a_{vrc}$ and
$y_v=\sum_{r,c} c\,a_{vrc}$, where the sums range over $(r,c)\in H_h$.
These expressions are the row and column coordinates induced by the
placement variables. For an edge $e=(u,v)$, Constraints~\eqref{eq:dx_split}--
\eqref{eq:beta_bound} linearize the condition
$|x_u-x_v|+|y_u-y_v|\le \beta$. The split variables are
nonnegative, and Constraints~\eqref{eq:dx_split} and~\eqref{eq:dy_split}
write the signed row and column differences as differences of such variables.
Lemma~\ref{lem:ilp-split-distance} in Appendix~\ref{app:ilp-linearization}
formalizes why this represents the Manhattan distance without using
absolute-value terms. Constraint~\eqref{eq:beta_bound} bounds the resulting
edge distance by $\beta$ for every guest edge. Finally,
Constraints~\eqref{eq:binary_domain}--\eqref{eq:ilp_domains} specify the
variable domains, and minimizing $\beta$ in~\eqref{eq:ilp} is therefore
equivalent to minimizing the maximum Manhattan distance over all guest edges.

\subsection{Related Work}\label{sec:related}

Research on 2DBMP comprises theoretical analyses and computational approaches based on heuristic search, exact formulations, or combinations of both~\citep{LinLin2010,RodriguezGarcia2021,Cavero2023,Cavero2026Linf}. These studies are not always directly comparable because they may adopt different distance metrics or host-grid dimensions~\citep{LinLin2010,Khandelwal2023SAGBMP,Cavero2026Linf}. Accordingly, the discussion below distinguishes the standard square-root-grid $L_1$ 2DBMP considered in this paper from $L_\infty$ formulations and alternative host grids.

Theoretical studies established bounds and exact values for selected graph classes under broader host-grid formulations. For the $L_1$ model,~\citet{Lin1996} derived density-based lower bounds and a general upper bound based on the minimum $L_1$-diameter of $n$ grid points. The corresponding $L_\infty$ formulation and exact results for several graph classes were developed by~\citet{LinHaoLi2000}. The $L_1$ and $L_\infty$ models were subsequently studied within a common framework by~\citet{LinLin2010}, while~\citet{LinLin2011} investigated the square-root relationship between one-dimensional and two-dimensional bandwidth. These studies use $n\times n$ or otherwise sufficiently large host grids rather than the compact $\lceil\sqrt n\rceil\times\lceil\sqrt n\rceil$ host used for the benchmark setting considered in this paper. Their results provide theoretical foundations for 2DBMP, but they do not directly establish optimal values for the compact-grid benchmark instances.

Heuristic methods have been used to construct feasible embeddings and obtain strong upper bounds without requiring an optimality proof. For the $L_1$ problem,~\citet{RodriguezGarcia2018GRASP} proposed a GRASP-based approach for square-grid hosts. Subsequent methods for the compact square-root-grid setting include the BVNS heuristic developed alongside the exact constraint satisfaction problem (CSP) models of~\citet{RodriguezGarcia2021}, iterated greedy with local search~\citep{Cavero2023}, dual-representation simulated annealing~\citep{TorresJimenez2025}, and reinforcement-learning-driven tabu search~\citep{Zhou2026}. Related $L_1$ studies considered simulated annealing on a $2\times\lceil n/2\rceil$ host~\citep{Khandelwal2023SAGBMP} and RVNS on balanced rectangular hosts~\citep{Khandelwal2023RVNS}. For the $L_\infty$ variant,~\citet{Cavero2026Linf} combined a multi-start heuristic with exact CSP formulations. Although these heuristic methods can produce high-quality feasible embeddings, the resulting bandwidths remain upper bounds or best-known values unless an exact method separately certifies optimality.

Exact methods complement heuristic search by seeking both a feasible embedding and a certificate that no embedding with a smaller bandwidth exists. Within this line of research,~\citet{RodriguezGarcia2021} developed CSP models for the standard $L_1$ 2DBMP and evaluated them on established benchmark instances. More recently,~\citet{Cavero2026Linf} developed CSP models for the $L_\infty$ variant of 2DBMP, in which edge lengths are measured using Chebyshev distance rather than the $L_1$ Manhattan distance considered in this study. To reduce the range of bandwidth values examined by the exact search, their procedure uses the maximum vertex degree to identify a lower bound below which no feasible embedding can exist. This lower bound is obtained by counting the distinct grid positions available around a vertex within each Chebyshev distance. These studies establish CSP as an exact approach for both distance metrics. However, $L_1$ and $L_\infty$ define different optimization objectives even on the same graph instances, so optimality under Chebyshev distance does not certify optimality under Manhattan distance. The reported exact results for the $L_1$ setting also leave many larger benchmark instances unresolved~\citep{RodriguezGarcia2021}, motivating the investigation of alternative exact formulations.

Overall, prior work provides theoretical bounds, heuristic upper bounds, and exact methods. 
However, there is still a gap between obtaining high-quality feasible embeddings and certifying their optimality in the standard square-root-grid Manhattan-distance setting.
This limitation motivates the alternative SAT-based exact formulation  developed in the next section.
\section{SAT-Based Exact Approach}
\label{sec:sat}
\subsection{SAT Encoding Overview}
The SAT approach reduces the 2DBMP to a sequence of decision problems, each associated with a bandwidth threshold $K$. For a fixed $K$, the method constructs a SAT formula representing the existence of a feasible embedding $\varphi$ such that $\beta(G,\varphi)\le K$. After obtaining a feasible embedding, the method sets the next threshold $K$ to one less than the bandwidth of that embedding. If the resulting formula is unsatisfiable, no feasible embedding with a smaller bandwidth exists. Therefore, the bandwidth of the current embedding is the optimal value $\beta^*(G)$.

The 2DBMP constraints are encoded by the position, occupancy, distance, and bandwidth components of the SAT formula in Equation~\eqref{eq:full_encoding}, while the optional symmetry component is added to reduce the search space. The position component $\mathcal{F}_{\mathrm{pos}}$ assigns one row and one column to each vertex, while the occupancy component $\mathcal{F}_{\mathrm{occ}}^{\mathrm{seq}}$ prevents two vertices from occupying the same grid position. The distance component $\mathcal{F}_{\mathrm{dist}}^{U^*}$ represents coordinate separations up to the cutoff $U^*$. The symmetry component $\mathcal{F}_{\mathrm{sym}}$ removes a safe subset of symmetry-equivalent placements by restricting a selected anchor vertex. The bandwidth component $\Delta_K$ contains the clauses that enforce the current bandwidth threshold. Thus, when $K$ decreases, the position, occupancy, distance, and symmetry components remain unchanged, while only the bandwidth component $\Delta_K$ changes.
\begin{equation}
\mathcal{F}_{\mathrm{full}}(K)=
\mathcal{F}_{\mathrm{pos}}
\land \mathcal{F}_{\mathrm{occ}}^{\mathrm{seq}}
\land \mathcal{F}_{\mathrm{dist}}^{U^*}
\land \mathcal{F}_{\mathrm{sym}}
\land \Delta_K
\label{eq:full_encoding}
\end{equation}

Among the component roles in Equation~\eqref{eq:full_encoding}, position, occupancy, distance, and bandwidth enforcement are required to represent a bandwidth-$K$ embedding, whereas symmetry breaking is an optional search optimization. The symmetry component $\mathcal{F}_{\mathrm{sym}}$ can therefore be removed without changing which bandwidth values are feasible. The required occupancy component can be represented by the sequential occupancy encoding $\mathcal{F}_{\mathrm{occ}}^{\mathrm{seq}}$, which uses cell-indicator variables and sequential-counter constraints, or by the pairwise occupancy encoding $\mathcal{F}_{\mathrm{occ}}^{\mathrm{pair}}$. Similarly, the required distance component can use the cutoff-aware form $\mathcal{F}_{\mathrm{dist}}^{U^*}$ or the full form $\mathcal{F}_{\mathrm{dist}}$. The cutoff-aware form preserves exactness only when $U^*$ is valid. The required bandwidth component $\Delta_K$ is shared by both solving schemes. The incremental scheme retains the solver state as $K$ decreases, whereas the non-incremental scheme rebuilds the solver for each value of $K$.

To construct the SAT formula, the method uses several types of Boolean variables. 
For each vertex $v\in V$ and each coordinate $i\in\{1,\ldots,h\}$, the Boolean variables $X_{v,i}$ and $Y_{v,i}$ encode the row and column choices of $v$. For each vertex $v\in V$ and each grid position $(r,c)\in H_h$, the Boolean variable $A_{v,r,c}$ indicates that $v$ occupies this position. For each edge $e=(u,v)$, each axis $q\in\{x,y\}$, and each represented threshold $d$, the Boolean variable $T^q_{e,d}$ indicates that the coordinate separation of $u$ and $v$ along axis $q$ is at least $d$. The following subsections present these components, their alternatives where applicable, and the conditions required for the optional restrictions to preserve correctness.

\subsection{Position and Occupancy Encoding}
The position component $\mathcal F_{\mathrm{pos}}$ corresponds to the
one-position-per-vertex condition in Constraint~\eqref{eq:assign_vertex} of
the ILP model. In the SAT encoding, this component requires each vertex to
select exactly one row and exactly one column. Equations~\eqref{eq:sat_x_exact}
and~\eqref{eq:sat_y_exact} state the corresponding cardinality constraints for
the row and column variables, respectively. These two constraints are encoded
into CNF using the Sequential Counter encoding of~\citet{Sinz2005}, as
implemented in PySAT~\citep{Ignatiev2018PySAT,Ignatiev2024AccessibleSAT}.
\begin{align}
\sum_{i=1}^{h}X_{v,i}&=1 && \forall v\in V
\label{eq:sat_x_exact}\\
\sum_{i=1}^{h}Y_{v,i}&=1 && \forall v\in V
\label{eq:sat_y_exact}
\end{align}

The occupancy component corresponds to Constraint~\eqref{eq:assign_cell} in the ILP model, which requires each grid position to contain at most one vertex. A direct way to enforce this condition is to consider every grid position $(r,c)\in H_h$ and every pair of distinct vertices $u,v\in V$, and forbid both vertices from occupying that position simultaneously. This construction applies the standard pairwise encoding of an at-most-one constraint to each grid position~\citep{Prestwich2009CNF}. Equation~\eqref{eq:pairwise_occupancy} defines the resulting pairwise occupancy encoding $\mathcal F_{\mathrm{occ}}^{\mathrm{pair}}$. Each clause in Equation~\eqref{eq:pairwise_occupancy} is violated exactly when both $u$ and $v$ occupy $(r,c)$, so the conjunction ensures that every grid position contains at most one vertex.
\begin{equation}
\mathcal F_{\mathrm{occ}}^{\mathrm{pair}}
=
\bigwedge_{(r,c)\in H_h}
\bigwedge_{\substack{u,v\in V\\u<v}}
\left(
\neg X_{u,r}\lor\neg Y_{u,c}\lor
\neg X_{v,r}\lor\neg Y_{v,c}
\right)
\label{eq:pairwise_occupancy}
\end{equation}

For a graph with $n$ vertices, the encoding introduces one clause for every unordered pair of distinct vertices at each grid position, resulting in $\binom{n}{2}$ clauses per position. Since the host grid has $\lceil\sqrt n\rceil \times \lceil\sqrt n\rceil$ positions, the complete pairwise occupancy encoding contains $\lceil\sqrt n\rceil \times \lceil\sqrt n\rceil \times \binom{n}{2}$ clauses. Thus, the pairwise occupancy encoding contains $\Theta(n^3)$ clauses in the square-root-grid setting. Instead of constructing the occupancy condition from pairwise clauses, the proposed method formulates an at-most-one cardinality constraint for each grid position. Specifically, for each vertex $v\in V$ and grid position $(r,c)\in H_h$, the method introduces a Boolean variable $A_{v,r,c}$ indicating whether $v$ occupies $(r,c)$. To enforce this meaning, Equations~\eqref{eq:cell_to_x}--\eqref{eq:xy_to_cell} link $A_{v,r,c}$ to the row and column variables.
\begin{align}
\neg A_{v,r,c}\lor X_{v,r} &
&& \forall v\in V,\ (r,c)\in H_h
\label{eq:cell_to_x}\\
\neg A_{v,r,c}\lor Y_{v,c} &
&& \forall v\in V,\ (r,c)\in H_h
\label{eq:cell_to_y}\\
A_{v,r,c}\lor\neg X_{v,r}\lor\neg Y_{v,c} &
&& \forall v\in V,\ (r,c)\in H_h
\label{eq:xy_to_cell}
\end{align}

For each grid position $(r,c)\in H_h$, the requirement that no more than one vertex occupies that position is expressed by applying an at-most-one constraint to the variables $\{A_{v,r,c}:v\in V\}$, as stated in Equation~\eqref{eq:cell_amo}.
This constraint can be translated into CNF using different cardinality encodings without changing the represented occupancy condition. In this study, these constraints are encoded using the Sequential Counter encoding of~\citet{Sinz2005}, as implemented in PySAT~\citep{Ignatiev2018PySAT,Ignatiev2024AccessibleSAT}. Together, the linking clauses in Equations~\eqref{eq:cell_to_x}--\eqref{eq:xy_to_cell} and the Sequential Counter encoding of Equation~\eqref{eq:cell_amo} define the sequential occupancy encoding $\mathcal F_{\mathrm{occ}}^{\mathrm{seq}}$. For the $n$ occupancy variables associated with each grid position, this encoding introduces $\Theta(n)$ clauses and $\Theta(n)$ auxiliary variables. Across the $\lceil\sqrt n\rceil \times \lceil\sqrt n\rceil$ grid positions, the resulting counts are therefore $\Theta(n^2)$ clauses and $\Theta(n^2)$ auxiliary variables. After the variables $A_{v,r,c}$ and their linking clauses are included, the complete sequential occupancy encoding remains $\Theta(n^2)$ in both clauses and auxiliary variables, whereas the pairwise occupancy encoding contains $\Theta(n^3)$ clauses. Despite their different CNF representations, $\mathcal F_{\mathrm{occ}}^{\mathrm{seq}}$ and $\mathcal F_{\mathrm{occ}}^{\mathrm{pair}}$ enforce the same occupancy condition over the placement variables.
\begin{equation}
\sum_{v\in V}A_{v,r,c}\le 1
\qquad \forall (r,c)\in H_h
\label{eq:cell_amo}
\end{equation}

\subsection{Distance Encoding}

The distance component corresponds to Constraints~\eqref{eq:dx_split}--
\eqref{eq:beta_bound} in the ILP model, which linearize the absolute-value
condition for the Manhattan distance of each guest edge. In the full SAT
encoding, these differences are represented by Boolean threshold variables.
For each edge $e=(u,v)\in E$, coordinate axis $q\in\{x,y\}$, where $q=x$
denotes the row axis and $q=y$ denotes the column axis, and threshold
$d\in\{1,\ldots,h-1\}$, the variable $T^q_{e,d}$ is true exactly when the
coordinate separation of $u$ and $v$ along axis $q$ is at least $d$. To express
the corresponding clauses uniformly, the method uses $Z^q_{w,i}$ as shorthand
for $X_{w,i}$ when $q=x$ and $Y_{w,i}$ when $q=y$, where $w\in V$ and
$i\in\{1,\ldots,h\}$.

Equations~\eqref{eq:dist_act_1}--\eqref{eq:dist_mono} enforce the threshold semantics of $T^q_{e,d}$. When the selected coordinates are exactly $d$ positions apart, Equations~\eqref{eq:dist_act_1} and~\eqref{eq:dist_act_2} force $T^q_{e,d}$ to be true for the two possible coordinate orders. Equations~\eqref{eq:dist_deact_1} and~\eqref{eq:dist_deact_2} encode the upper boundary of an exact separation. If the selected coordinates are exactly $d$ positions apart, the separation cannot be at least $d+1$, so these equations force $T^q_{e,d+1}$ to be false for the two possible coordinate orders. Equation~\eqref{eq:dist_zero} handles zero separation along one axis. If $u$ and $v$ select the same row for $q=x$, or the same column for $q=y$, then $T^q_{e,1}$ is forced to be false. Equation~\eqref{eq:dist_mono} enforces consistency among the threshold variables. If the coordinate separation is at least $d+1$, it must also be at least $d$, so $T^q_{e,d+1}$ implies $T^q_{e,d}$. Together with the unique coordinate assignments enforced by $\mathcal F_{\mathrm{pos}}$, these clauses establish the intended threshold semantics. Specifically, $T^q_{e,d}$ is true if and only if the coordinate separation of $u$ and $v$ is at least $d$. The conjunction of these clauses over all $e=(u,v)\in E$, $q\in\{x,y\}$, and the stated ranges of $d$ and $k$ is denoted by $\mathcal F_{\mathrm{dist}}$.
\begin{align}
Z^q_{v,k}\land Z^q_{u,k-d}
&\rightarrow T^q_{e,d}
&& 1\le d\le h-1,\ d+1\le k\le h
\label{eq:dist_act_1}\\
Z^q_{u,k}\land Z^q_{v,k-d}
&\rightarrow T^q_{e,d}
&& 1\le d\le h-1,\ d+1\le k\le h
\label{eq:dist_act_2}\\
Z^q_{v,k}\land Z^q_{u,k-d}
&\rightarrow \neg T^q_{e,d+1}
&& 1\le d\le h-2,\ d+1\le k\le h
\label{eq:dist_deact_1}\\
Z^q_{u,k}\land Z^q_{v,k-d}
&\rightarrow \neg T^q_{e,d+1}
&& 1\le d\le h-2,\ d+1\le k\le h
\label{eq:dist_deact_2}\\
Z^q_{u,k}\land Z^q_{v,k}
&\rightarrow \neg T^q_{e,1}
&& 1\le k\le h
\label{eq:dist_zero}\\
T^q_{e,d+1}
&\rightarrow T^q_{e,d}
&& 1\le d\le h-2
\label{eq:dist_mono}
\end{align}

The cutoff-aware encoding $\mathcal F_{\mathrm{dist}}^{U^*}$ uses a valid initial bandwidth upper bound $U^*$ to limit the represented axis-distance thresholds, and the descending search starts from the same value. Lin and Lin showed that $\delta(n)$ is the minimum possible maximum Manhattan distance among $n$ distinct lattice points with integer coordinates~\citep{LinLin2010}. However, their result does not guarantee that a configuration attaining $\delta(n)$ fits within the compact $\lceil\sqrt n\rceil\times\lceil\sqrt n\rceil$ host used in this study. Let $U_{\mathrm{DRSA}}$ denote the bandwidth of a feasible embedding on this compact host returned by DRSA~\citep{TorresJimenez2025}. Since $U^*=\max\{\delta(n),U_{\mathrm{DRSA}}\}\ge U_{\mathrm{DRSA}}$, the known DRSA embedding is feasible at the initial threshold $U^*$. When $\delta(n)>U_{\mathrm{DRSA}}$, the first SAT call uses $K=\delta(n)$. This initialization is intended to avoid starting the incremental search directly at a tighter bandwidth threshold. Additionally, if $U_{\mathrm{DRSA}}$ is unavailable, the bandwidth of another feasible embedding or, for the $h\times h$ host considered here, the host-diameter upper bound $2(h-1)$ can be used instead. Equations~\eqref{eq:delta_ub} and~\eqref{eq:ub_start} formally define $\delta(n)$ and $U^*$, respectively.
\begin{align}
\delta(n)
&=\min\left\{
2\left\lceil\frac{\sqrt{2n-1}-1}{2}\right\rceil,\,
2\left\lceil\sqrt{\frac n2}\right\rceil-1
\right\}
\label{eq:delta_ub}\\
U^*
&=\max\{\delta(n),U_{\mathrm{DRSA}}\}
\label{eq:ub_start}
\end{align}

Using this bound, $\mathcal F_{\mathrm{dist}}^{U^*}$ restricts the threshold index range to $d\in\{1,\ldots,\min\{U^*,h-1\}\}$ and adds the clauses in Equation~\eqref{eq:cutoff_pair} to exclude axis separations greater than $U^*$. By contrast, $\mathcal F_{\mathrm{dist}}$ represents the full threshold range $d\in\{1,\ldots,h-1\}$ and does not include the cutoff clauses. Since the separation along each axis is at most $h-1$, the Manhattan distance of an edge is at most $2(h-1)$. The descending search with $\mathcal F_{\mathrm{dist}}$ therefore starts from the initial bandwidth threshold $K=2(h-1)$. Every embedding with bandwidth at most $U^*$ also has axis separations at most $U^*$, so the cutoff preserves every embedding within the searched bandwidth range and does not affect the exactness of the encoding.

\begin{align}
\neg Z^q_{u,i}\lor\neg Z^q_{v,k}
&\qquad
\substack{e=(u,v)\in E,\ q\in\{x,y\},\\
i,k\in\{1,\ldots,h\},\ |i-k|>U^*}
\label{eq:cutoff_pair}
\end{align}

\subsection{Q1 Anchor Symmetry Breaking}

Reflecting an embedding horizontally or vertically changes the grid positions of its vertices but preserves all Manhattan edge distances. Following the quarter-grid symmetry-breaking strategy of~\citet{RodriguezGarcia2021}, the SAT formulation restricts a deterministically selected anchor vertex $a$ to one quarter of the host grid. The anchor is chosen by maximum degree, then by the maximum sum of the degrees of its neighbors, and finally by the smallest vertex index if a tie remains.

Define $Q_1=\{(r,c)\in H_h:1\le r,c\le \lceil h/2\rceil\}$ as the region allowed for the anchor vertex $a$. Figure~\ref{fig:q1-region} illustrates this definition for a $4\times4$ host grid, where $Q_1$ is the shaded $2\times2$ region. To restrict $a$ to $Q_1$, the SAT formulation forbids every row and column coordinate greater than $\lceil h/2\rceil$. Using the position variables $X_{a,p}$ and $Y_{a,p}$, Equations~\eqref{eq:q1_x} and~\eqref{eq:q1_y} encode these row and column exclusions, respectively, as unit clauses. Together, these unit clauses force both coordinates of $a$ to lie in $\{1,\ldots,\lceil h/2\rceil\}$, and therefore restrict $a$ to $Q_1$. The conjunction of the unit clauses in Equations~\eqref{eq:q1_x} and~\eqref{eq:q1_y} is denoted by $\mathcal F_{\mathrm{sym}}$.
\begin{align}
\neg X_{a,p} && \forall p\in
\left\{\left\lceil\frac{h}{2}\right\rceil+1,\ldots,h\right\}
\label{eq:q1_x}\\
\neg Y_{a,p} && \forall p\in
\left\{\left\lceil\frac{h}{2}\right\rceil+1,\ldots,h\right\}
\label{eq:q1_y}
\end{align}

\begin{figure}[pos=ht]
\centering
\small
\tikzset{
  qvertex/.style={circle, draw=black, fill=blue!12, minimum size=5.2mm, inner sep=0pt, font=\scriptsize},
  qgridline/.style={gray!35, line width=0.3pt},
  qgridborder/.style={black!50, line width=0.55pt},
  qaxis/.style={gray!70, dashed, line width=0.75pt},
  qregion/.style={draw=blue!55!black, line width=0.85pt}
}
\begin{tikzpicture}[scale=0.78]
\fill[blue!8] (0.5,-0.5) rectangle (2.5,-2.5);
\draw[qgridline] (0.5,-0.5) grid (4.5,-4.5);
\draw[qgridborder] (0.5,-0.5) rectangle (4.5,-4.5);
\draw[qaxis] (2.5,-0.5) -- (2.5,-4.5);
\draw[qaxis] (0.5,-2.5) -- (4.5,-2.5);
\draw[qregion] (0.5,-0.5) rectangle (2.5,-2.5);
\foreach \r in {1,...,4}{
  \foreach \c in {1,...,4}{
    \node[font=\tiny, gray!55] at (\c,-\r+0.32) {$(\r,\c)$};
  }
}
\node[qvertex] (anchor) at (1,-1) {$a$};
\node[font=\bfseries\scriptsize, blue!55!black, fill=blue!8, inner sep=1pt]
  at (1.5,-2.15) {Q1};
\node[font=\scriptsize] at (2.5,-5.05) {$4\times4$ host grid};
\end{tikzpicture}
\caption{Q1 region used for anchor symmetry breaking on a $4\times4$ host grid. The dashed reflection axes divide the grid into four $2\times2$ regions. The anchor vertex $a$ is restricted to the shaded Q1 region.}
\label{fig:q1-region}
\end{figure}
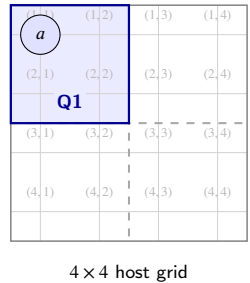

For any feasible embedding, horizontal and vertical reflections can be applied independently so that the anchor vertex lies in $Q_1$. Each reflection is a bijection of $H_h$ and preserves Manhattan distances, so the transformed embedding remains injective and has the same bandwidth. Consequently, adding $\mathcal F_{\mathrm{sym}}$ does not change whether a bandwidth-$K$ embedding exists and does not change the optimal value. Lemma~\ref{lem:q1-symmetry} in Appendix~\ref{app:correctness} provides the formal preservation proof.

\subsection{Incremental Solving and Certification}

The bandwidth component $\Delta_K$ is the only component that depends on $K$. It enforces the bound $K$ through clauses constructed from the threshold variables of the selected distance encoding. Equations~\eqref{eq:delta_k_cutoff} and~\eqref{eq:delta_k_full} define $\Delta_K$ for $\mathcal F_{\mathrm{dist}}^{U^*}$ and $\mathcal F_{\mathrm{dist}}$, respectively. Both forms contain the same two clause groups. Their only difference is that the largest encoded threshold is $\min\{U^*,h-1\}$ in the cutoff-aware form and $h-1$ in the full-distance form. The first group contains a unit clause for each edge and axis when the threshold at $K+1$ does not exceed the applicable limit. This clause restricts the separation along that axis to at most $K$. The second group considers every decomposition $K+1=i+(K-i+1)$ whose two terms do not exceed the applicable limit and forbids the two axis separations from reaching these thresholds simultaneously. Because the two thresholds sum to $K+1$, every forbidden combination would give the edge a Manhattan distance greater than $K$.
\begin{subequations}
\begin{align}
\Delta_K
&=
\bigwedge_{e\in E}
\left(
\bigwedge_{\substack{q\in\{x,y\}\\K+1\le \min\{U^*,h-1\}}}
\neg T^q_{e,K+1}
\land
\bigwedge_{\substack{1\le i\le K\\i\le \min\{U^*,h-1\},\ K-i+1\le \min\{U^*,h-1\}}}
\left(\neg T^x_{e,i}\lor\neg T^y_{e,K-i+1}\right)
\right)
\label{eq:delta_k_cutoff}\\
\Delta_K
&=
\bigwedge_{e\in E}
\left(
\bigwedge_{\substack{q\in\{x,y\}\\K+1\le h-1}}
\neg T^q_{e,K+1}
\land
\bigwedge_{\substack{1\le i\le K\\i\le h-1,\ K-i+1\le h-1}}
\left(\neg T^x_{e,i}\lor\neg T^y_{e,K-i+1}\right)
\right)
\label{eq:delta_k_full}
\end{align}
\end{subequations}

Table~\ref{tab:encoding-configurations} summarizes the four complete encoding configurations considered in this study. The proposed encoding is denoted by $\mathcal F_{\mathrm{full}}(K)$. The remaining notation identifies the component changed relative to the proposed encoding. The notation $\mathcal F_{\mathrm{no\text{-}sym}}(K)$ removes Q1 anchor symmetry breaking. The notation $\mathcal F_{\mathrm{pair}}(K)$ uses pairwise occupancy, while $\mathcal F_{\mathrm{full\text{-}dist}}(K)$ uses the full distance encoding without the cutoff. All four configurations encode the same bandwidth-$K$ decision problem. In $\mathcal F_{\mathrm{full\text{-}dist}}(K)$, $\Delta_K$ is defined by Equation~\eqref{eq:delta_k_full}, and the descending search starts from $K=2(h-1)$, the maximum possible Manhattan distance between two positions on an $h\times h$ grid. In configurations containing $\mathcal F_{\mathrm{dist}}^{U^*}$, $\Delta_K$ is defined by Equation~\eqref{eq:delta_k_cutoff}, and the search starts from $K=U^*$, the initial bandwidth upper bound used to construct $\mathcal F_{\mathrm{dist}}^{U^*}$.

\begin{table}[pos=H]
\centering
\caption{Complete SAT encoding configurations considered in this study.}
\label{tab:encoding-configurations}
{\renewcommand{\arraystretch}{1.5}
\setlength{\tabcolsep}{0pt}
\begin{tabular}{p{0.15\textwidth}p{0.3\textwidth}p{0.55\textwidth}}
\toprule
Notation & Definition & Description \\
\midrule
$\mathcal F_{\mathrm{full}}(K)$ &
$\mathcal F_{\mathrm{pos}} \land \mathcal F_{\mathrm{occ}}^{\mathrm{seq}} \land \mathcal F_{\mathrm{dist}}^{U^*} \land \mathcal F_{\mathrm{sym}} \land \Delta_K$ &
Proposed encoding \\

$\mathcal F_{\mathrm{no\text{-}sym}}(K)$ &
$\mathcal F_{\mathrm{pos}} \land \mathcal F_{\mathrm{occ}}^{\mathrm{seq}} \land \mathcal F_{\mathrm{dist}}^{U^*} \land \Delta_K$ &
Proposed encoding without Q1 anchor symmetry breaking \\

$\mathcal F_{\mathrm{pair}}(K)$ &
$\mathcal F_{\mathrm{pos}} \land \mathcal F_{\mathrm{occ}}^{\mathrm{pair}} \land \mathcal F_{\mathrm{dist}}^{U^*} \land \mathcal F_{\mathrm{sym}} \land \Delta_K$ &
Proposed encoding with pairwise occupancy in place of sequential occupancy \\

$\mathcal F_{\mathrm{full\text{-}dist}}(K)$ &
$\mathcal F_{\mathrm{pos}} \land \mathcal F_{\mathrm{occ}}^{\mathrm{seq}} \land \mathcal F_{\mathrm{dist}} \land \mathcal F_{\mathrm{sym}} \land \Delta_K$ &
Proposed encoding without the cutoff-aware distance restriction \\
\bottomrule
\end{tabular}
}
\end{table}

The incremental procedure applies to any encoding configuration in Table~\ref{tab:encoding-configurations} because $\Delta_K$ is the only component that changes with $K$. For the selected configuration, $\mathit{fixedClauses}$ denotes the clauses generated from all remaining components, which are added to the SAT solver once before the threshold search begins. For example, when $\mathcal F_{\mathrm{full}}(K)$ is selected, $\mathit{fixedClauses}$ contains the clauses generated from $\mathcal F_{\mathrm{pos}}$, $\mathcal F_{\mathrm{occ}}^{\mathrm{seq}}$, $\mathcal F_{\mathrm{dist}}^{U^*}$, and $\mathcal F_{\mathrm{sym}}$. At each tested value of $K$, $\mathit{thresholdClauses}$ denotes the clauses generated from $\Delta_K$, which are added before the corresponding solver call. Table~\ref{tab:incremental-procedure} summarizes the complete procedure, including the update of $K$ from the actual bandwidth computed from a satisfiable model and termination with status \textsc{optimal}, \textsc{feasible}, or \textsc{timeout}. 
Here, \textsc{optimal} indicates that the returned $\mathit{bestBandwidth}$ is a certified optimal value. Status \textsc{feasible} indicates that the search timed out after finding a feasible upper bound, whereas \textsc{timeout} indicates that it timed out without finding a feasible embedding.

When the solver returns \textsc{sat} at threshold $K$, the decoded model yields a feasible embedding $\varphi$ whose actual bandwidth $\beta(G,\varphi)$ is at most $K$. Because this value may be smaller than $K$, the procedure sets the next threshold to $\mathit{actualBandwidth}-1$, thereby skipping intermediate thresholds that cannot improve the current feasible upper bound. If the formula at this next threshold is \textsc{unsat}, no embedding with bandwidth below $\mathit{actualBandwidth}$ exists. Together with the feasible embedding already found, this establishes $\mathit{actualBandwidth}$ as a certified optimal value. If $\mathit{actualBandwidth}=1$, it already attains the trivial lower bound for any graph containing at least one edge and can therefore be certified without another solver call. A timeout establishes neither feasibility nor infeasibility at the current threshold, so the procedure terminates without an optimality certificate.

As $K$ decreases, each newly added $\Delta_K$ strengthens the current formula, while the previously added threshold clauses remain in the same solver. Reusing the solver retains both these clauses and the learned information obtained from earlier calls. The non-incremental ablation instead creates a fresh solver, adds $\mathit{fixedClauses}$ and the current $\mathit{thresholdClauses}$ for each tested threshold, and discards that solver after the call. Thus, both procedures solve the same bandwidth-$K$ decision problems but differ in whether clauses and learned information from earlier calls are retained. Appendix~\ref{app:correctness} presents the correctness proofs for the encoding configurations in Table~\ref{tab:encoding-configurations} and proves that a value returned with status \textsc{optimal} by the incremental procedure is a certified optimal value.

\begin{table}[pos=H]
\centering
\caption{Incremental threshold search and certification.}
\label{tab:incremental-procedure}
{\renewcommand{\arraystretch}{1.15}
\setlength{\tabcolsep}{4pt}
\begin{tabular}{p{0.08\textwidth}p{0.88\textwidth}}
\toprule
\textbf{Input} & A graph $G$, an encoding configuration selected from Table~\ref{tab:encoding-configurations}, a time limit, and an initial bandwidth upper bound $U$. \\
\midrule
\textit{Step 1} & For $\mathcal F_{\mathrm{full}}(K)$, set $\mathit{fixedClauses}$ to the clauses generated from $\mathcal F_{\mathrm{pos}}$, $\mathcal F_{\mathrm{occ}}^{\mathrm{seq}}$, $\mathcal F_{\mathrm{dist}}^{U^*}$, and $\mathcal F_{\mathrm{sym}}$. The other configurations in Table~\ref{tab:encoding-configurations} are handled analogously, with $\Delta_K$ excluded from $\mathit{fixedClauses}$ in every case. \\
\midrule
\textit{Step 2} & Create a SAT solver and add $\mathit{fixedClauses}$ to it. Set $K\leftarrow U$ and $\mathit{bestBandwidth}\leftarrow\textsc{undefined}$. \\
\midrule
\textit{Step 3} & \textbf{While} $K\ge 1$ \textbf{do} \newline
\hspace*{1.5em}Set $\mathit{thresholdClauses}$ to the clauses generated from $\Delta_K$. \newline
\hspace*{1.5em}Add $\mathit{thresholdClauses}$ to the same solver. \newline
\hspace*{1.5em}Set $\mathit{result}\leftarrow\operatorname{Solve}()$. \newline
\hspace*{1.5em}\textbf{If} the time limit is reached and $\mathit{bestBandwidth}$ is defined, \textbf{return} $\mathit{bestBandwidth}$ with status \textsc{feasible}. \newline
\hspace*{1.5em}\textbf{If} the time limit is reached and $\mathit{bestBandwidth}$ is undefined, \textbf{return} status \textsc{timeout}. \newline
\hspace*{1.5em}\textbf{If} $\mathit{result}=\textsc{unsat}$, \textbf{return} $\mathit{bestBandwidth}$ with status \textsc{optimal}. \newline
\hspace*{1.5em}Decode the SAT model into an embedding $\varphi$. \newline
\hspace*{1.5em}Set $\mathit{actualBandwidth}\leftarrow\beta(G,\varphi)$. \newline
\hspace*{1.5em}Set $\mathit{bestBandwidth}\leftarrow\mathit{actualBandwidth}$. \newline
\hspace*{1.5em}\textbf{If} $\mathit{bestBandwidth}=1$, \textbf{return} $\mathit{bestBandwidth}$ with status \textsc{optimal}. \newline
\hspace*{1.5em}Set $K\leftarrow\mathit{bestBandwidth}-1$. \newline
\textbf{End while}. \\
\midrule
\textbf{Output} & The best bandwidth found, if available, together with status \textsc{optimal}, \textsc{feasible}, or \textsc{timeout}. \\
\bottomrule
\end{tabular}
}
\end{table}

\section{Computational Experiments}
\label{sec:experiments}


\subsection{Experimental Setup}
\begin{table}[pos=H]
\caption{Summary of the benchmark sets evaluated in this study.}
\label{tab:benchmark-sets}
\centering
\footnotesize
\begin{tabular}{lcccc}
\toprule
Benchmark set & Instances & $n$ range & $|E|$ range & Host grid \\
\midrule
Regular & 43 & 5--21 & 6--190 & $\lceil\sqrt n\rceil\times\lceil\sqrt n\rceil$ \\
Harwell--Boeing & 93 & 9--362 & 50--5856 & $\lceil\sqrt n\rceil\times\lceil\sqrt n\rceil$ \\
Regular, 2-row & 43 & 5--21 & 6--190 & $2\times\lceil n/2\rceil$ \\
Harwell--Boeing, 2-row & 93 & 9--362 & 50--5856 & $2\times\lceil n/2\rceil$ \\
Regular, extended & 43 & 5--21 & 6--190 & $n\times n$ \\
Harwell--Boeing, extended & 49 & 9--147 & 50--2211 & $n\times n$ \\
\bottomrule
\end{tabular}
\end{table}

Table~\ref{tab:benchmark-sets} summarizes the benchmark sets evaluated in this study. The Regular set contains 43 guest graphs from the benchmark collection accompanying RLTS~\citep{Zhou2DBMPRepository}, which has also been used in earlier 2DBMP studies~\citep{RodriguezGarcia2021,Cavero2023,Zhou2026}. The Harwell--Boeing set contains 93 guest graphs derived from sparse-matrix patterns with $n\le 400$ from the Matrix Market repository~\citep{MatrixMarket2014,Zhou2DBMPRepository}. This size limit was imposed to keep the exact experiments within the available computational resources. The same 93 instances were evaluated by DRSA~\citep{TorresJimenez2025}, while overlapping subsets were considered in other 2DBMP studies~\citep{RodriguezGarcia2021,Cavero2023,Khandelwal2023RVNS,Zhou2026}. The remaining columns report the ranges of the numbers of vertices $n$ and edges $|E|$, together with the host-grid dimensions. The standard experiments use the $\lceil\sqrt n\rceil\times\lceil\sqrt n\rceil$ host grid. The additional host-grid experiments include a $2\times\lceil n/2\rceil$ host-grid evaluation on all 43 Regular and 93 Harwell--Boeing instances, and an $n\times n$ host-grid evaluation on all 43 Regular and 49 instances from the same Harwell--Boeing set.

All experiments were conducted on a workstation equipped with an Intel Core i7-12700KF processor and 32~GB of RAM, running Ubuntu 24.04.3 LTS. The SAT implementation invoked CaDiCaL 1.9.5 through the PySAT toolkit~\citep{Ignatiev2018PySAT,Ignatiev2024AccessibleSAT,Biere2024Cadical}. The commercial-solver comparisons used Gurobi Optimizer 12.0.3~\citep{Gurobi1203} and two formulations in IBM ILOG CPLEX Optimization Studio 22.1.1~\citep{IBMCPLEX2211}: a mixed-integer programming formulation (CPLEX MIP) and a constraint-programming formulation (CPLEX CP). A time limit of 3600~s was applied to each local run. For all locally evaluated methods, the reported solving time includes the complete run for each instance, including model or encoding construction and all solver calls.

The initial SAT thresholds and the domains of the commercial bandwidth variable $\beta$ are specified separately for each host grid. On the standard $h\times h$ grid, where $h=\lceil\sqrt n\rceil$, the cutoff-aware SAT configurations use $U^*$ from Equation~\eqref{eq:ub_start} as both the distance cutoff and the initial threshold. By contrast, $\mathcal F_{\mathrm{full\text{-}dist}}(K)$ starts from $K=2(h-1)$, and the commercial solvers (Gurobi, CPLEX MIP, and CPLEX CP) minimize $\beta$ over $0\le\beta\le2(h-1)$. On the $n\times n$ grid, the cutoff-aware SAT configuration uses $\delta(n)$ from Equation~\eqref{eq:delta_ub} as both the distance cutoff and the initial threshold. The commercial solvers minimize $\beta$ over $0\le\beta\le2(n-1)$ on this host.
On the $2\times\lceil n/2\rceil$ host, $\delta(n)$ is not guaranteed to provide a feasible bandwidth threshold and is therefore used only as the starting value of the initialization procedure. Starting from $\delta(n)$, successive candidate thresholds are generated by increasing the multiplier applied to $\delta(n)$ in steps of $0.5$ (yielding $\delta(n),\lceil1.5\delta(n)\rceil,2\delta(n),\lceil2.5\delta(n)\rceil,\ldots$) until the resulting threshold reaches or exceeds the lower bound in Equation~\eqref{eq:2xm-lower-bound}. SAT feasibility checks then determine a valid initial upper bound. Appendix~\ref{app:initial-bound-2xm} provides the complete initialization procedure. The resulting upper bound is used as both the distance cutoff and the initial SAT threshold. The commercial solvers minimize $\beta$ over $0\le\beta\le\lceil n/2\rceil$ on this host.

The proposed SAT method, Gurobi, CPLEX MIP, and CPLEX CP were run locally, with each method run once on each instance.
The SAT experiments were performed as single runs because preliminary repeated runs on representative instances showed no material variation. Similarly, Gurobi and CPLEX MIP/CP were each run once per instance under their respective default configurations. The two CPLEX formulations are kept separate because they use different optimization models. The third constraint satisfaction programming model (M3) of \citet{RodriguezGarcia2021} is included as a published exact reference.
Published results from BVNS~\citep{RodriguezGarcia2021}, IG~\citep{Cavero2023}, RVNS~\citep{Khandelwal2023RVNS}, DRSA~\citep{TorresJimenez2025}, and RLTS~\citep{Zhou2026} are included only for comparison of objective values and are not treated as evidence of optimality. 
For the SAT runs, \texttt{OPT} is assigned according to Proposition~\ref{prop:incremental-optimality}, when a feasible bandwidth $B$ is followed by an \textsc{unsat} result at $B-1$ or when a feasible bandwidth of $1$ is found. 
For Gurobi and CPLEX, \texttt{OPT} denotes the solver's native proven-optimal status. 
In the summary result tables, \texttt{OPT} indicates a certified optimal value, \texttt{FEAS} indicates termination at the time limit with a feasible incumbent but without an optimality proof, and \texttt{TO} indicates timeout without a feasible solution. In the detailed Appendix tables, a star marks a certified optimal bandwidth, \texttt{FEAS} in parentheses marks a feasible incumbent at the time limit, and \texttt{-- (TO)} marks timeout without a feasible solution. Subsections~\ref{sec:exp-ablation}--\ref{sec:exp-commercial} then report the standard-grid results on the $\lceil\sqrt n\rceil\times\lceil\sqrt n\rceil$ host. Specifically, Subsection~\ref{sec:exp-ablation} evaluates the SAT encoding and incremental solving components, Subsection~\ref{sec:exp-published} compares certified SAT values with published reference values, and Subsection~\ref{sec:exp-commercial} compares the proposed SAT method with commercial exact solvers. Subsection~\ref{sec:exp-additional-host-grids} reports the additional host-grid comparisons on the $2\times\lceil n/2\rceil$ and $n\times n$ grids.
\subsection{Results on the \texorpdfstring{$\lceil\sqrt n\rceil\times\lceil\sqrt n\rceil$}{ceil(sqrt(n)) x ceil(sqrt(n))} Grid}
\label{sec:exp-sqrt-grid}

\subsubsection{Ablation Study of SAT Encoding and Incremental Solving}
\label{sec:exp-ablation}
\begin{table}[pos=ht]
\centering
\small
\setlength{\tabcolsep}{2pt}
\caption{\texttt{OPT} and \texttt{FEAS} counts for the five SAT variants on the Regular and Harwell--Boeing sets using the $\lceil\sqrt n\rceil\times\lceil\sqrt n\rceil$ host grid.}
\label{tab:ablation-results-summary}
\begin{tabular}{l@{\hspace{1.5em}}rr@{}p{3em}@{}rr}
\toprule
& \multicolumn{2}{c}{Regular (43)} & & \multicolumn{2}{c}{Harwell--Boeing (93)} \\
\cmidrule(lr{1.2em}){2-3}\cmidrule(lr){5-6}
Variant & \texttt{OPT} & \texttt{FEAS} & & \texttt{OPT} & \texttt{FEAS} \\
\midrule
$\mathcal{F}_{\mathrm{full}}$ & 41 & 2 &  & \textbf{42} & \textbf{51} \\
$\mathcal{F}_{\mathrm{no\text{-}sym}}$ & 41 & 2 &  & 34 & 59 \\
$\mathcal{F}_{\mathrm{pair}}$ & 41 & 2 &  & 35 & 58 \\
$\mathcal{F}_{\mathrm{full\text{-}dist}}$ & 41 & 2 &  & 37 & 56 \\
Non-incremental $\mathcal{F}_{\mathrm{full}}$ & 41 & 2 &  & 37 & 56 \\
\bottomrule
\end{tabular}
\end{table}

This ablation compares the proposed $\mathcal F_{\mathrm{full}}(K)$ with $\mathcal F_{\mathrm{no\text{-}sym}}(K)$, $\mathcal F_{\mathrm{pair}}(K)$, and $\mathcal F_{\mathrm{full\text{-}dist}}(K)$, whose configurations are summarized in Table~\ref{tab:encoding-configurations}. It also includes a non-incremental variant of $\mathcal F_{\mathrm{full}}(K)$ to evaluate the effect of reusing the SAT solver across bandwidth thresholds. Table~\ref{tab:ablation-results-summary} reports the numbers of instances for which each variant proves optimality or reaches the time limit. On the Regular set, all five variants return the same bandwidth and status on every instance, so no difference in solution quality is observed. On the Harwell--Boeing set, $\mathcal F_{\mathrm{full}}$ has the highest certification coverage with 42 instances, compared with 34 for $\mathcal F_{\mathrm{no\text{-}sym}}$, 35 for $\mathcal F_{\mathrm{pair}}$, and 37 for both $\mathcal F_{\mathrm{full\text{-}dist}}$ and non-incremental $\mathcal F_{\mathrm{full}}$. Appendix~\ref{app:ablation-encoding-size-comparison} reports aggregate variable and clause counts before $\Delta_K$ is added.

\begin{figure}[pos=htbp]
\centering
\includegraphics[width=0.96\textwidth]{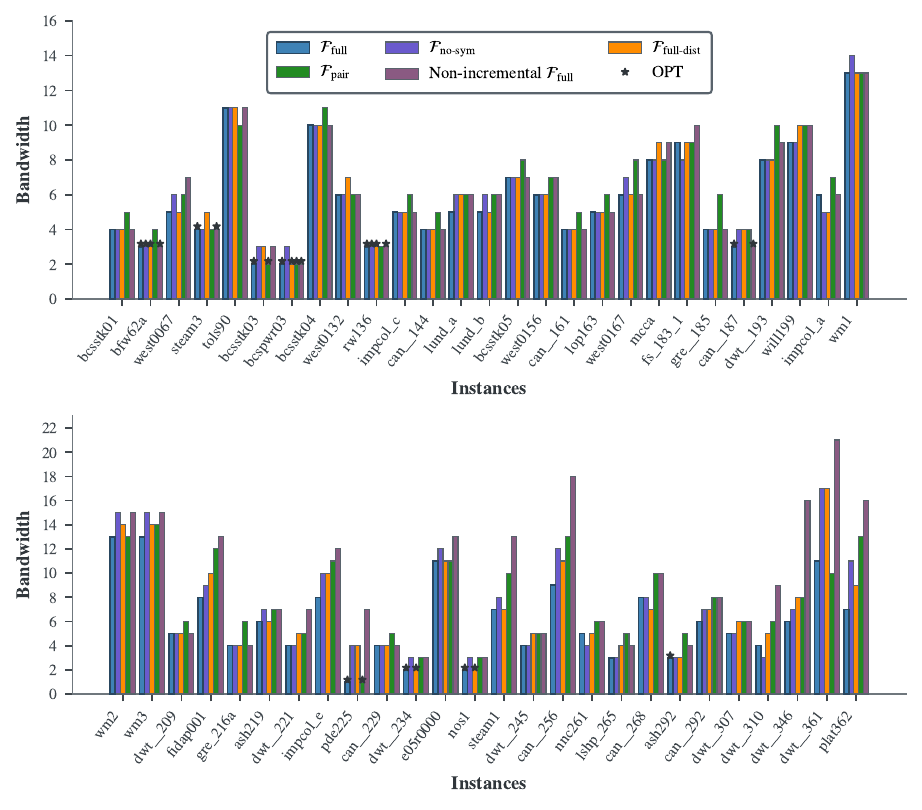}
\caption{Returned-bandwidth comparison of the five SAT variants on the 53 Harwell--Boeing instances with differing bandwidths or statuses.}
\label{fig:ablation-bandwidth-comparison}
\end{figure}

\begin{figure}[pos=htbp]
\centering
\includegraphics[width=0.96\textwidth]{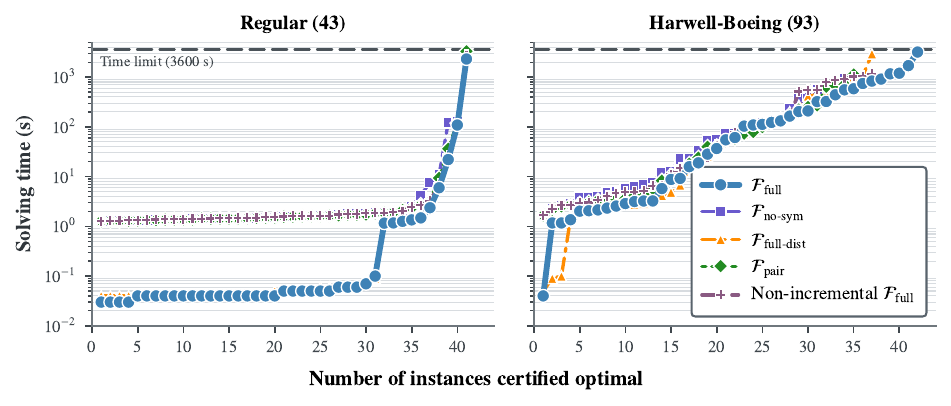}
\caption{Cumulative solving-time comparison of the five SAT variants on the standard $\lceil\sqrt{n}\rceil\times\lceil\sqrt{n}\rceil$ host grid.}
\label{fig:ablation-opt-cactus}
\end{figure}

Figures~\ref{fig:ablation-bandwidth-comparison} and~\ref{fig:ablation-opt-cactus} complement Table~\ref{tab:ablation-results-summary} with instance-level views of returned bandwidths and certification performance, respectively. Figure~\ref{fig:ablation-bandwidth-comparison} shows the 53 Harwell--Boeing instances for which at least two variants differ in returned bandwidth or status, whereas Figure~\ref{fig:ablation-opt-cactus} plots cumulative solving time over the instances certified by each variant. In this and subsequent bandwidth figures, a star marks a certified optimal bandwidth, an unstarred bar marks a feasible incumbent at the time limit, and \texttt{--} indicates that no feasible bandwidth was found. In Figure~\ref{fig:ablation-bandwidth-comparison}, $\mathcal F_{\mathrm{full}}$ attains the smallest returned bandwidth, including ties, on 46 of the 53 instances. On the remaining seven instances, at least one comparison variant returns a smaller bandwidth. In all seven cases, however, $\mathcal F_{\mathrm{full}}$ has status \texttt{FEAS}, so these differences concern incumbent quality rather than certified optimal values. Figure~\ref{fig:ablation-opt-cactus} shows that $\mathcal F_{\mathrm{full}}$ certifies additional Harwell--Boeing instances for which one or more comparison variants reaches the time limit. Over the certification-count range attained by all five variants, however, the cumulative solving-time profiles remain broadly similar. Taken together, these results indicate that the seven weaker \texttt{FEAS} incumbents constitute a limited trade-off alongside the higher Harwell--Boeing certification coverage and the generally consistent returned-bandwidth performance of $\mathcal F_{\mathrm{full}}$.



\subsubsection{Comparison with Published Results}
\label{sec:exp-published}

The comparison distinguishes the published exact constraint satisfaction model M3 of~\citet{RodriguezGarcia2021} from the heuristic methods BVNS~\citep{RodriguezGarcia2021}, IG~\citep{Cavero2023}, RVNS~\citep{Khandelwal2023RVNS}, DRSA~\citep{TorresJimenez2025}, and RLTS~\citep{Zhou2026}. The bandwidths published for the heuristic methods are treated as feasible upper bounds rather than evidence of optimality. The published methods were evaluated under different experimental settings. For example, the M3 experiments used a maximum CPU time of 72~h~\citep{RodriguezGarcia2021}, whereas DRSA reported, for each matrix, the best bandwidth obtained from five runs~\citep{TorresJimenez2025}. The comparison in this subsection therefore focuses on published bandwidths and available exact-certification status rather than solving time. Within this audited comparison, $\mathcal{F}_{\mathrm{full}}$ improves the best published bandwidth on three Harwell--Boeing instances: \texttt{bfw62a}, \texttt{bcsstk03}, and \texttt{bcsstk22}. Table~\ref{tab:published-square-grid-improvements} reports the corresponding improvements from 4 to 3, 3 to 2, and 3 to 2, all certified optimal by $\mathcal{F}_{\mathrm{full}}$.

{
\begin{table}[pos=ht]
\caption{Certified optimal bandwidths below the compared published values for three Harwell--Boeing instances on the $\lceil\sqrt n\rceil\times\lceil\sqrt n\rceil$ host grid.}
\label{tab:published-square-grid-improvements}
\centering
\small
\setlength{\tabcolsep}{3pt}
\begin{tabular}{lrrrrrrr@{\hspace{2.5em}}r}
\toprule
Instance & $n$ & M3 & BVNS & IG & RVNS & DRSA & RLTS & \textbf{$\mathcal{F}_{\mathrm{full}}$} \\
\midrule
\texttt{bfw62a} & 62 & -- & -- & -- & -- & 4 & -- & $\mathbf{3}^{*}$ \\
\texttt{bcsstk03} & 112 & -- & -- & -- & -- & 3 & -- & $\mathbf{2}^{*}$ \\
\texttt{bcsstk22} & 138 & 8 & 4 & 3 & 4 & 3 & 3 & $\mathbf{2}^{*}$ \\
\bottomrule
\end{tabular}
\end{table}
}

Table~\ref{tab:published-square-grid-summary} places the three improvements within the broader comparison with each published method. For each method, \emph{Published rows} gives the number of instances with an available published bandwidth, and the parenthetical counts for M3 identify how many of those bandwidths were marked optimal in the source study. The \emph{Lower}, \emph{Equal}, and \emph{Higher} columns classify the bandwidth returned by $\mathcal{F}_{\mathrm{full}}$ relative to the corresponding published bandwidth. When a method provides multiple published bandwidths for an instance, the comparison uses its best published bandwidth. The \texttt{OPT} block contains comparisons for which $\mathcal{F}_{\mathrm{full}}$ certifies its bandwidth as optimal, whereas the \texttt{FEAS} block contains comparisons based on feasible incumbents without an optimality proof.

{
\begin{table}[pos=ht]
\caption{Aggregate comparison of $\mathcal{F}_{\mathrm{full}}$ with published bandwidths on the $\lceil\sqrt n\rceil\times\lceil\sqrt n\rceil$ host grid, separated into \texttt{OPT} and \texttt{FEAS} groups.}
\label{tab:published-square-grid-summary}
\centering
\small
\setlength{\tabcolsep}{2pt}
\begin{tabular}{@{}l@{\hspace{1.2em}}r@{\hspace{2.0em}}rrr@{\hspace{2.5em}}rrr@{}}
\toprule
Published method & Published rows & \multicolumn{3}{c}{\textbf{$\mathcal{F}_{\mathrm{full}}$}: \texttt{OPT}} & \multicolumn{3}{c}{\textbf{$\mathcal{F}_{\mathrm{full}}$}: \texttt{FEAS}} \\
\cmidrule(lr{1.2em}){3-5}\cmidrule(l{1.2em}r){6-8}
& & Lower & Equal & Higher & Lower & Equal & Higher \\
\midrule
\multicolumn{8}{l}{\textit{Regular (43)}} \\
M3 & 39 (39 \texttt{OPT}) & 0 & 39 & 0 & 0 & 0 & 0 \\
\addlinespace[0.4em]
BVNS & 43 & 6 & 35 & 0 & 0 & 2 & 0 \\
IG & 43 & 0 & 41 & 0 & 0 & 2 & 0 \\
RLTS & 43 & 0 & 41 & 0 & 0 & 2 & 0 \\
\addlinespace[0.8em]
\multicolumn{8}{l}{\textit{Harwell--Boeing (93)}} \\
M3 & 24 (1 \texttt{OPT}) & 7 & 1 & 0 & 16 & 0 & 0 \\
\addlinespace[0.4em]
BVNS & 24 & 5 & 3 & 0 & 12 & 2 & 2 \\
IG & 53 & 10 & 13 & 0 & 0 & 18 & 12 \\
RVNS & 24 & 6 & 2 & 0 & 10 & 2 & 4 \\
DRSA & 93 & 4 & 38 & 0 & 0 & 14 & 37 \\
RLTS & 53 & 6 & 17 & 0 & 0 & 16 & 14 \\
\bottomrule
\end{tabular}
\end{table}
}

In comparison with M3, $\mathcal{F}_{\mathrm{full}}$ matches all 39 Regular bandwidths certified by the source study. On Harwell--Boeing, $\mathcal{F}_{\mathrm{full}}$ returns \texttt{OPT} for eight of the 24 instances with a published M3 bandwidth, including seven lower bandwidths and one bandwidth equal to the only M3-certified value. The remaining 16 comparisons are \texttt{FEAS}, and each incumbent returned by $\mathcal{F}_{\mathrm{full}}$ is lower than the corresponding M3 bandwidth but is not proved optimal. Thus, within this published M3 subset, $\mathcal{F}_{\mathrm{full}}$ extends exact certification from one to eight Harwell--Boeing instances. For the remaining 69 Harwell--Boeing instances without a published M3 bandwidth, $\mathcal{F}_{\mathrm{full}}$ returns a feasible bandwidth in every case and certifies 34 of these values as optimal.

In comparison with the heuristic methods, the principal strength of $\mathcal{F}_{\mathrm{full}}$ lies in exact certification rather than uniform dominance in incumbent quality. Every \texttt{OPT} bandwidth returned by $\mathcal{F}_{\mathrm{full}}$ is no higher than the corresponding published upper bound. A lower entry is a method-specific certified improvement, whereas an equal entry certifies the matching heuristic upper bound as optimal. The \texttt{FEAS} comparisons show a more mixed pattern. All six Regular comparisons are equal. On Harwell--Boeing, most \texttt{FEAS} incumbents are lower than or equal to the corresponding BVNS and RVNS values. Equal incumbents occur slightly more often than higher ones for IG and RLTS, while most incumbents are higher than the DRSA values. These results show that $\mathcal{F}_{\mathrm{full}}$ complements heuristic search by certifying matching or improved bandwidths, while its incomplete runs still provide competitive incumbents for several method-specific subsets. The higher \texttt{FEAS} values, particularly relative to DRSA, show that this advantage is not uniform under the 3600~s time limit. Detailed instance-level comparisons are reported in Table~\ref{tab:standard-grid-combined-detail} in Appendix~\ref{app:detailed-results}.


\subsubsection{Comparison with Commercial Solvers}
\label{sec:exp-commercial}

The proposed $\mathcal F_{\mathrm{full}}$ configuration is compared with Gurobi, CPLEX MIP, and CPLEX CP on the same instances and standard square-root host grid, using the same hardware and 3600~s time limit.
Table~\ref{tab:commercial-results-summary} summarizes the certification outcomes. On the Regular set, $\mathcal F_{\mathrm{full}}$, Gurobi, and CPLEX MIP certify 41 of 43 instances, while CPLEX CP certifies 42. On the Harwell--Boeing set, the corresponding counts are 42, 10, 8, and 20. The main distinction therefore appears on Harwell--Boeing, where $\mathcal F_{\mathrm{full}}$ provides broader exact-certification coverage.

\begin{table}[pos=ht]
\centering
\small
\setlength{\tabcolsep}{2pt}
\caption{\texttt{OPT} and \texttt{FEAS} counts for $\mathcal F_{\mathrm{full}}$, Gurobi, CPLEX MIP, and CPLEX CP on the Regular and Harwell--Boeing sets using the $\lceil\sqrt n\rceil\times\lceil\sqrt n\rceil$ host grid.}
\label{tab:commercial-results-summary}
\begin{tabular}{l@{\hspace{1.5em}}rr@{}p{3em}@{}rr}
\toprule
& \multicolumn{2}{c}{Regular (43)} & & \multicolumn{2}{c}{Harwell--Boeing (93)} \\
\cmidrule(lr{1.2em}){2-3}\cmidrule(lr){5-6}
Method & \texttt{OPT} & \texttt{FEAS} & & \texttt{OPT} & \texttt{FEAS} \\
\midrule
$\mathcal{F}_{\mathrm{full}}$ & 41 & 2 &  & \textbf{42} & \textbf{51} \\
Gurobi & 41 & 2 &  & 10 & 83 \\
CPLEX MIP & 41 & 2 &  & 8 & 85 \\
CPLEX CP & \textbf{42} & \textbf{1} &  & 20 & 73 \\
\bottomrule
\end{tabular}
\end{table}

\begin{figure}[pos=H]
\centering
\includegraphics[width=0.96\textwidth]{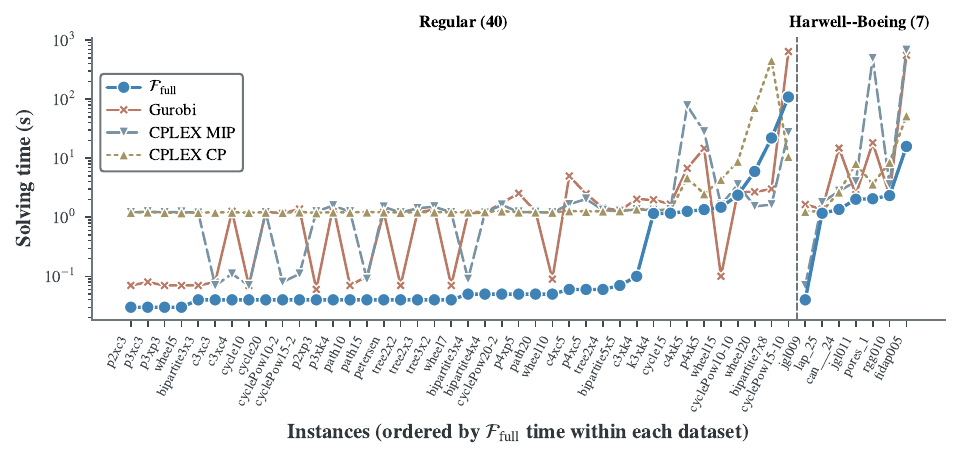}
\caption{Solving-time comparison of $\mathcal F_{\mathrm{full}}$, Gurobi, CPLEX MIP, and CPLEX CP on the $\lceil\sqrt n\rceil\times\lceil\sqrt n\rceil$ host grid for the 47 instances on which all four methods certify the same bandwidth.}
\label{fig:commercial-common-opt-runtime}
\end{figure}


\vspace{-0.6em}
\begin{figure}[pos=H]
\centering
\includegraphics[width=0.96\textwidth]{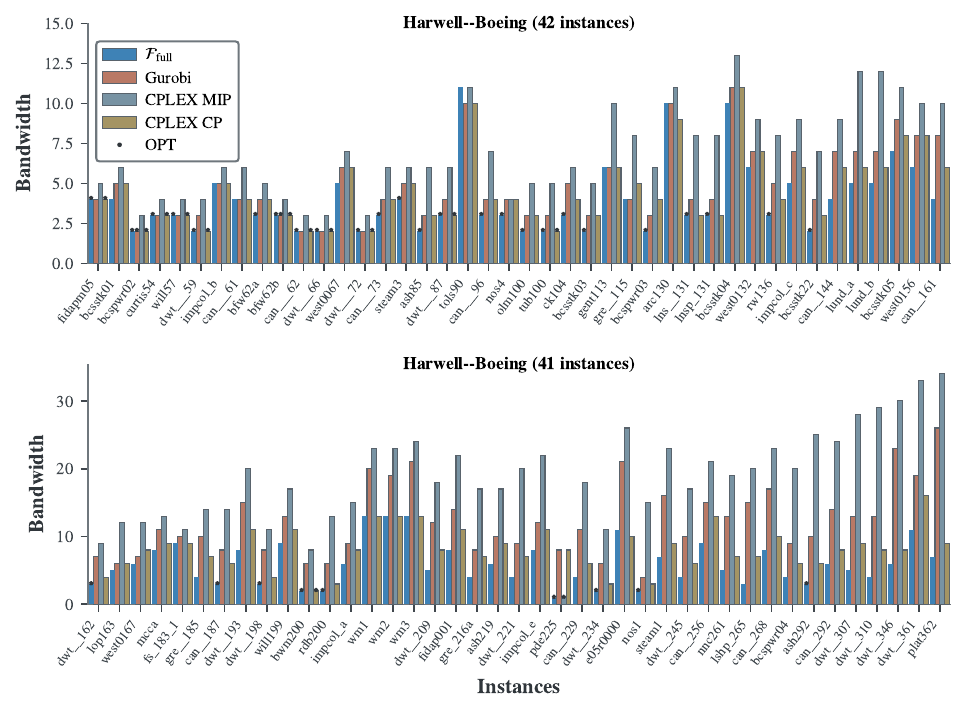}
\caption{Returned-bandwidth comparison of $\mathcal F_{\mathrm{full}}$, Gurobi, CPLEX MIP, and CPLEX CP on the $\lceil\sqrt n\rceil\times\lceil\sqrt n\rceil$ host grid for the 83 Harwell--Boeing instances with at least one \texttt{FEAS} outcome and differing returned bandwidths.}
\label{fig:commercial-timeout-values}
\end{figure}

Figures~\ref{fig:commercial-common-opt-runtime} and~\ref{fig:commercial-timeout-values} complement the certification counts with comparisons of solving time and returned bandwidth, respectively. Figure~\ref{fig:commercial-timeout-values} shows the 83 Harwell--Boeing instances on which at least one method has status \texttt{FEAS} and the returned bandwidths differ. On this subset, $\mathcal F_{\mathrm{full}}$ returns a lower bandwidth than Gurobi on 69 instances, CPLEX MIP on 81, and CPLEX CP on 59. Thus, $\mathcal F_{\mathrm{full}}$ generally returns lower bandwidths on these timeout-involving instances, although the improvement is not uniform across every comparison. Figure~\ref{fig:commercial-common-opt-runtime} then restricts the solving-time comparison to the 47 instances on which all four methods certify the same bandwidth. On this matched subset, $\mathcal F_{\mathrm{full}}$ is faster than Gurobi and CPLEX MIP on 44 of 47 instances and faster than CPLEX CP on 46 of 47. Hence, when both the certified bandwidth and status are matched, $\mathcal F_{\mathrm{full}}$ also has a lower solving time in most of the displayed cases.

\subsection{Additional Host-Grid Results}
\label{sec:exp-additional-host-grids}

\subsubsection{Results on the \texorpdfstring{$2\times\lceil n/2\rceil$}{2 x ceil(n/2)} Grid}
\label{sec:exp-2xn-grid}

This experiment compares $\mathcal F_{\mathrm{full}}$, Gurobi, CPLEX MIP, and CPLEX CP under the same local experimental setup and 3600~s time limit on the $2\times\lceil n/2\rceil$ host grid. Lemma~\ref{lem:two-row-correctness} establishes the exactness of the incremental SAT procedure on this host. Table~\ref{tab:two-row-results-summary} summarizes the outcome counts. On the Regular set, CPLEX MIP certifies all 43 instances, while $\mathcal F_{\mathrm{full}}$ and Gurobi certify 42 and CPLEX CP certifies 38. The distinction is clearer on Harwell--Boeing, where $\mathcal F_{\mathrm{full}}$ returns \texttt{OPT} on 23 instances, compared with 15 for Gurobi, 7 for CPLEX MIP, and 8 for CPLEX CP. When \texttt{OPT} and \texttt{FEAS} are counted together, the four methods return a feasible bandwidth on 64, 43, 33, and 62 instances, respectively. Thus, $\mathcal F_{\mathrm{full}}$ provides the broadest exact-certification coverage and returns a feasible bandwidth on the largest number of Harwell--Boeing instances, although 29 instances remain \texttt{TO}.

\begin{table}[pos=H]
\centering
\caption{Summary of results on the Regular and Harwell--Boeing datasets for $\mathcal F_{\mathrm{full}}$, Gurobi, CPLEX MIP, and CPLEX CP using the $2\times\lceil n/2\rceil$ host grid.}
\label{tab:two-row-results-summary}
\small
\setlength{\tabcolsep}{4pt}
\begin{tabular}{lrrrrrr}
\toprule
& \multicolumn{3}{c}{Regular (43)} & \multicolumn{3}{c}{Harwell--Boeing (93)} \\
\cmidrule(lr){2-4}\cmidrule(lr){5-7}
Method & \texttt{OPT} & \texttt{FEAS} & \texttt{TO} & \texttt{OPT} & \texttt{FEAS} & \texttt{TO} \\
\midrule
$\mathcal{F}_{\mathrm{full}}$ & 42 & 1 & \textbf{0} & \textbf{23} & 41 & \textbf{29} \\
Gurobi & 42 & 0 & 1 & 15 & 28 & 50 \\
CPLEX MIP & \textbf{43} & 0 & \textbf{0} & 7 & 26 & 60 \\
CPLEX CP & 38 & 0 & 5 & 8 & 54 & 31 \\
\bottomrule
\end{tabular}
\end{table}

Figure~\ref{fig:two-row-common-opt-solving-time} restricts the solving-time comparison to the 44 instances on which all four methods return the same bandwidth with status \texttt{OPT}. To avoid conflating runtime with differences in certification outcome, runtime comparisons are also reported on the subset of instances for which all compared methods certify the same bandwidth. On this matched subset, under the reported experimental configuration and single-run protocol, $\mathcal F_{\mathrm{full}}$ achieved a lower observed solving time than Gurobi, CPLEX MIP, and CPLEX CP on 41, 38, and 39 instances, respectively. These counts describe the observed single-run results and should not be interpreted as evidence of statistically significant runtime differences. The complete instance-level results for this host grid are reported in Table~\ref{tab:auxiliary-host-detail} in Appendix~\ref{app:detailed-results}.

Figure~\ref{fig:two-row-bandwidth-comparison} shows the 66 instances on which at least two methods differ in bandwidth or status. In the pairwise comparisons with Gurobi, CPLEX MIP, and CPLEX CP, $\mathcal F_{\mathrm{full}}$ returns a lower bandwidth or returns a bandwidth when the compared solver returns none on 35, 53, and 30 instances, respectively. The reverse occurs on 5, 1, and 11 instances. Thus, $\mathcal F_{\mathrm{full}}$ generally provides broader returned-bandwidth coverage and lower incumbents in these comparisons, although the difference is not uniform.

\begin{figure}[pos=h]
\centering
\includegraphics[width=0.96\textwidth]{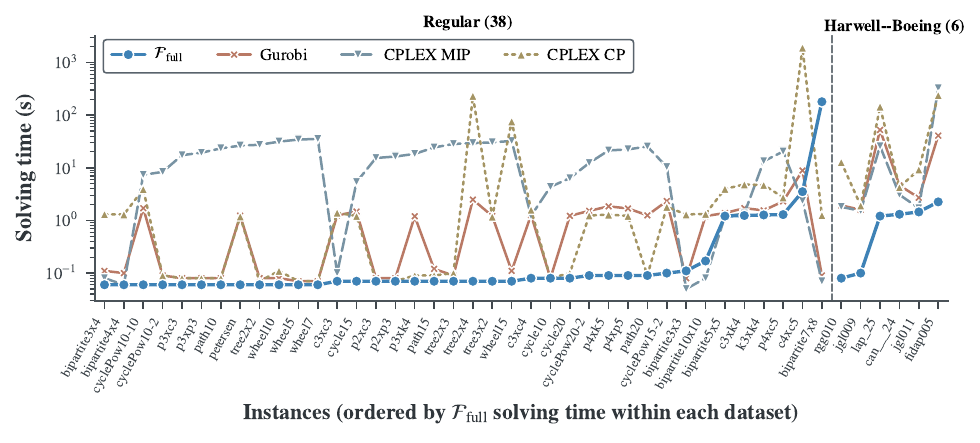}
\caption{Solving-time comparison of $\mathcal F_{\mathrm{full}}$, Gurobi, CPLEX MIP, and CPLEX CP on the $2\times\lceil n/2\rceil$ host grid for the 44 instances where all four methods return the same bandwidth with status \texttt{OPT}.}
\label{fig:two-row-common-opt-solving-time}
\end{figure}

\begin{figure}[pos=h]
\centering
\includegraphics[width=0.96\textwidth]{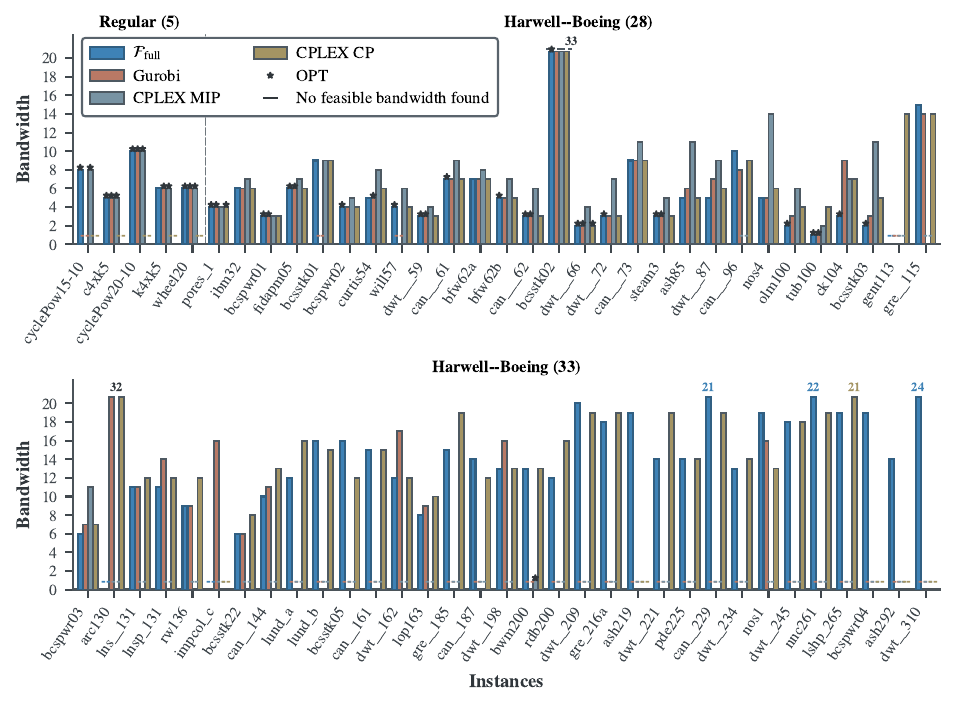}
\caption{Bandwidth comparison of $\mathcal F_{\mathrm{full}}$, Gurobi, CPLEX MIP, and CPLEX CP on the $2\times\lceil n/2\rceil$ host grid for the 66 instances where at least two methods differ in bandwidth or status.}
\label{fig:two-row-bandwidth-comparison}
\end{figure}

\FloatBarrier

\subsubsection{Results on the \texorpdfstring{$n\times n$}{n x n} Grid}
\label{sec:exp-extended-nxn}

These experiments are intended to assess the behavior of the exact SAT approach under alternative host geometries rather than to compare objective values across different host-grid settings. We further extend the matched local comparison of $\mathcal F_{\mathrm{full}}$, Gurobi, CPLEX MIP, and CPLEX CP to the $n\times n$ host grid, using the same 3600~s time limit and experimental protocol.

For the Harwell--Boeing benchmark, we restrict this experiment to 49 small-graph instances ($n\leq150$), as larger instances are substantially more challenging for exact approaches and are difficult to solve or certify within the imposed time limit. This selection provides a more informative comparison under the $n\times n$ host geometry.

Table~\ref{tab:nxn-results-summary} summarizes the outcomes for 43 Regular and 49 Harwell--Boeing instances. On the Regular set, $\mathcal F_{\mathrm{full}}$ certifies 38 instances, compared with 36 for CPLEX CP and 34 each for CPLEX MIP and Gurobi. On Harwell--Boeing, the corresponding counts are 32, 18, 2, and 5. Both $\mathcal F_{\mathrm{full}}$ and CPLEX CP return either an \texttt{OPT} or \texttt{FEAS} outcome on every evaluated instance, whereas CPLEX MIP has 1 Regular and 16 Harwell--Boeing \texttt{TO} outcomes and Gurobi has 19 Harwell--Boeing \texttt{TO} outcomes. Within this evaluated subset, the main distinction is therefore the broader exact-certification coverage of $\mathcal F_{\mathrm{full}}$.

\begin{table}[pos=ht]
\centering
\caption{Summary of results on the Regular and Harwell--Boeing datasets for $\mathcal F_{\mathrm{full}}$, Gurobi, CPLEX MIP, and CPLEX CP using the $n\times n$ host grid.}
\label{tab:nxn-results-summary}
\small
\setlength{\tabcolsep}{3.5pt}
\begin{tabular}{lrrr@{\hspace{3em}}rrr}
\toprule
& \multicolumn{3}{c}{Regular (43)} & \multicolumn{3}{c}{Harwell--Boeing (49)} \\
\cmidrule(lr){2-4}\cmidrule(lr){5-7}
Method & \texttt{OPT} & \texttt{FEAS} & \texttt{TO} & \texttt{OPT} & \texttt{FEAS} & \texttt{TO} \\
\midrule
$\mathcal{F}_{\mathrm{full}}$ & \textbf{38} & 5 & \textbf{0} & \textbf{32} & 17 & \textbf{0} \\
CPLEX CP & 36 & 7 & \textbf{0} & 18 & 31 & \textbf{0} \\
CPLEX MIP & 34 & 8 & 1 & 2 & 31 & 16 \\
Gurobi & 34 & 9 & \textbf{0} & 5 & 25 & 19 \\
\bottomrule
\end{tabular}
\end{table}


\begin{figure}[pos=H]
\centering
\includegraphics[width=0.9\textwidth]{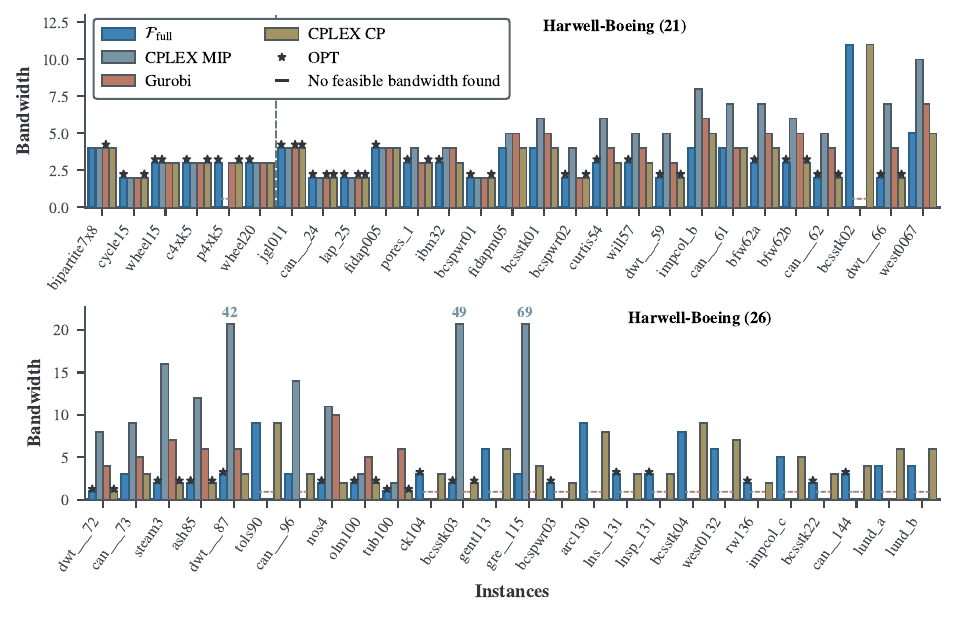}
\caption{Returned-bandwidth comparison of $\mathcal F_{\mathrm{full}}$, Gurobi, CPLEX MIP, and CPLEX CP on the 53 extended $n\times n$ host instances with differing returned bandwidths or statuses.}
\label{fig:nxn-timeout-values}
\end{figure}

\begin{figure}[pos=H]
\centering
\includegraphics[width=0.9\textwidth]{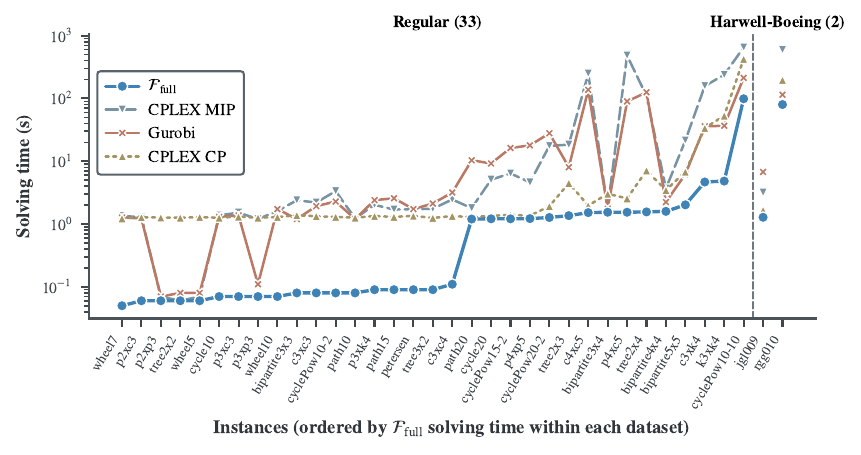}
\caption{Solving-time comparison of $\mathcal F_{\mathrm{full}}$, Gurobi, CPLEX MIP, and CPLEX CP on the $n\times n$ host grid for the 35 instances certifying the same bandwidth.}
\label{fig:nxn-common-opt-solving-time}
\end{figure}
Figure~\ref{fig:nxn-timeout-values} shows the 53 instances on which the returned bandwidths or statuses differ. In the pairwise comparisons with Gurobi, CPLEX MIP, and CPLEX CP, $\mathcal F_{\mathrm{full}}$ returns a lower bandwidth or returns a bandwidth when the compared solver returns none on 39, 43, and 9 instances, respectively. The reverse does not occur against Gurobi or CPLEX MIP and occurs only on \texttt{arc130} against CPLEX CP. These comparisons complement the certification counts by describing returned-bandwidth coverage and quality on the displayed subset.

Figure~\ref{fig:nxn-common-opt-solving-time} restricts the solving-time comparison to the 35 instances on which all four methods certify the same bandwidth, comprising 33 Regular and 2 Harwell--Boeing instances. On this matched subset, $\mathcal F_{\mathrm{full}}$ is faster than Gurobi and CPLEX CP on all 35 instances and faster than CPLEX MIP on 34, with one tie. The complete instance-level results for this host grid are reported in Table~\ref{tab:auxiliary-host-detail}.

\section{Conclusion}\label{sec:conclusion}

This paper presented an exact SAT-based approach for the two-dimensional bandwidth minimization problem (2DBMP), together with an incremental threshold-search procedure for certifying optimal bandwidths. The approach is designed for the standard square-root-grid host and combines a compact SAT representation of vertex placement, occupancy, and Manhattan-distance constraints with symmetry breaking and incremental solving.

On the standard $\lceil\sqrt n\rceil\times\lceil\sqrt n\rceil$ host grid, the proposed method certifies optimal bandwidths for 41 of 43 Regular instances and 42 of 93 Harwell--Boeing instances within the 3600~s time limit, substantially extending the optimality certification achieved by the previously published exact approach under a 72-hour time limit.
It also certifies four bandwidth values that are lower than those reported by DRSA, the strongest heuristic method among the published approaches considered in this study; three of these improve upon all previously published comparison values and are certified optimal.
Compared with the commercial exact solvers under the reported local experimental protocol, $\mathcal F_{\mathrm{full}}$ generally achieves substantially broader optimality certification, particularly on the Harwell--Boeing benchmark, while maintaining competitive solving times on matched instances. On the standard host grid, it certifies 42 of 93 Harwell--Boeing instances, compared with 10, 8, and 20 for Gurobi, CPLEX MIP, and CPLEX CP, respectively.
The additional experiments on $2\times\lceil n/2\rceil$ and $n\times n$ host grids further demonstrate that the approach remains effective under alternative host geometries.

The results should be interpreted primarily as evidence of the effectiveness of SAT for exact certification, rather than as evidence of dominance over heuristic methods, which may still provide better feasible incumbents when the exact solver reaches the time limit without completing a certificate. Moreover, the reported runtime comparisons are based on single runs under the specified experimental protocol, while published runtimes were obtained under different hardware and experimental settings and are therefore not directly comparable or statistically validated. Overall, the proposed approach provides an effective exact method for the small- and medium-sized benchmark instances considered in this study, with fewer than 400 vertices, while larger and more challenging instances, particularly dense guest graphs, remain difficult for exact approaches. Future work will focus on stronger lower bounds, additional sound symmetry reductions, and more compact incremental encodings to improve scalability.

\section*{CRediT authorship contribution statement}

Pham Quang Minh: Methodology, Software, Validation, Writing -- review and editing. Dao Xuan Nghia: Methodology, Supervision, Validation, Writing -- original draft, Writing -- review and editing. To Van Khanh: Conceptualization, Investigation, Supervision, Writing -- review and editing.

\section*{Declaration of competing interest}
The authors declare that they have no known competing financial interests or personal relationships that could have appeared to influence the work reported in this paper.

\section*{Data availability}
The benchmark result files and source code used in this manuscript are available in the project repository: \url{https://github.com/PQMlmaoxd/Two-Dimensional-Bandwidth-Minimization-Problem}.

\appendix
\counterwithin{table}{section}
\counterwithin{figure}{section}
\makeatletter
\renewcommand{\@seccntformat}[1]{%
  \ifstrequal{#1}{section}{Appendix~\csname the#1\endcsname\quad}{%
    \csname the#1\endcsname\quad}}
\makeatother
\section{ILP Linearization}
\label{app:ilp-linearization}

This appendix records the linearization of the absolute-value constraint used
in the reference ILP formulation.

\begin{lemma}[Linearization of the absolute-value constraint]
\label{lem:ilp-split-distance}
For a fixed placement and an edge $e=(u,v)$, let
$x_u=\sum_{r=1}^{h}\sum_{c=1}^{h} r\,a_{urc}$,
$y_u=\sum_{r=1}^{h}\sum_{c=1}^{h} c\,a_{urc}$,
$x_v=\sum_{r=1}^{h}\sum_{c=1}^{h} r\,a_{vrc}$, and
$y_v=\sum_{r=1}^{h}\sum_{c=1}^{h} c\,a_{vrc}$ denote the induced row and
column coordinates. If the nonnegative split variables satisfy
Constraints~\eqref{eq:dx_split} and~\eqref{eq:dy_split} for an edge
$e=(u,v)$, then
\[
|x_u-x_v|+|y_u-y_v|
\le
d^{x,+}_{e}+d^{x,-}_{e}+d^{y,+}_{e}+d^{y,-}_{e}
\]
Consequently, Constraint~\eqref{eq:beta_bound} implies that edge $e$ has
Manhattan distance at most $\beta$. Conversely, any fixed placement whose
edge $e$ has Manhattan distance at most $\beta$ can be extended to
nonnegative split variables satisfying
Constraints~\eqref{eq:dx_split}--\eqref{eq:beta_bound}.
\end{lemma}

\begin{proof}
Constraint~\eqref{eq:dx_split} gives
$x_u-x_v=d^{x,+}_{e}-d^{x,-}_{e}$. Since
$d^{x,+}_{e}$ and $d^{x,-}_{e}$ are nonnegative, the triangle inequality gives
\[
|x_u-x_v|
=|d^{x,+}_{e}-d^{x,-}_{e}|
\le d^{x,+}_{e}+d^{x,-}_{e}
\]
The same argument applied to Constraint~\eqref{eq:dy_split} gives
\[
|y_u-y_v|
\le d^{y,+}_{e}+d^{y,-}_{e}
\]
Adding the two inequalities and using Constraint~\eqref{eq:beta_bound}
proves the forward implication.

For the converse direction, set
$d^{x,+}_{e}=\max\{x_u-x_v,0\}$ and
$d^{x,-}_{e}=\max\{x_v-x_u,0\}$, and choose
$d^{y,+}_{e}$ and $d^{y,-}_{e}$ analogously for the column difference. These
values are nonnegative and at most $h-1$ because all coordinates lie in
$\{1,\ldots,h\}$. They satisfy Constraints~\eqref{eq:dx_split}
and~\eqref{eq:dy_split}, and their total equals
$|x_u-x_v|+|y_u-y_v|$. Therefore Constraint~\eqref{eq:beta_bound} holds
whenever the edge has Manhattan distance at most $\beta$.
\end{proof}

\section{Correctness of the SAT Encoding}\label{app:correctness}

The first four lemmas establish the correctness of the SAT components, leading to Proposition~\ref{prop:full-correctness} for the proposed full encoding. The remaining results establish the equivalence of the alternative configurations, the exactness of the incremental search, and the exactness of the encoding on the $2\times\lceil n/2\rceil$ host.

\begin{lemma}[Position and occupancy correctness]
\label{lem:position-occupancy}
Every satisfying assignment of either
$\mathcal F_{\mathrm{pos}}\land\mathcal F_{\mathrm{occ}}^{\mathrm{seq}}$
or
$\mathcal F_{\mathrm{pos}}\land\mathcal F_{\mathrm{occ}}^{\mathrm{pair}}$
determines an injective embedding $\varphi:V\rightarrow H_h$ through the position variables. Conversely, every injective embedding is represented by at least one satisfying assignment under each occupancy encoding. Thus, replacing $\mathcal F_{\mathrm{occ}}^{\mathrm{seq}}$ with $\mathcal F_{\mathrm{occ}}^{\mathrm{pair}}$ does not change the set of feasible embeddings.
\end{lemma}

\begin{proof}
First, let an assignment satisfy either $\mathcal F_{\mathrm{pos}}\land\mathcal F_{\mathrm{occ}}^{\mathrm{seq}}$ or $\mathcal F_{\mathrm{pos}}\land\mathcal F_{\mathrm{occ}}^{\mathrm{pair}}$. Equations~\eqref{eq:sat_x_exact} and~\eqref{eq:sat_y_exact} ensure that each vertex $v$ has a unique true row variable $X_{v,r_v}$ and a unique true column variable $Y_{v,c_v}$. Define $\varphi(v)=(r_v,c_v)$.

Suppose first that the assignment satisfies $\mathcal F_{\mathrm{occ}}^{\mathrm{seq}}$. Equations~\eqref{eq:cell_to_x}--\eqref{eq:xy_to_cell} ensure that $A_{v,r,c}$ is true exactly when both $X_{v,r}$ and $Y_{v,c}$ are true. If two distinct vertices $u$ and $v$ were mapped to the same grid position $(r,c)$, then both $A_{u,r,c}$ and $A_{v,r,c}$ would be true, contradicting the at-most-one constraint in Equation~\eqref{eq:cell_amo}. Suppose instead that the assignment satisfies $\mathcal F_{\mathrm{occ}}^{\mathrm{pair}}$. If two distinct vertices $u$ and $v$ were mapped to the same position $(r,c)$, the clause in Equation~\eqref{eq:pairwise_occupancy} indexed by $u$, $v$, and $(r,c)$ would have all four literals false. This contradicts the assumption that $\mathcal F_{\mathrm{occ}}^{\mathrm{pair}}$ is satisfied. Thus, in both cases, no two vertices share a grid position, so $\varphi$ is injective.

For the converse direction, let $\varphi:V\rightarrow H_h$ be an injective embedding. Set $X_{v,r}$ true exactly when $r$ is the row of $\varphi(v)$, and set $Y_{v,c}$ true exactly when $c$ is its column. These assignments satisfy the exactly-one requirements in Equations~\eqref{eq:sat_x_exact} and~\eqref{eq:sat_y_exact}, and the auxiliary variables of their CNF encodings can be assigned accordingly.

For $\mathcal F_{\mathrm{occ}}^{\mathrm{seq}}$, set $A_{v,r,c}$ true exactly when $\varphi(v)=(r,c)$. The linking clauses in Equations~\eqref{eq:cell_to_x}--\eqref{eq:xy_to_cell} then hold. Injectivity ensures that at most one variable in $\{A_{v,r,c}:v\in V\}$ is true for each grid position, so the sequential-counter auxiliary variables can be assigned to satisfy Equation~\eqref{eq:cell_amo}. For $\mathcal F_{\mathrm{occ}}^{\mathrm{pair}}$, injectivity ensures that no two distinct vertices occupy the same position. Hence, every clause in Equation~\eqref{eq:pairwise_occupancy} contains at least one true literal. Thus, the assignment can be extended to satisfy either conjunction.
\end{proof}

\begin{lemma}[Distance-threshold semantics and cutoff validity]
\label{lem:distance-threshold}
Suppose that the position variables satisfy $\mathcal F_{\mathrm{pos}}$. For each edge $e=(u,v)\in E$, define $a_x(e)=|x_u-x_v|$ and $a_y(e)=|y_u-y_v|$. Every satisfying extension to $\mathcal F_{\mathrm{dist}}$ satisfies $T^q_{e,d}$ if and only if $d\le a_q(e)$ for $q\in\{x,y\}$ and $d\in\{1,\ldots,h-1\}$, and every position assignment has such an extension. A position assignment can be extended to satisfy $\mathcal F_{\mathrm{dist}}^{U^*}$ if and only if $a_q(e)\le U^*$ for every edge $e$ and axis $q$. In every such extension, $T^q_{e,d}$ is true if and only if $d\le a_q(e)$ for $d\in\{1,\ldots,\min\{U^*,h-1\}\}$. Consequently, the cutoff-aware encoding preserves every embedding whose bandwidth is at most $U^*$.
\end{lemma}

\begin{proof}
Fix an edge $e=(u,v)$ and an axis $q\in\{x,y\}$. The position component selects unique coordinates for $u$ and $v$ on this axis, so their separation $a=a_q(e)$ is fixed. If $a=0$ and threshold $1$ is represented, Equation~\eqref{eq:dist_zero} forces $T^q_{e,1}$ to be false. Equation~\eqref{eq:dist_mono} then prevents any higher represented threshold from being true, because a true $T^q_{e,d}$ would propagate down to $T^q_{e,1}$. If no threshold is represented, the threshold-semantics statement is vacuous.

Now suppose that $a>0$. According to the order of the two selected coordinates, Equation~\eqref{eq:dist_act_1} or Equation~\eqref{eq:dist_act_2} forces $T^q_{e,a}$ to be true whenever that threshold is represented. Repeated application of Equation~\eqref{eq:dist_mono} therefore forces every $T^q_{e,d}$ with $d\le a$ to be true. If $a+1$ is represented, Equation~\eqref{eq:dist_deact_1} or Equation~\eqref{eq:dist_deact_2} forces $T^q_{e,a+1}$ to be false. Any true threshold above $a+1$ would propagate down to $T^q_{e,a+1}$ and produce a contradiction, while no threshold above $a$ exists when $a$ is the largest represented value. Thus, $T^q_{e,d}$ is true exactly for the represented thresholds satisfying $d\le a$.

For the reverse direction, set each represented $T^q_{e,d}$ true exactly when $d\le a_q(e)$. The activation, deactivation, zero-separation, and monotonicity clauses in Equations~\eqref{eq:dist_act_1}--\eqref{eq:dist_mono} then hold by construction. In the cutoff-aware encoding, Equation~\eqref{eq:cutoff_pair} is satisfied exactly when no edge has an axis separation greater than $U^*$. If an embedding has bandwidth at most $U^*$, then $a_x(e)+a_y(e)\le U^*$ for every edge $e$. Since both axis separations are nonnegative, $a_x(e)\le U^*$ and $a_y(e)\le U^*$. Equation~\eqref{eq:cutoff_pair} is therefore satisfied, and the threshold assignment above gives a satisfying extension to $\mathcal F_{\mathrm{dist}}^{U^*}$.
\end{proof}

\begin{lemma}[Bandwidth-component correctness]
\label{lem:bandwidth-component}
Let the position variables satisfy $\mathcal F_{\mathrm{pos}}$. For every integer $K$ with $0\le K\le U^*$, the position assignment can be extended to satisfy $\mathcal F_{\mathrm{dist}}^{U^*}$ together with the form of $\Delta_K$ in Equation~\eqref{eq:delta_k_cutoff} if and only if every guest edge has Manhattan distance at most $K$. For every integer $K$ with $0\le K\le 2(h-1)$, the same equivalence holds for $\mathcal F_{\mathrm{dist}}$ together with the form of $\Delta_K$ in Equation~\eqref{eq:delta_k_full}.
\end{lemma}

\begin{proof}
Consider either distance encoding and let $L=\min\{U^*,h-1\}$ in the cutoff-aware case and $L=h-1$ in the full-distance case. Fix an edge $e=(u,v)$ and define $a_x(e)=|x_u-x_v|$ and $a_y(e)=|y_u-y_v|$.

Suppose first that the distance encoding and the applicable form of $\Delta_K$ are satisfied. Lemma~\ref{lem:distance-threshold} gives $0\le a_x(e),a_y(e)\le L$ and makes $T^x_{e,d}$ and $T^y_{e,d}$ true exactly up to the corresponding axis separations. If $a_x(e)\ge K+1$, then $K+1\le a_x(e)\le L$, so the unit clause $\neg T^x_{e,K+1}$ is present and is violated. The same argument on the other axis shows that $a_x(e)\le K$ and $a_y(e)\le K$. For $K=0$, these two inequalities give $a_x(e)+a_y(e)\le K$. For $K\ge1$, suppose that $a_x(e)+a_y(e)\ge K+1$. Both axis separations are then positive. Setting $i=a_x(e)$ gives $1\le i\le K$, $i\le L$, and $K-i+1\le a_y(e)\le L$. The corresponding binary clause in Equation~\eqref{eq:delta_k_cutoff} or Equation~\eqref{eq:delta_k_full} is violated because both threshold variables are true. Hence every edge satisfies $a_x(e)+a_y(e)\le K$.

For the other direction, suppose that every edge satisfies $a_x(e)+a_y(e)\le K$. In the cutoff-aware case, $a_x(e),a_y(e)\le K\le U^*$, and the grid-coordinate domains also give $a_x(e),a_y(e)\le h-1$. Hence, $a_x(e),a_y(e)\le\min\{U^*,h-1\}$. In the full-distance case, the grid-coordinate domains give $a_x(e),a_y(e)\le h-1$. Lemma~\ref{lem:distance-threshold} therefore provides the required threshold valuation for either distance encoding. Any represented unit clause at threshold $K+1$ is satisfied because neither axis separation exceeds $K$. If both threshold variables in a binary clause were true, then $a_x(e)\ge i$ and $a_y(e)\ge K-i+1$, which would imply $a_x(e)+a_y(e)\ge K+1$. Therefore every binary clause is also satisfied, and the conclusion holds for both forms of $\Delta_K$.
\end{proof}

\begin{lemma}[Q1 symmetry preservation]
\label{lem:q1-symmetry}
Every injective embedding has an embedding obtained by horizontal and vertical reflections whose anchor vertex lies in $Q_1$. These reflections preserve injectivity and the Manhattan distance of every guest edge, so adding $\mathcal F_{\mathrm{sym}}$ does not change bandwidth-$K$ feasibility.
\end{lemma}

\begin{proof}
Let $\varphi$ be an injective embedding and write $\varphi(a)=(r_a,c_a)$ for the position of the selected anchor vertex. If $r_a>\lceil h/2\rceil$, reflect every grid position horizontally by replacing its row $r$ with $h+1-r$. The new anchor row satisfies $h+1-r_a\le\lfloor h/2\rfloor\le\lceil h/2\rceil$. If $c_a>\lceil h/2\rceil$, apply the analogous vertical reflection by replacing every column $c$ with $h+1-c$, which places the anchor column in the same allowed range.

Each reflection is a bijection of $H_h$, so applying it to all vertices preserves injectivity. A horizontal reflection preserves row differences because $|(h+1-r_u)-(h+1-r_v)|=|r_u-r_v|$, and a vertical reflection preserves column differences by the same argument. Their composition therefore preserves the Manhattan distance of every guest edge and leaves the bandwidth unchanged. The transformed anchor lies in $Q_1$, so its position variables satisfy Equations~\eqref{eq:q1_x} and~\eqref{eq:q1_y}, establishing the claimed preservation.
\end{proof}

\begin{proposition}[Correctness of $\mathcal F_{\mathrm{full}}(K)$]
\label{prop:full-correctness}
For every integer $K$ with $0\le K\le U^*$, the formula $\mathcal F_{\mathrm{full}}(K)$ is satisfiable if and only if there exists an injective embedding $\varphi:V\to H_h$ such that $\beta(G,\varphi)\le K$.
\end{proposition}

\begin{proof}
Suppose that $\mathcal F_{\mathrm{full}}(K)$ is satisfiable. By Lemma~\ref{lem:position-occupancy}, its position and sequential occupancy components determine an injective embedding $\varphi:V\to H_h$. Lemma~\ref{lem:bandwidth-component}, applied to $\mathcal F_{\mathrm{dist}}^{U^*}$ and the form of $\Delta_K$ in Equation~\eqref{eq:delta_k_cutoff}, shows that every guest edge has Manhattan distance at most $K$. It follows that $\beta(G,\varphi)\le K$.

For the reverse direction, let $\varphi$ be an injective embedding with $\beta(G,\varphi)\le K$. Lemma~\ref{lem:q1-symmetry} provides an injective embedding $\varphi'$ with the same bandwidth whose anchor lies in $Q_1$. Lemma~\ref{lem:position-occupancy} extends $\varphi'$ to a satisfying assignment of $\mathcal F_{\mathrm{pos}}\land\mathcal F_{\mathrm{occ}}^{\mathrm{seq}}$, while Lemmas~\ref{lem:distance-threshold} and~\ref{lem:bandwidth-component} extend it to satisfy $\mathcal F_{\mathrm{dist}}^{U^*}\land\Delta_K$. Since the auxiliary variables introduced by the occupancy and distance components are disjoint, these extensions can be combined. The anchor position satisfies $\mathcal F_{\mathrm{sym}}$, so all components in Equation~\eqref{eq:full_encoding} are satisfied.
\end{proof}

\begin{lemma}[Equivalence of the alternative configurations]
\label{lem:configuration-equivalence}
For every $K$ with $0\le K\le U^*$, each alternative encoding configuration in Table~\ref{tab:encoding-configurations} is satisfiable if and only if $\mathcal F_{\mathrm{full}}(K)$ is satisfiable.
\end{lemma}

\begin{proof}
Any satisfying assignment of $\mathcal F_{\mathrm{full}}(K)$ satisfies $\mathcal F_{\mathrm{no\text{-}sym}}(K)$ after the symmetry clauses are omitted. In the other direction, Lemmas~\ref{lem:position-occupancy} and~\ref{lem:bandwidth-component} show that a model of $\mathcal F_{\mathrm{no\text{-}sym}}(K)$ determines an injective embedding of bandwidth at most $K$. Proposition~\ref{prop:full-correctness} then shows that $\mathcal F_{\mathrm{full}}(K)$ is satisfiable.

For $\mathcal F_{\mathrm{pair}}(K)$, Lemma~\ref{lem:position-occupancy} shows that pairwise and sequential occupancy represent the same injective embeddings. A model of $\mathcal F_{\mathrm{pair}}(K)$ therefore yields an embedding of bandwidth at most $K$ by Lemma~\ref{lem:bandwidth-component}, and Proposition~\ref{prop:full-correctness} gives a model of $\mathcal F_{\mathrm{full}}(K)$. Conversely, an embedding supplied by Proposition~\ref{prop:full-correctness} can be reflected into $Q_1$ by Lemma~\ref{lem:q1-symmetry}, represented with pairwise occupancy by Lemma~\ref{lem:position-occupancy}, and extended to the cutoff-aware distance and bandwidth components by Lemmas~\ref{lem:distance-threshold} and~\ref{lem:bandwidth-component}.

For $\mathcal F_{\mathrm{full\text{-}dist}}(K)$, a satisfying assignment determines an injective embedding by Lemma~\ref{lem:position-occupancy}, and Lemma~\ref{lem:bandwidth-component} shows that its bandwidth is at most $K$. Proposition~\ref{prop:full-correctness} then gives a satisfying assignment of $\mathcal F_{\mathrm{full}}(K)$. Conversely, Proposition~\ref{prop:full-correctness} supplies an injective bandwidth-$K$ embedding from a model of $\mathcal F_{\mathrm{full}}(K)$. Lemma~\ref{lem:q1-symmetry} places its anchor in $Q_1$, Lemma~\ref{lem:position-occupancy} represents its position and occupancy components, and Lemmas~\ref{lem:distance-threshold} and~\ref{lem:bandwidth-component} extend it to satisfy the full-distance component and the applicable form of $\Delta_K$. This completes the equivalence of all four configurations.
\end{proof}

\begin{lemma}[Threshold monotonicity and cumulative threshold clauses]
\label{lem:incremental-monotonicity}
Let $\mathcal F_{\mathrm{cfg}}(K)$ denote an encoding configuration in Table~\ref{tab:encoding-configurations}. For configurations containing $\mathcal F_{\mathrm{dist}}^{U^*}$, satisfiability is monotone over $0\le K\le U^*$. For $\mathcal F_{\mathrm{full\text{-}dist}}(K)$, satisfiability is monotone over $0\le K\le 2(h-1)$. Moreover, let $\mathcal B_{\mathrm{cfg}}$ denote the conjunction of the $K$-independent components of the selected configuration. For any decreasing sequence $U^*\ge K_1>\cdots>K_t\ge1$ in a cutoff-aware configuration, or $2(h-1)\ge K_1>\cdots>K_t\ge1$ in $\mathcal F_{\mathrm{full\text{-}dist}}(K)$, the cumulative formula $\mathcal B_{\mathrm{cfg}}\land\bigwedge_{j=1}^{t}\Delta_{K_j}$ is satisfiable if and only if $\mathcal F_{\mathrm{cfg}}(K_t)$ is satisfiable.
\end{lemma}

\begin{proof}
For configurations containing $\mathcal F_{\mathrm{dist}}^{U^*}$, Proposition~\ref{prop:full-correctness} and Lemma~\ref{lem:configuration-equivalence} show that $\mathcal F_{\mathrm{cfg}}(K)$ is satisfiable exactly when an injective embedding of bandwidth at most $K$ exists over $0\le K\le U^*$. For $\mathcal F_{\mathrm{full\text{-}dist}}(K)$, the same characterization over $0\le K\le 2(h-1)$ follows from Lemmas~\ref{lem:position-occupancy}, \ref{lem:distance-threshold}, \ref{lem:bandwidth-component}, and~\ref{lem:q1-symmetry}. In either case, the same embedding has bandwidth at most every larger threshold in the applicable range, so the corresponding configuration is satisfiable at each such threshold. This establishes threshold monotonicity for every configuration in Table~\ref{tab:encoding-configurations}.

The cumulative formula contains $\mathcal B_{\mathrm{cfg}}\land\Delta_{K_t}$, so any model of the cumulative formula is also a model of $\mathcal F_{\mathrm{cfg}}(K_t)$. For the reverse direction, let an assignment satisfy $\mathcal F_{\mathrm{cfg}}(K_t)$. Its position and threshold variables describe an embedding of bandwidth at most $K_t$ by Lemmas~\ref{lem:distance-threshold} and~\ref{lem:bandwidth-component}. Since $K_j>K_t$ for every earlier threshold, the same threshold valuation satisfies each $\Delta_{K_j}$ by Lemma~\ref{lem:bandwidth-component}. The assignment therefore satisfies the complete cumulative formula.

Any learned clause retained by a sound SAT solver is a logical consequence of the clauses present when it was learned. Every later cumulative formula contains those earlier clauses, so retaining learned clauses also preserves the satisfiability equivalence at the current threshold. Therefore, keeping both earlier threshold clauses and learned clauses changes the solver state but not the bandwidth decision problem being solved.
\end{proof}

\begin{proposition}[Exactness of the incremental search]
\label{prop:incremental-optimality}
Assume that $G$ contains at least one edge whose endpoints are distinct, $U^*$ is a valid upper bound, and the incremental procedure in Table~\ref{tab:incremental-procedure} uses an encoding configuration from Table~\ref{tab:encoding-configurations}. If no timeout occurs, the procedure terminates with status \textnormal{\textsc{optimal}} and returns the certified optimal value $\beta^*(G)$. If a timeout occurs, any returned incumbent is only a feasible upper bound.
\end{proposition}

\begin{proof}
Because $U^*$ is a valid upper bound, an injective embedding of bandwidth at most $U^*$ exists. For a configuration containing $\mathcal F_{\mathrm{dist}}^{U^*}$, Proposition~\ref{prop:full-correctness} and Lemma~\ref{lem:configuration-equivalence} show that the selected configuration is satisfiable at its initial threshold $U^*$. For $\mathcal F_{\mathrm{full\text{-}dist}}(K)$, every embedding on the $h\times h$ grid has bandwidth at most $2(h-1)$, and Lemmas~\ref{lem:position-occupancy}, \ref{lem:distance-threshold}, \ref{lem:bandwidth-component}, and~\ref{lem:q1-symmetry} show that this configuration is satisfiable at its initial threshold $2(h-1)$. Hence, unless the first solver call times out, the procedure obtains a feasible incumbent before any unsatisfiable result can occur.

Consider a satisfiable call at threshold $K$. The decoded model gives an embedding with actual bandwidth $B\le K$, so $B$ is a feasible upper bound and $\beta^*(G)\le B$. The next threshold is $B-1$, which is strictly smaller than $K$. If $B>\beta^*(G)$, integrality gives $B-1\ge\beta^*(G)$, and the next decision problem remains satisfiable. If $B=\beta^*(G)>1$, the next threshold is below the optimal value and is unsatisfiable. Lemma~\ref{lem:incremental-monotonicity} ensures that the cumulative solver state has exactly the same satisfiability status as the selected configuration at this new threshold.

The tested thresholds decrease after every satisfiable call, so the process is finite when no timeout occurs. It either reaches an incumbent $B=\beta^*(G)>1$ followed by an unsatisfiable call at $B-1$, or reaches a feasible incumbent $B=1$. In the first case, feasibility at $B$ and infeasibility at $B-1$ certify that $B$ is the optimal value. In the second case, injectivity and the existence of an edge with distinct endpoints imply $\beta^*(G)\ge1$, so the feasible value $1$ is also certified optimal. The smart jump omits only thresholds from $B$ through $K-1$, all of which admit the current embedding. Every embedding with bandwidth smaller than $B$ remains feasible at the next tested threshold $B-1$, so the jump excludes no improving embedding.

A timeout establishes neither satisfiability nor unsatisfiability at the current threshold. Any incumbent obtained earlier remains a feasible upper bound, but the procedure has not excluded a smaller bandwidth.
\end{proof}

\begin{lemma}[Exactness on the $2\times\lceil n/2\rceil$ host]
\label{lem:two-row-correctness}
Let $U$ be a feasible initial threshold for $\mathcal F_{\mathrm{full}}(K)$ instantiated on the $2\times\lceil n/2\rceil$ host. For every integer $K$ with $0\le K\le U$, $\mathcal F_{\mathrm{full}}(K)$ on the $2\times\lceil n/2\rceil$ host is satisfiable if and only if there exists an injective embedding on the $2\times\lceil n/2\rceil$ host with bandwidth at most $K$. Consequently, a feasible bandwidth $B$ is a certified optimal value if $B=1$ or if the incremental solver returns \textnormal{\textsc{unsat}} at $B-1$. If a solver call times out before either condition is established, any returned incumbent is only a feasible upper bound.
\end{lemma}

\begin{proof}
On the $2\times\lceil n/2\rceil$ host, the row-coordinate domain is $\{1,2\}$ and the column-coordinate domain is $\{1,\ldots,\lceil n/2\rceil\}$. The position and occupancy argument in Lemma~\ref{lem:position-occupancy} remains valid for these two coordinate domains and therefore represents exactly the injective embeddings on the $2\times\lceil n/2\rceil$ host.

For an embedding $\varphi$, write $\varphi(v)=(r_v,c_v)$. For every guest edge $(u,v)$, the row separation $|r_u-r_v|$ is at most $1$, while the column separation $|c_u-c_v|$ is at most $\lceil n/2\rceil-1$. The row-distance thresholds range up to $\min\{U,1\}$, and the column-distance thresholds range up to $\min\{U,\lceil n/2\rceil-1\}$. Applying the argument in Lemma~\ref{lem:distance-threshold} separately to these two threshold ranges shows that the corresponding threshold variables represent the row and column separations exactly. The argument in Lemma~\ref{lem:bandwidth-component} then shows that the distance components and $\Delta_K$ are satisfiable if and only if
\[
|r_u-r_v|+|c_u-c_v|\le K
\]
for every guest edge. The left-hand side is the Manhattan distance between the positions of $u$ and $v$.

A row reflection maps $r$ to $3-r$, while a column reflection maps $c$ to $\lceil n/2\rceil+1-c$. Both reflections are bijections of the $2\times\lceil n/2\rceil$ host and preserve the row and column separations. The symmetry-preservation argument in Lemma~\ref{lem:q1-symmetry} therefore remains valid. Adding $\mathcal F_{\mathrm{sym}}$ does not change bandwidth-$K$ feasibility. This proves the stated satisfiability equivalence.

The argument in Lemma~\ref{lem:incremental-monotonicity} applies with the row and column threshold ranges above, so the cumulative solver state at each tested threshold has the same satisfiability status as $\mathcal F_{\mathrm{full}}(K)$ on the $2\times\lceil n/2\rceil$ host. The certification and timeout arguments in Proposition~\ref{prop:incremental-optimality} therefore apply unchanged.
\end{proof}

\section{Initial Upper Bound for the
\texorpdfstring{$2\times\lceil n/2\rceil$}{2 x ceil(n/2)} Grid}
\label{app:initial-bound-2xm}

This section describes how a valid initial upper bound is obtained for the $2\times\lceil n/2\rceil$ grid. Throughout this section, $m=\lceil n/2\rceil$ and $H_{2,m}=[2]\times[m]$. Because $\delta(n)$ concerns unrestricted grid positions rather than positions restricted to $H_{2,m}$, it supplies only the starting scale for the candidate thresholds. The initialization then combines these candidates with an instance-specific lower bound and SAT feasibility checks, and adopts a candidate as the initial upper bound only after SAT verifies a feasible embedding at that threshold.

The lower bound combines a maximum-degree capacity condition with a connected-component condition. For an integer bandwidth threshold $B\ge1$, Equation~\eqref{eq:2xm-strip-capacity} presents the strip-capacity term $S_m(B)$:
\begin{equation}
S_m(B)
=
\bigl(\min\{2B+1,m\}-1\bigr)+\min\{2B-1,m\}.
\label{eq:2xm-strip-capacity}
\end{equation}
For any position $p\in H_{2,m}$, the first summand in Equation~\eqref{eq:2xm-strip-capacity} bounds the number of other positions in the same row as $p$ that lie within Manhattan distance $B$ of $p$. The second summand in Equation~\eqref{eq:2xm-strip-capacity} bounds the corresponding number of positions in the other row. Here, $\Delta(G)$ denotes the maximum degree of $G$, and $D_C=\operatorname{diam}_G(C)$ for each connected component $C$ with $|C|\ge2$. The capacity requirement $\Delta(G)\le S_m(B)$ leads to the maximum-degree bound, while comparing the host distance forced by $|C|$ distinct positions with $D_C$ leads to the component bound. Equations~\eqref{eq:2xm-strip-degree-lb} and~\eqref{eq:2xm-component-lb} present these bounds, respectively. Equation~\eqref{eq:2xm-lower-bound} takes their maximum, and Lemma~\ref{lem:2xm-lower-bound} establishes its validity:
\begin{align}
\mathrm{LB}_{\mathrm{strip\text{-}degree}}(G)
&=
\min\left\{
B\in\mathbb Z_{\ge1}:\Delta(G)\le S_m(B)
\right\},
\label{eq:2xm-strip-degree-lb}\\
\mathrm{LB}_{\mathrm{component}}(G)
&=
\max_{\substack{C\text{ connected component of }G\\|C|\ge2}}
\left\lceil
\frac{\lceil |C|/2\rceil}{D_C}
\right\rceil,
\label{eq:2xm-component-lb}\\
\mathrm{LB}_{2,m}(G)
&=
\max\left\{
\mathrm{LB}_{\mathrm{strip\text{-}degree}}(G),
\mathrm{LB}_{\mathrm{component}}(G)
\right\}.
\label{eq:2xm-lower-bound}
\end{align}
The lower bound determines where feasibility testing begins, whereas SAT determines whether a selected candidate is feasible. Equation~\eqref{eq:2xm-candidate-sequence} presents the candidate-threshold sequence generated from $\delta(n)$:
\begin{equation}
P_j
=
\min\left\{
m,
\left\lceil\left(1+\frac{j}{2}\right)\delta(n)\right\rceil
\right\},
\qquad j=0,1,2,\ldots.
\label{eq:2xm-candidate-sequence}
\end{equation}
The sequence starts with $\delta(n)$, increases its multiplier in steps of $0.5$, and caps each candidate at the host diameter $m$. The initialization selects the first candidate $P_j$ satisfying $P_j\ge\mathrm{LB}_{2,m}(G)$ and tests it with the cutoff-aware SAT encoding. If the result is \textsc{unsat}, the threshold is increased one unit at a time and the encoding is rebuilt until SAT verifies a feasible embedding. The first satisfiable threshold is a valid upper bound and is used as both the distance cutoff and the initial threshold of the descending search. By Lemma~\ref{lem:2xm-lower-bound}, candidate values below $\mathrm{LB}_{2,m}(G)$ are infeasible and need not be tested during initialization. In addition, this lower bound can be used to certify optimality when it equals a verified feasible bandwidth.

\begin{lemma}[Lower bound for $H_{2,m}$]
\label{lem:2xm-lower-bound}
Assume that $E\ne\varnothing$. Every injective embedding $\varphi:V\to H_{2,m}$ satisfies:
\begin{equation}
\beta(G,\varphi)\ge \mathrm{LB}_{2,m}(G).
\end{equation}
\end{lemma}

\begin{proof}
Fix an injective embedding $\varphi:V\to H_{2,m}$ and write $B=\beta(G,\varphi)$. Since $E\ne\varnothing$, $B\ge1$.
Choose a vertex $v^\star$ of maximum degree. At most $\min\{2B+1,m\}-1$ positions other than $\varphi(v^\star)$ in its row can be within Manhattan distance $B$. In the other row, the row separation is one, so at most $\min\{2B-1,m\}$ positions can be within distance $B$. Injectivity places the distinct neighbors of $v^\star$ in these positions, which gives $\Delta(G)\le S_m(B)$. Hence,
\begin{equation}
B\ge \mathrm{LB}_{\mathrm{strip\text{-}degree}}(G).
\end{equation}

For the component bound, fix a connected component $C$ with $|C|\ge2$, and let $q=\lceil |C|/2\rceil$. Any $|C|$ distinct positions in $H_{2,m}$ contain two positions at Manhattan distance at least $q$. To see this, let $a$ and $b$ be the minimum and maximum occupied columns and set $L=b-a$. Since each column contains at most two positions, $L\ge q-1$. If $L\ge q$, the claim follows from the column separation. If $L=q-1$, the number of occupied positions forces an opposite-row pair across the two extreme columns, whose Manhattan distance is $q$.

Let $u,v\in C$ occupy such a pair of positions. A shortest path between $u$ and $v$ in $G$ has length at most $D_C$, and every edge on this path has host distance at most $B$. The triangle inequality gives
\begin{equation}
q
\le d_{H_{2,m}}\bigl(\varphi(u),\varphi(v)\bigr)
\le B D_C.
\end{equation}
Therefore,
\begin{equation}
B
\ge
\left\lceil\frac{q}{D_C}\right\rceil
=
\left\lceil
\frac{\lceil |C|/2\rceil}{D_C}
\right\rceil.
\end{equation}
Taking the maximum over all non-singleton connected components gives $B\ge\mathrm{LB}_{\mathrm{component}}(G)$. Combining the two inequalities proves the claim.
\end{proof}

\section{Encoding-Size Comparison in the Ablation Study}
\label{app:ablation-encoding-size-comparison}

\begin{table}[pos=htbp]
\caption{Variable and clause counts for the SAT variants on the Regular and Harwell--Boeing sets using the $\lceil\sqrt n\rceil\times\lceil\sqrt n\rceil$ host grid. The counts exclude $\Delta_K$.}
\label{tab:ablation-encoding-size-aggregate}
\centering
\small
\begin{tabular}{@{}llrrrr@{}}
\toprule
Dataset & Variant & Median variables & Max variables & Median clauses & Max clauses \\
\midrule
Regular & $\mathcal F_{\mathrm{full}}$ & $716$ & $2{,}855$ & $2{,}494$ & $20{,}150$ \\
Regular & $\mathcal F_{\mathrm{no\text{-}sym}}$ & $716$ & $2{,}855$ & $2{,}648$ & $20{,}860$ \\
Regular & $\mathcal F_{\mathrm{pair}}$ & $312$ & $1{,}880$ & $2{,}508$ & $21{,}754$ \\
Regular & $\mathcal F_{\mathrm{full\text{-}dist}}$ & $716$ & $2{,}855$ & $2{,}494$ & $20{,}150$ \\
Regular & non-incremental $\mathcal F_{\mathrm{full}}$ & $716$ & $2{,}855$ & $2{,}494$ & $20{,}150$ \\
\addlinespace
Harwell--Boeing & $\mathcal F_{\mathrm{full}}$ & $69{,}263$ & $420{,}492$ & $512{,}972$ & $5{,}210{,}530$ \\
Harwell--Boeing & $\mathcal F_{\mathrm{no\text{-}sym}}$ & $69{,}263$ & $420{,}492$ & $526{,}444$ & $5{,}225{,}972$ \\
Harwell--Boeing & $\mathcal F_{\mathrm{pair}}$ & $20{,}340$ & $131{,}292$ & $2{,}225{,}780$ & $30{,}372{,}830$ \\
Harwell--Boeing & $\mathcal F_{\mathrm{full\text{-}dist}}$ & $69{,}263$ & $420{,}492$ & $512{,}972$ & $5{,}210{,}530$ \\
Harwell--Boeing & non-incremental $\mathcal F_{\mathrm{full}}$ & $69{,}263$ & $420{,}492$ & $512{,}972$ & $5{,}210{,}530$ \\
\bottomrule
\end{tabular}
\end{table}

Table~\ref{tab:ablation-encoding-size-aggregate} reports the median and maximum variable and clause counts after the fixed encoding components are constructed and before $\Delta_K$ is added. $\mathcal F_{\mathrm{pair}}$ has lower median and maximum variable counts but the largest maximum clause count in both datasets. For every instance in both datasets, $U^*\ge\lceil\sqrt n\rceil-1$, where $\lceil\sqrt n\rceil-1$ is the largest possible separation along one axis of the host grid. The cutoff therefore removes no distance thresholds, and the index set in Equation~\eqref{eq:cutoff_pair} is empty. Consequently, no cutoff clauses are generated, and $\mathcal F_{\mathrm{full}}$ and $\mathcal F_{\mathrm{full\text{-}dist}}$ have identical variable and clause counts. The non-incremental variant also matches $\mathcal F_{\mathrm{full}}$ because the two variants use the same encoded components before $\Delta_K$ is added.

Figure~\ref{fig:ablation-initial-size-scatter} gives an instance-level view of the occupancy-encoding tradeoff between $\mathcal F_{\mathrm{full}}$ and $\mathcal F_{\mathrm{pair}}$. The variable panel shows that $\mathcal F_{\mathrm{full}}$ uses more variables on all displayed instances. In contrast, the clause panel shows that $\mathcal F_{\mathrm{pair}}$ incurs substantially larger clause counts on many larger Harwell--Boeing instances, while some smaller instances remain above the diagonal.

\begin{figure*}[pos=t]
\centering
\begin{minipage}[t]{0.48\textwidth}
\centering
\includegraphics[width=\linewidth]{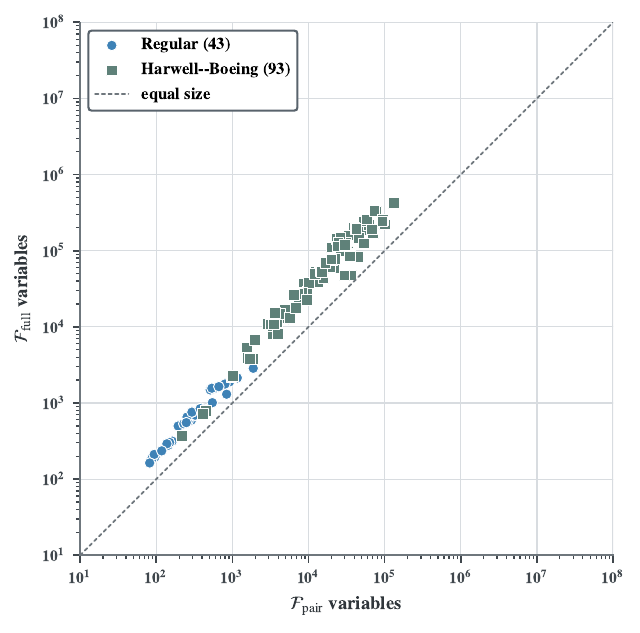}
\par\vspace{0.25em}(a) Variables
\end{minipage}\hfill
\begin{minipage}[t]{0.48\textwidth}
\centering
\includegraphics[width=\linewidth]{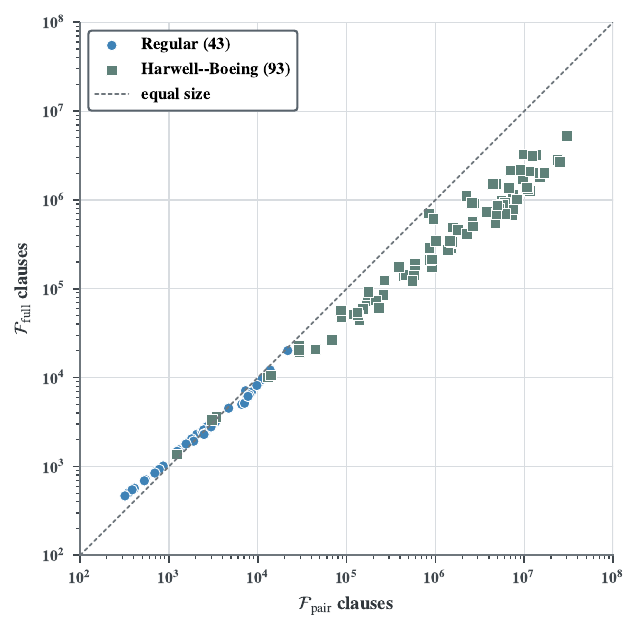}
\par\vspace{0.25em}(b) Clauses
\end{minipage}
\caption{Comparison of variable and clause counts between $\mathcal F_{\mathrm{full}}$ and $\mathcal F_{\mathrm{pair}}$ for the Regular and Harwell--Boeing sets on the $\lceil\sqrt n\rceil\times\lceil\sqrt n\rceil$ host grid. The counts exclude $\Delta_K$. Points below the diagonal indicate smaller counts for $\mathcal F_{\mathrm{full}}$.}
\label{fig:ablation-initial-size-scatter}
\end{figure*}

\section{Detailed Benchmark Tables}\label{app:detailed-results}

The detailed comparison rows from the benchmark files used in Section~\ref{sec:experiments} are reported in Tables~\ref{tab:standard-grid-combined-detail} and~\ref{tab:auxiliary-host-detail}. In Table~\ref{tab:auxiliary-host-detail}, dashes in the $n\times n$ block mark instances for which this additional host-grid evaluation could not be run under the available hardware resources.

\begingroup
\scriptsize
\setlength{\tabcolsep}{1.15pt}
\begin{longtable}{@{}lrrcccccccccc@{}}
\caption{Detailed comparison on the $\lceil\sqrt n\rceil \times \lceil\sqrt n\rceil$ host grid for Regular and Harwell--Boeing instances. Published columns report M3, BVNS, IG, RVNS, DRSA, and RLTS bandwidths. Local columns report $\mathcal F_{\mathrm{full}}$, Gurobi, CPLEX MIP, and CPLEX CP bandwidths with solving times in seconds. Boldface marks the three certified bandwidths below every previously published value.}
\label{tab:standard-grid-combined-detail}\\
\toprule
Instance & $n$ & $|E|$ & M3 & BVNS & IG & RVNS & DRSA & RLTS & $\mathcal{F}_{\mathrm{full}}$ & Gurobi & CPLEX MIP & CPLEX CP \\
\midrule
\endfirsthead
\caption[]{Detailed comparison on the $\lceil\sqrt n\rceil \times \lceil\sqrt n\rceil$ host grid for Regular and Harwell--Boeing instances (continued).}\\
\toprule
Instance & $n$ & $|E|$ & M3 & BVNS & IG & RVNS & DRSA & RLTS & $\mathcal{F}_{\mathrm{full}}$ & Gurobi & CPLEX MIP & CPLEX CP \\
\midrule
\endhead
\bottomrule
\endfoot
\multicolumn{13}{@{}l}{\textit{Regular}}\\
\texttt{wheel5} & 5 & 8 & 2$^{*}$ & 2 & 2 & -- & -- & 2 & 2* (0.03) & 2* (0.07) & 2* (1.23) & 2* (1.19) \\
\texttt{bipartite3x3} & 6 & 9 & 2$^{*}$ & 2 & 2 & -- & -- & 2 & 2* (0.04) & 2* (0.07) & 2* (1.19) & 2* (1.20) \\
\texttt{p2xc3} & 6 & 9 & 2$^{*}$ & 2 & 2 & -- & -- & 2 & 2* (0.03) & 2* (0.07) & 2* (1.21) & 2* (1.20) \\
\texttt{p2xp3} & 6 & 7 & 1$^{*}$ & 2 & 1 & -- & -- & 1 & 1* (0.04) & 1* (0.06) & 1* (1.25) & 1* (1.19) \\
\texttt{bipartite3x4} & 7 & 12 & 3$^{*}$ & 3 & 3 & -- & -- & 3 & 3* (0.05) & 3* (1.17) & 3* (0.09) & 3* (1.22) \\
\texttt{tree2x2} & 7 & 6 & 1$^{*}$ & 1 & 1 & -- & -- & 1 & 1* (0.04) & 1* (0.07) & 1* (1.20) & 1* (1.20) \\
\texttt{wheel7} & 7 & 12 & 2$^{*}$ & 2 & 2 & -- & -- & 2 & 2* (0.04) & 2* (0.07) & 2* (1.21) & 2* (1.21) \\
\texttt{bipartite4x4} & 8 & 16 & 3$^{*}$ & 3 & 3 & -- & -- & 3 & 3* (0.05) & 3* (1.20) & 3* (1.19) & 3* (1.24) \\
\texttt{c3xc3} & 9 & 18 & 2$^{*}$ & 2 & 2 & -- & -- & 2 & 2* (0.04) & 2* (0.08) & 2* (0.07) & 2* (1.21) \\
\texttt{p3xc3} & 9 & 15 & 2$^{*}$ & 2 & 2 & -- & -- & 2 & 2* (0.03) & 2* (0.08) & 2* (1.25) & 2* (1.22) \\
\texttt{p3xp3} & 9 & 12 & 1$^{*}$ & 2 & 1 & -- & -- & 1 & 1* (0.03) & 1* (0.07) & 1* (1.21) & 1* (1.20) \\
\texttt{bipartite5x5} & 10 & 25 & 3$^{*}$ & 3 & 3 & -- & -- & 3 & 3* (0.07) & 3* (1.31) & 3* (1.27) & 3* (1.29) \\
\texttt{cycle10} & 10 & 10 & 1$^{*}$ & 1 & 1 & -- & -- & 1 & 1* (0.04) & 1* (0.07) & 1* (0.07) & 1* (1.20) \\
\texttt{cyclePow10-10} & 10 & 45 & 4$^{*}$ & 4 & 4 & -- & -- & 4 & 4* (2.37) & 4* (2.67) & 4* (3.57) & 4* (8.84) \\
\texttt{cyclePow10-2} & 10 & 20 & 2$^{*}$ & 2 & 2 & -- & -- & 2 & 2* (0.04) & 2* (1.17) & 2* (0.08) & 2* (1.20) \\
\texttt{path10} & 10 & 9 & 1$^{*}$ & 1 & 1 & -- & -- & 1 & 1* (0.04) & 1* (0.07) & 1* (1.24) & 1* (1.21) \\
\texttt{petersen} & 10 & 15 & 2$^{*}$ & 2 & 2 & -- & -- & 2 & 2* (0.04) & 2* (1.22) & 2* (1.50) & 2* (1.22) \\
\texttt{wheel10} & 10 & 18 & 2$^{*}$ & 2 & 2 & -- & -- & 2 & 2* (0.05) & 2* (0.09) & 2* (1.20) & 2* (1.21) \\
\texttt{c3xc4} & 12 & 24 & 2$^{*}$ & 2 & 2 & -- & -- & 2 & 2* (0.04) & 2* (1.26) & 2* (0.11) & 2* (1.21) \\
\texttt{c3xk4} & 12 & 30 & 3$^{*}$ & 3 & 3 & -- & -- & 3 & 3* (0.10) & 3* (2.02) & 3* (1.36) & 3* (1.37) \\
\texttt{k3xk4} & 12 & 30 & 3$^{*}$ & 3 & 3 & -- & -- & 3 & 3* (1.16) & 3* (1.98) & 3* (1.32) & 3* (1.38) \\
\texttt{p3xk4} & 12 & 26 & 2$^{*}$ & 2 & 2 & -- & -- & 2 & 2* (0.04) & 2* (1.53) & 2* (1.58) & 2* (1.22) \\
\texttt{tree2x3} & 13 & 12 & 2$^{*}$ & 2 & 2 & -- & -- & 2 & 2* (0.04) & 2* (1.22) & 2* (1.42) & 2* (1.22) \\
\texttt{bipartite7x8} & 15 & 56 & 4$^{*}$ & 4 & 4 & -- & -- & 4 & 4* (22.01) & 4* (3.01) & 4* (1.64) & 4* (450.14) \\
\texttt{cycle15} & 15 & 15 & 2$^{*}$ & 2 & 2 & -- & -- & 2 & 2* (1.17) & 2* (1.67) & 2* (1.34) & 2* (1.30) \\
\texttt{cyclePow15-10} & 15 & 105 & -- & 6 & 6 & -- & -- & 6 & 6* (108.85) & 6* (640.33) & 6* (27.41) & 6* (10.67) \\
\texttt{cyclePow15-2} & 15 & 30 & 2$^{*}$ & 2 & 2 & -- & -- & 2 & 2* (0.04) & 2* (1.38) & 2* (0.11) & 2* (1.23) \\
\texttt{path15} & 15 & 14 & 1$^{*}$ & 2 & 1 & -- & -- & 1 & 1* (0.04) & 1* (0.10) & 1* (0.09) & 1* (1.23) \\
\texttt{tree3x2} & 15 & 14 & 2$^{*}$ & 2 & 2 & -- & -- & 2 & 2* (0.04) & 2* (1.34) & 2* (1.51) & 2* (1.23) \\
\texttt{wheel15} & 15 & 28 & 3$^{*}$ & 3 & 3 & -- & -- & 3 & 3* (1.48) & 3* (0.10) & 3* (1.62) & 3* (4.33) \\
\texttt{bipartite10x10} & 20 & 100 & -- & 5 & 5 & -- & -- & 5 & 5 (FEAS) & 5* (45.15) & 5* (92.25) & 5* (20.18) \\
\texttt{c4xc5} & 20 & 40 & 2$^{*}$ & 3 & 2 & -- & -- & 2 & 2* (0.06) & 2* (4.99) & 2* (1.66) & 2* (1.28) \\
\texttt{c4xk5} & 20 & 60 & 3$^{*}$ & 3 & 3 & -- & -- & 3 & 3* (1.26) & 3* (6.77) & 3* (77.34) & 3* (4.64) \\
\texttt{cycle20} & 20 & 20 & 1$^{*}$ & 1 & 1 & -- & -- & 1 & 1* (0.04) & 1* (1.23) & 1* (1.21) & 1* (1.21) \\
\texttt{cyclePow20-10} & 20 & 190 & -- & 6 & 6 & -- & -- & 6 & 6 (FEAS) & 6 (FEAS) & 6 (FEAS) & 6 (FEAS) \\
\texttt{cyclePow20-2} & 20 & 40 & 2$^{*}$ & 2 & 2 & -- & -- & 2 & 2* (0.05) & 2* (1.74) & 2* (1.62) & 2* (1.27) \\
\texttt{k4xk5} & 20 & 70 & 4$^{*}$ & 4 & 4 & -- & -- & 4 & 4* (2312.08) & 4 (FEAS) & 4 (FEAS) & 4* (2754.32) \\
\texttt{p4xc5} & 20 & 35 & 2$^{*}$ & 3 & 2 & -- & -- & 2 & 2* (0.06) & 2* (2.50) & 2* (2.04) & 2* (1.25) \\
\texttt{p4xk5} & 20 & 55 & 3$^{*}$ & 3 & 3 & -- & -- & 3 & 3* (1.35) & 3* (14.70) & 3* (28.16) & 3* (2.49) \\
\texttt{p4xp5} & 20 & 31 & 1$^{*}$ & 2 & 1 & -- & -- & 1 & 1* (0.05) & 1* (2.55) & 1* (1.23) & 1* (1.22) \\
\texttt{path20} & 20 & 19 & 1$^{*}$ & 1 & 1 & -- & -- & 1 & 1* (0.05) & 1* (1.26) & 1* (1.22) & 1* (1.22) \\
\texttt{wheel20} & 20 & 38 & -- & 3 & 3 & -- & -- & 3 & 3* (5.97) & 3* (2.69) & 3* (1.53) & 3* (71.87) \\
\texttt{tree2x4} & 21 & 20 & 2$^{*}$ & 2 & 2 & -- & -- & 2 & 2* (0.06) & 2* (1.44) & 2* (1.32) & 2* (1.27) \\
\addlinespace
\multicolumn{13}{@{}l}{\textit{Harwell--Boeing}}\\
\texttt{jgl009} & 9 & 50 & -- & -- & 3 & -- & 3 & 3 & 3* (0.04) & 3* (1.65) & 3* (0.07) & 3* (1.25) \\
\texttt{rgg010} & 10 & 76 & -- & -- & -- & -- & 4 & -- & 4* (2.33) & 4* (2.37) & 4* (3.58) & 4* (8.53) \\
\texttt{jgl011} & 11 & 76 & -- & -- & 4 & -- & 4 & 4 & 4* (2.02) & 4* (2.41) & 4* (4.03) & 4* (8.12) \\
\texttt{can\_\_\_24} & 24 & 92 & -- & -- & 3 & -- & 3 & 3 & 3* (1.36) & 3* (14.83) & 3* (2.77) & 3* (2.66) \\
\texttt{lap\_25} & 25 & 97 & -- & -- & -- & -- & 2 & -- & 2* (1.17) & 2* (1.37) & 2* (1.78) & 2* (1.27) \\
\texttt{fidap005} & 27 & 279 & -- & -- & -- & -- & 4 & -- & 4* (15.84) & 4* (553.90) & 4* (675.42) & 4* (52.33) \\
\texttt{pores\_1} & 30 & 180 & -- & -- & 3 & -- & 3 & 3 & 3* (2.06) & 3* (18.25) & 3* (491.58) & 3* (3.63) \\
\texttt{ibm32} & 32 & 126 & -- & -- & 4 & -- & 4 & 4 & 4 (FEAS) & 4 (FEAS) & 4 (FEAS) & 4 (FEAS) \\
\texttt{bcspwr01} & 39 & 85 & -- & -- & 2 & -- & 2 & 2 & 2* (1.17) & 2* (13.83) & 2 (FEAS) & 2* (1.36) \\
\texttt{fidapm05} & 42 & 520 & -- & -- & -- & -- & 4 & -- & 4* (109.25) & 4 (FEAS) & 5 (FEAS) & 4* (710.15) \\
\texttt{bcsstk01} & 48 & 224 & 6 & 5 & 4 & 5 & 4 & 4 & 4 (FEAS) & 5 (FEAS) & 6 (FEAS) & 5 (FEAS) \\
\texttt{bcspwr02} & 49 & 108 & -- & -- & 2 & -- & 2 & 2 & 2* (2.14) & 2* (39.37) & 3 (FEAS) & 2* (3.78) \\
\texttt{curtis54} & 54 & 291 & -- & -- & 3 & -- & 3 & 3 & 3* (3.15) & 3 (FEAS) & 4 (FEAS) & 3* (20.18) \\
\texttt{will57} & 57 & 281 & -- & -- & 3 & -- & 3 & 3 & 3* (3.23) & 3 (FEAS) & 4 (FEAS) & 3* (83.75) \\
\texttt{dwt\_\_\_59} & 59 & 163 & -- & -- & 3 & -- & 2 & 2 & 2* (2.87) & 3 (FEAS) & 4 (FEAS) & 2* (2.07) \\
\texttt{impcol\_b} & 59 & 312 & -- & -- & 5 & -- & 5 & 5 & 5 (FEAS) & 5 (FEAS) & 6 (FEAS) & 5 (FEAS) \\
\texttt{can\_\_\_61} & 61 & 309 & -- & -- & 4 & -- & 4 & 4 & 4 (FEAS) & 4 (FEAS) & 6 (FEAS) & 4 (FEAS) \\
\texttt{bfw62a} & 62 & 450 & -- & -- & -- & -- & 4 & -- & \textbf{3* (584.02)} & 4 (FEAS) & 5 (FEAS) & 4 (FEAS) \\
\texttt{bfw62b} & 62 & 342 & -- & -- & -- & -- & 3 & -- & 3* (8.77) & 3* (75.08) & 4 (FEAS) & 3* (44.22) \\
\texttt{can\_\_\_62} & 62 & 140 & 2$^{*}$ & 2 & 2 & 3 & 2 & 2 & 2* (3.24) & 2 (FEAS) & 3 (FEAS) & 2* (3.45) \\
\texttt{bcsstk02} & 66 & 2211 & -- & -- & -- & -- & 12 & -- & 12 (FEAS) & 12 (FEAS) & 12 (FEAS) & 12 (FEAS) \\
\texttt{dwt\_\_\_66} & 66 & 193 & -- & -- & 2 & -- & 2 & 2 & 2* (2.56) & 2 (FEAS) & 3 (FEAS) & 2* (2.43) \\
\texttt{west0067} & 67 & 294 & -- & -- & 5 & -- & 5 & 5 & 5 (FEAS) & 6 (FEAS) & 7 (FEAS) & 6 (FEAS) \\
\texttt{dwt\_\_\_72} & 72 & 147 & -- & -- & -- & -- & 2 & -- & 2* (5.75) & 2 (FEAS) & 3 (FEAS) & 2* (186.93) \\
\texttt{can\_\_\_73} & 73 & 225 & -- & -- & 4 & -- & 3 & 3 & 3* (112.13) & 4 (FEAS) & 6 (FEAS) & 4 (FEAS) \\
\texttt{steam3} & 80 & 928 & -- & -- & 4 & -- & 4 & 4 & 4* (911.66) & 5 (FEAS) & 6 (FEAS) & 5 (FEAS) \\
\texttt{ash85} & 85 & 304 & -- & -- & 3 & -- & 2 & 3 & 2* (18.92) & 3 (FEAS) & 6 (FEAS) & 3 (FEAS) \\
\texttt{dwt\_\_\_87} & 87 & 314 & -- & -- & 3 & -- & 3 & 3 & 3* (9.19) & 4 (FEAS) & 6 (FEAS) & 3* (144.45) \\
\texttt{tols90} & 90 & 1746 & -- & -- & -- & -- & 10 & -- & 11 (FEAS) & 10 (FEAS) & 11 (FEAS) & 10 (FEAS) \\
\texttt{can\_\_\_96} & 96 & 432 & -- & -- & 4 & -- & 3 & 3 & 3* (28.30) & 4 (FEAS) & 7 (FEAS) & 4 (FEAS) \\
\texttt{nos4} & 100 & 347 & 8 & 4 & 4 & 4 & 3 & 3 & 3* (131.90) & 4 (FEAS) & 4 (FEAS) & 4 (FEAS) \\
\texttt{olm100} & 100 & 396 & -- & -- & -- & -- & 2 & -- & 2* (555.15) & 3 (FEAS) & 5 (FEAS) & 3 (FEAS) \\
\texttt{tub100} & 100 & 396 & -- & -- & -- & -- & 2 & -- & 2* (164.51) & 3 (FEAS) & 5 (FEAS) & 2* (76.38) \\
\texttt{ck104} & 104 & 992 & -- & -- & -- & -- & 3 & -- & 3* (122.95) & 5 (FEAS) & 6 (FEAS) & 4 (FEAS) \\
\texttt{bcsstk03} & 112 & 376 & -- & -- & -- & -- & 3 & -- & \textbf{2* (3190.90)} & 3 (FEAS) & 5 (FEAS) & 3 (FEAS) \\
\texttt{gent113} & 113 & 655 & -- & -- & 5 & -- & 5 & 5 & 6 (FEAS) & 6 (FEAS) & 10 (FEAS) & 6 (FEAS) \\
\texttt{gre\_\_115} & 115 & 421 & -- & -- & -- & -- & 4 & -- & 4 (FEAS) & 4 (FEAS) & 8 (FEAS) & 5 (FEAS) \\
\texttt{bcspwr03} & 118 & 297 & 15 & 4 & 3 & 4 & 2 & 3 & 2* (209.45) & 3 (FEAS) & 6 (FEAS) & 4 (FEAS) \\
\texttt{arc130} & 130 & 1282 & -- & -- & 10 & -- & 9 & 9 & 10 (FEAS) & 10 (FEAS) & 11 (FEAS) & 9 (FEAS) \\
\texttt{lns\_\_131} & 131 & 536 & -- & -- & -- & -- & 3 & -- & 3* (54.92) & 4 (FEAS) & 8 (FEAS) & 3 (FEAS) \\
\texttt{lnsp\_131} & 131 & 536 & -- & -- & -- & -- & 3 & -- & 3* (36.66) & 4 (FEAS) & 8 (FEAS) & 3 (FEAS) \\
\texttt{bcsstk04} & 132 & 1890 & 16 & 9 & 8 & 8 & 8 & 8 & 10 (FEAS) & 11 (FEAS) & 13 (FEAS) & 11 (FEAS) \\
\texttt{west0132} & 132 & 414 & -- & -- & -- & -- & 5 & -- & 6 (FEAS) & 7 (FEAS) & 9 (FEAS) & 7 (FEAS) \\
\texttt{rw136} & 136 & 479 & -- & -- & -- & -- & 3 & -- & 3* (441.74) & 5 (FEAS) & 8 (FEAS) & 4 (FEAS) \\
\texttt{impcol\_c} & 137 & 411 & -- & -- & -- & -- & 5 & -- & 5 (FEAS) & 7 (FEAS) & 9 (FEAS) & 6 (FEAS) \\
\texttt{bcsstk22} & 138 & 417 & 8 & 4 & 3 & 4 & 3 & 3 & \textbf{2* (61.01)} & 4 (FEAS) & 7 (FEAS) & 3 (FEAS) \\
\texttt{can\_\_144} & 144 & 720 & 16 & 4 & 4 & 4 & 3 & 4 & 4 (FEAS) & 7 (FEAS) & 9 (FEAS) & 6 (FEAS) \\
\texttt{lund\_a} & 147 & 1298 & -- & -- & 5 & -- & 5 & 5 & 5 (FEAS) & 7 (FEAS) & 12 (FEAS) & 6 (FEAS) \\
\texttt{lund\_b} & 147 & 1294 & -- & -- & 5 & -- & 5 & 5 & 5 (FEAS) & 7 (FEAS) & 12 (FEAS) & 6 (FEAS) \\
\texttt{bcsstk05} & 153 & 1288 & 17 & 7 & 6 & 5 & 6 & 6 & 7 (FEAS) & 9 (FEAS) & 11 (FEAS) & 8 (FEAS) \\
\texttt{west0156} & 156 & 371 & -- & -- & -- & -- & 5 & -- & 6 (FEAS) & 8 (FEAS) & 10 (FEAS) & 8 (FEAS) \\
\texttt{can\_\_161} & 161 & 769 & 9 & 6 & 4 & 5 & 4 & 4 & 4 (FEAS) & 8 (FEAS) & 10 (FEAS) & 6 (FEAS) \\
\texttt{dwt\_\_162} & 162 & 672 & -- & -- & -- & -- & 3 & -- & 3* (322.01) & 7 (FEAS) & 9 (FEAS) & 4 (FEAS) \\
\texttt{lop163} & 163 & 935 & -- & -- & -- & -- & 4 & -- & 5 (FEAS) & 6 (FEAS) & 12 (FEAS) & 6 (FEAS) \\
\texttt{west0167} & 167 & 507 & -- & -- & -- & -- & 5 & -- & 6 (FEAS) & 7 (FEAS) & 12 (FEAS) & 8 (FEAS) \\
\texttt{mcca} & 180 & 2659 & -- & -- & 8 & -- & 8 & 8 & 8 (FEAS) & 11 (FEAS) & 13 (FEAS) & 9 (FEAS) \\
\texttt{fs\_183\_1} & 183 & 1069 & -- & -- & -- & -- & 8 & -- & 9 (FEAS) & 10 (FEAS) & 11 (FEAS) & 9 (FEAS) \\
\texttt{gre\_\_185} & 185 & 1005 & -- & -- & -- & -- & 3 & -- & 4 (FEAS) & 10 (FEAS) & 14 (FEAS) & 7 (FEAS) \\
\texttt{can\_\_187} & 187 & 839 & -- & -- & -- & -- & 3 & -- & 3* (1163.91) & 8 (FEAS) & 14 (FEAS) & 6 (FEAS) \\
\texttt{dwt\_\_193} & 193 & 1843 & -- & -- & 6 & -- & 6 & 6 & 8 (FEAS) & 15 (FEAS) & 20 (FEAS) & 11 (FEAS) \\
\texttt{dwt\_\_198} & 198 & 795 & 10 & 3 & 3 & 3 & 3 & 3 & 3* (103.43) & 8 (FEAS) & 11 (FEAS) & 4 (FEAS) \\
\texttt{will199} & 199 & 701 & -- & -- & -- & -- & 7 & -- & 9 (FEAS) & 13 (FEAS) & 17 (FEAS) & 11 (FEAS) \\
\texttt{bwm200} & 200 & 796 & -- & -- & -- & -- & 2 & -- & 2* (323.52) & 6 (FEAS) & 8 (FEAS) & 2* (123.14) \\
\texttt{rdb200} & 200 & 1120 & -- & -- & -- & -- & 2 & -- & 2* (203.08) & 6 (FEAS) & 13 (FEAS) & 3 (FEAS) \\
\texttt{impcol\_a} & 207 & 572 & -- & -- & -- & -- & 5 & -- & 6 (FEAS) & 9 (FEAS) & 15 (FEAS) & 8 (FEAS) \\
\texttt{wm1} & 207 & 2909 & -- & -- & 10 & -- & 9 & 10 & 13 (FEAS) & 20 (FEAS) & 23 (FEAS) & 13 (FEAS) \\
\texttt{wm2} & 207 & 2942 & -- & -- & 10 & -- & 9 & 10 & 13 (FEAS) & 19 (FEAS) & 23 (FEAS) & 13 (FEAS) \\
\texttt{wm3} & 207 & 2948 & -- & -- & 10 & -- & 9 & 10 & 13 (FEAS) & 21 (FEAS) & 24 (FEAS) & 13 (FEAS) \\
\texttt{dwt\_\_209} & 209 & 976 & 10 & 6 & 5 & 5 & 4 & 4 & 5 (FEAS) & 12 (FEAS) & 18 (FEAS) & 8 (FEAS) \\
\texttt{fidap001} & 216 & 4374 & -- & -- & -- & -- & 6 & -- & 8 (FEAS) & 14 (FEAS) & 22 (FEAS) & 11 (FEAS) \\
\texttt{gre\_216a} & 216 & 876 & -- & -- & -- & -- & 4 & -- & 4 (FEAS) & 8 (FEAS) & 17 (FEAS) & 7 (FEAS) \\
\texttt{ash219} & 219 & 438 & -- & -- & -- & -- & 5 & -- & 6 (FEAS) & 10 (FEAS) & 17 (FEAS) & 9 (FEAS) \\
\texttt{dwt\_\_221} & 221 & 925 & 10 & 5 & 4 & 6 & 3 & 4 & 4 (FEAS) & 9 (FEAS) & 20 (FEAS) & 7 (FEAS) \\
\texttt{impcol\_e} & 225 & 1308 & -- & -- & -- & -- & 7 & -- & 8 (FEAS) & 12 (FEAS) & 22 (FEAS) & 11 (FEAS) \\
\texttt{pde225} & 225 & 1065 & -- & -- & -- & -- & 1 & -- & 1* (1701.68) & 8 (FEAS) & 1* (1.98) & 8 (FEAS) \\
\texttt{can\_\_229} & 229 & 1003 & 21 & 5 & 4 & 5 & 4 & 4 & 4 (FEAS) & 11 (FEAS) & 18 (FEAS) & 6 (FEAS) \\
\texttt{dwt\_\_234} & 234 & 534 & 21 & 2 & 3 & 2 & 3 & 3 & 2* (743.33) & 6 (FEAS) & 11 (FEAS) & 3 (FEAS) \\
\texttt{e05r0000} & 236 & 5856 & -- & -- & -- & -- & 6 & -- & 11 (FEAS) & 21 (FEAS) & 26 (FEAS) & 10 (FEAS) \\
\texttt{nos1} & 237 & 627 & 10 & 4 & 3 & 4 & 2 & 3 & 2* (832.54) & 4 (FEAS) & 15 (FEAS) & 3 (FEAS) \\
\texttt{steam1} & 240 & 3762 & -- & -- & 6 & -- & 6 & 6 & 7 (FEAS) & 16 (FEAS) & 23 (FEAS) & 9 (FEAS) \\
\texttt{dwt\_\_245} & 245 & 853 & 11 & 6 & 4 & 5 & 3 & 4 & 4 (FEAS) & 10 (FEAS) & 17 (FEAS) & 6 (FEAS) \\
\texttt{can\_\_256} & 256 & 1586 & -- & -- & -- & -- & 7 & -- & 9 (FEAS) & 15 (FEAS) & 21 (FEAS) & 13 (FEAS) \\
\texttt{nnc261} & 261 & 1500 & -- & -- & -- & -- & 4 & -- & 5 (FEAS) & 13 (FEAS) & 19 (FEAS) & 7 (FEAS) \\
\texttt{lshp\_265} & 265 & 1009 & 11 & 6 & 3 & 6 & 2 & 3 & 3 (FEAS) & 15 (FEAS) & 20 (FEAS) & 7 (FEAS) \\
\texttt{can\_\_268} & 268 & 1675 & -- & -- & -- & -- & 6 & -- & 8 (FEAS) & 17 (FEAS) & 23 (FEAS) & 10 (FEAS) \\
\texttt{bcspwr04} & 274 & 943 & 12 & 6 & 4 & 6 & 3 & 4 & 4 (FEAS) & 9 (FEAS) & 20 (FEAS) & 6 (FEAS) \\
\texttt{ash292} & 292 & 1250 & 12 & 6 & 4 & 5 & 3 & 4 & 3* (1197.86) & 10 (FEAS) & 25 (FEAS) & 6 (FEAS) \\
\texttt{can\_\_292} & 292 & 1416 & 24 & 7 & 5 & 8 & 5 & 5 & 6 (FEAS) & 14 (FEAS) & 24 (FEAS) & 8 (FEAS) \\
\texttt{dwt\_\_307} & 307 & 1415 & 12 & 6 & 4 & 6 & 4 & 4 & 5 (FEAS) & 13 (FEAS) & 28 (FEAS) & 9 (FEAS) \\
\texttt{dwt\_\_310} & 310 & 1379 & 12 & 5 & 4 & 6 & 3 & 4 & 4 (FEAS) & 13 (FEAS) & 29 (FEAS) & 8 (FEAS) \\
\texttt{dwt\_\_346} & 346 & 1786 & -- & -- & -- & -- & 5 & -- & 6 (FEAS) & 23 (FEAS) & 30 (FEAS) & 8 (FEAS) \\
\texttt{dwt\_\_361} & 361 & 1657 & 13 & 7 & 4 & 6 & 3 & 4 & 11 (FEAS) & 19 (FEAS) & 33 (FEAS) & 16 (FEAS) \\
\texttt{plat362} & 362 & 3074 & 13 & 8 & 6 & 6 & 4 & 5 & 7 (FEAS) & 26 (FEAS) & 34 (FEAS) & 9 (FEAS) \\
\end{longtable}
\endgroup

\begingroup
\scriptsize
\setlength{\tabcolsep}{1.0pt}
\begin{longtable}{@{}lrrcccccccc@{}}
\caption{Detailed auxiliary-host results of $\mathcal F_{\mathrm{full}}$, Gurobi, CPLEX MIP, and CPLEX CP under the $2\times\lceil n/2\rceil$ and $n\times n$ host grids for Regular and Harwell--Boeing instances.}
\label{tab:auxiliary-host-detail}\\
\toprule
Instance & $n$ & $|E|$ & \multicolumn{4}{c}{$2\times\lceil n/2\rceil$} & \multicolumn{4}{c}{$n\times n$} \\
\cmidrule(lr){4-7}\cmidrule(l){8-11}
 & & & $\mathcal{F}_{\mathrm{full}}$ & Gurobi & CPLEX MIP & CPLEX CP & $\mathcal{F}_{\mathrm{full}}$ & Gurobi & CPLEX MIP & CPLEX CP \\
\midrule
\endfirsthead
\caption[]{Detailed auxiliary-host results under the $2\times\lceil n/2\rceil$ and $n\times n$ host grids (continued).}\\
\toprule
Instance & $n$ & $|E|$ & \multicolumn{4}{c}{$2\times\lceil n/2\rceil$} & \multicolumn{4}{c}{$n\times n$} \\
\cmidrule(lr){4-7}\cmidrule(l){8-11}
 & & & $\mathcal{F}_{\mathrm{full}}$ & Gurobi & CPLEX MIP & CPLEX CP & $\mathcal{F}_{\mathrm{full}}$ & Gurobi & CPLEX MIP & CPLEX CP \\
\midrule
\endhead
\bottomrule
\endfoot
\multicolumn{11}{@{}l}{\textit{Regular}}\\
\texttt{wheel5} & 5 & 8 & 2* (0.06) & 2* (0.07) & 2* (34.35) & 2* (0.07) & 2* (0.06) & 2* (0.08) & 2* (0.07) & 2* (1.29) \\
\texttt{bipartite3x3} & 6 & 9 & 2* (0.11) & 2* (0.08) & 2* (0.05) & 2* (1.30) & 2* (0.08) & 2* (1.20) & 2* (2.39) & 2* (1.38) \\
\texttt{p2xc3} & 6 & 9 & 2* (0.07) & 2* (0.08) & 2* (15.35) & 2* (0.07) & 2* (0.06) & 2* (1.23) & 2* (1.28) & 2* (1.28) \\
\texttt{p2xp3} & 6 & 7 & 1* (0.07) & 1* (0.08) & 1* (16.35) & 1* (0.07) & 1* (0.06) & 1* (0.07) & 1* (0.07) & 1* (1.27) \\
\texttt{bipartite3x4} & 7 & 12 & 3* (0.06) & 3* (0.11) & 3* (0.08) & 3* (1.31) & 3* (1.54) & 3* (1.75) & 3* (1.77) & 3* (3.03) \\
\texttt{tree2x2} & 7 & 6 & 2* (0.06) & 2* (0.08) & 2* (27.35) & 2* (0.07) & 1* (0.06) & 1* (0.08) & 1* (0.06) & 1* (1.26) \\
\texttt{wheel7} & 7 & 12 & 2* (0.06) & 2* (0.07) & 2* (35.35) & 2* (0.07) & 2* (0.05) & 2* (1.27) & 2* (1.35) & 2* (1.23) \\
\texttt{bipartite4x4} & 8 & 16 & 3* (0.06) & 3* (0.10) & 3* (0.06) & 3* (1.30) & 3* (1.59) & 3* (2.24) & 3* (3.67) & 3* (3.47) \\
\texttt{c3xc3} & 9 & 18 & 3* (0.07) & 3* (1.25) & 3* (0.10) & 3* (1.40) & 2* (0.08) & 2* (1.93) & 2* (2.20) & 2* (1.33) \\
\texttt{p3xc3} & 9 & 15 & 2* (0.06) & 2* (0.08) & 2* (17.35) & 2* (0.08) & 2* (0.07) & 2* (1.37) & 2* (1.52) & 2* (1.31) \\
\texttt{p3xp3} & 9 & 12 & 2* (0.06) & 2* (0.08) & 2* (19.35) & 2* (0.08) & 1* (0.07) & 1* (0.11) & 1* (1.20) & 1* (1.26) \\
\texttt{bipartite5x5} & 10 & 25 & 4* (1.21) & 4* (1.38) & 4* (1.20) & 4* (3.97) & 3* (2.02) & 3* (6.21) & 3* (21.44) & 3* (6.78) \\
\texttt{cycle10} & 10 & 10 & 1* (0.08) & 1* (0.08) & 1* (4.35) & 1* (0.08) & 1* (0.07) & 1* (1.30) & 1* (1.37) & 1* (1.26) \\
\texttt{cyclePow10-10} & 10 & 45 & 5* (0.06) & 5* (1.71) & 5* (7.35) & 5* (3.91) & 4* (98.79) & 4* (212.68) & 4* (649.24) & 4* (426.16) \\
\texttt{cyclePow10-2} & 10 & 20 & 2* (0.06) & 2* (0.09) & 2* (8.35) & 2* (0.09) & 2* (0.08) & 2* (2.28) & 2* (3.34) & 2* (1.30) \\
\texttt{path10} & 10 & 9 & 1* (0.06) & 1* (0.08) & 1* (23.35) & 1* (0.08) & 1* (0.08) & 1* (1.20) & 1* (1.23) & 1* (1.26) \\
\texttt{petersen} & 10 & 15 & 3* (0.06) & 3* (1.23) & 3* (26.35) & 3* (1.19) & 2* (0.09) & 2* (1.73) & 2* (1.72) & 2* (1.35) \\
\texttt{wheel10} & 10 & 18 & 3* (0.06) & 3* (0.08) & 3* (31.35) & 3* (0.11) & 2* (0.07) & 2* (1.73) & 2* (1.61) & 2* (1.29) \\
\texttt{c3xc4} & 12 & 24 & 3* (0.08) & 3* (1.27) & 3* (1.22) & 3* (1.59) & 2* (0.11) & 2* (3.17) & 2* (2.46) & 2* (1.34) \\
\texttt{c3xk4} & 12 & 30 & 4* (1.24) & 4* (1.71) & 4* (1.35) & 4* (4.83) & 3* (4.67) & 3* (35.96) & 3* (157.89) & 3* (33.75) \\
\texttt{k3xk4} & 12 & 30 & 4* (1.27) & 4* (1.55) & 4* (13.35) & 4* (4.70) & 3* (4.83) & 3* (36.68) & 3* (237.82) & 3* (53.16) \\
\texttt{p3xk4} & 12 & 26 & 2* (0.07) & 2* (1.20) & 2* (18.35) & 2* (0.09) & 2* (0.09) & 2* (2.40) & 2* (2.02) & 2* (1.34) \\
\texttt{tree2x3} & 13 & 12 & 2* (0.07) & 2* (0.09) & 2* (28.35) & 2* (0.10) & 2* (1.36) & 2* (8.04) & 2* (18.38) & 2* (4.48) \\
\texttt{bipartite7x8} & 15 & 56 & 6* (179.58) & 6* (0.09) & 6* (0.07) & 6* (1.25) & 4 (FEAS) & 4* (928.99) & 4 (FEAS) & 4 (FEAS) \\
\texttt{cycle15} & 15 & 15 & 2* (0.07) & 2* (1.47) & 2* (5.35) & 2* (1.21) & 2* (28.22) & 2 (FEAS) & 2 (FEAS) & 2* (278.46) \\
\texttt{cyclePow15-10} & 15 & 105 & 8* (0.06) & -- (TO) & 8* (9.35) & -- (TO) & 5 (FEAS) & 5 (FEAS) & 5 (FEAS) & 5 (FEAS) \\
\texttt{cyclePow15-2} & 15 & 30 & 3* (0.10) & 3* (2.33) & 3* (10.35) & 3* (1.81) & 2* (1.21) & 2* (16.24) & 2* (6.39) & 2* (1.40) \\
\texttt{path15} & 15 & 14 & 1* (0.07) & 1* (0.12) & 1* (24.35) & 1* (0.09) & 1* (0.09) & 1* (2.57) & 1* (1.70) & 1* (1.30) \\
\texttt{tree3x2} & 15 & 14 & 2* (0.07) & 2* (1.21) & 2* (30.35) & 2* (1.18) & 1* (0.09) & 1* (2.12) & 1* (1.75) & 1* (1.25) \\
\texttt{wheel15} & 15 & 28 & 4* (0.07) & 4* (0.11) & 4* (32.35) & 4* (76.10) & 3* (126.04) & 3 (FEAS) & 3* (42.28) & 3 (FEAS) \\
\texttt{bipartite10x10} & 20 & 100 & 8* (0.17) & 8* (1.21) & 8* (0.08) & 8* (1.33) & 5 (FEAS) & 5 (FEAS) & 5 (FEAS) & 5 (FEAS) \\
\texttt{c4xc5} & 20 & 40 & 4* (3.52) & 4* (8.86) & 4* (2.35) & 4* (1906.58) & 2* (1.52) & 2* (137.01) & 2* (248.87) & 2* (1.97) \\
\texttt{c4xk5} & 20 & 60 & 5* (44.65) & 5* (22.03) & 5* (3.35) & -- (TO) & 3* (7.64) & 3 (FEAS) & 3 (FEAS) & 3* (167.45) \\
\texttt{cycle20} & 20 & 20 & 1* (0.08) & 1* (1.22) & 1* (6.35) & 1* (0.10) & 1* (1.21) & 1* (9.22) & 1* (5.09) & 1* (1.35) \\
\texttt{cyclePow20-10} & 20 & 190 & 10* (0.08) & 10* (29.34) & 10* (11.35) & -- (TO) & 6 (FEAS) & 6 (FEAS) & 6 (FEAS) & 6 (FEAS) \\
\texttt{cyclePow20-2} & 20 & 40 & 2* (0.09) & 2* (1.52) & 2* (12.35) & 2* (1.24) & 2* (1.27) & 2* (27.87) & 2* (17.54) & 2* (1.92) \\
\texttt{k4xk5} & 20 & 70 & 6 (FEAS) & 6* (179.53) & 6* (14.35) & -- (TO) & 4 (FEAS) & 4 (FEAS) & 4 (FEAS) & 4 (FEAS) \\
\texttt{p4xc5} & 20 & 35 & 3* (1.29) & 3* (2.30) & 3* (20.35) & 3* (2.74) & 2* (1.54) & 2* (89.66) & 2* (482.54) & 2* (2.56) \\
\texttt{p4xk5} & 20 & 55 & 3* (0.09) & 3* (1.85) & 3* (21.35) & 3* (1.29) & 3* (7.14) & 3 (FEAS) & -- (TO) & 3* (43.01) \\
\texttt{p4xp5} & 20 & 31 & 2* (0.09) & 2* (1.69) & 2* (22.35) & 2* (1.22) & 1* (1.22) & 1* (17.88) & 1* (4.61) & 1* (1.36) \\
\texttt{path20} & 20 & 19 & 1* (0.09) & 1* (1.25) & 1* (25.35) & 1* (0.09) & 1* (1.20) & 1* (10.36) & 1* (1.79) & 1* (1.29) \\
\texttt{wheel20} & 20 & 38 & 6* (0.07) & 6* (1.43) & 6* (33.35) & -- (TO) & 3* (285.25) & 3 (FEAS) & 3 (FEAS) & 3 (FEAS) \\
\texttt{tree2x4} & 21 & 20 & 3* (0.07) & 3* (2.49) & 3* (29.35) & 3* (231.48) & 2* (1.56) & 2* (125.16) & 2* (109.53) & 2* (7.11) \\
\addlinespace
\multicolumn{11}{@{}l}{\textit{Harwell--Boeing}}\\
\texttt{jgl009} & 9 & 50 & 4* (0.10) & 4* (1.59) & 4* (1.52) & 4* (1.92) & 3* (1.28) & 3* (6.75) & 3* (3.17) & 3* (1.66) \\
\texttt{rgg010} & 10 & 76 & 5* (0.08) & 5* (1.93) & 5* (1.81) & 5* (12.75) & 4* (79.95) & 4* (114.31) & 4* (594.88) & 4* (197.21) \\
\texttt{jgl011} & 11 & 76 & 5* (1.45) & 5* (2.70) & 5* (1.71) & 5* (9.25) & 4* (127.24) & 4* (349.30) & 4 (FEAS) & 4* (738.78) \\
\texttt{can\_\_\_24} & 24 & 92 & 3* (1.31) & 3* (4.63) & 3* (3.08) & 3* (4.26) & 2* (1.35) & 2* (164.95) & 2 (FEAS) & 2* (1.45) \\
\texttt{lap\_25} & 25 & 97 & 4* (1.21) & 4* (52.32) & 4* (25.70) & 4* (145.02) & 2* (1.39) & 2* (373.56) & 2 (FEAS) & 2* (1.91) \\
\texttt{fidap005} & 27 & 279 & 5* (2.25) & 5* (40.80) & 5* (323.79) & 5* (239.33) & 4* (1482.38) & 4 (FEAS) & 4 (FEAS) & 4 (FEAS) \\
\texttt{pores\_1} & 30 & 180 & 4* (1.89) & 4* (70.24) & 4 (FEAS) & 4* (446.92) & 3* (13.54) & 3 (FEAS) & 4 (FEAS) & 3* (131.12) \\
\texttt{ibm32} & 32 & 126 & 6 (FEAS) & 6 (FEAS) & 7 (FEAS) & 6 (FEAS) & 3* (462.15) & 4 (FEAS) & 4 (FEAS) & 3 (FEAS) \\
\texttt{bcspwr01} & 39 & 85 & 3* (3.75) & 3* (178.98) & 3 (FEAS) & 3 (FEAS) & 2* (2.33) & 2 (FEAS) & 2 (FEAS) & 2* (7.11) \\
\texttt{fidapm05} & 42 & 520 & 6* (8.77) & 6* (661.87) & 7 (FEAS) & 6 (FEAS) & 4 (FEAS) & 5 (FEAS) & 5 (FEAS) & 4 (FEAS) \\
\texttt{bcsstk01} & 48 & 224 & 9 (FEAS) & -- (TO) & 9 (FEAS) & 9 (FEAS) & 4 (FEAS) & 5 (FEAS) & 6 (FEAS) & 4 (FEAS) \\
\texttt{bcspwr02} & 49 & 108 & 4* (1037.80) & 4 (FEAS) & 5 (FEAS) & 4 (FEAS) & 2* (4.13) & 2 (FEAS) & 4 (FEAS) & 2* (9.58) \\
\texttt{curtis54} & 54 & 291 & 5 (FEAS) & 5* (1030.38) & 8 (FEAS) & 6 (FEAS) & 3* (41.30) & 4 (FEAS) & 6 (FEAS) & 3 (FEAS) \\
\texttt{will57} & 57 & 281 & 4* (10.97) & -- (TO) & 6 (FEAS) & 4 (FEAS) & 3* (126.32) & 4 (FEAS) & 5 (FEAS) & 3 (FEAS) \\
\texttt{dwt\_\_\_59} & 59 & 163 & 3* (4.16) & 3* (339.85) & 4 (FEAS) & 3 (FEAS) & 2* (8.05) & 3 (FEAS) & 5 (FEAS) & 2* (36.98) \\
\texttt{impcol\_b} & 59 & 312 & -- (TO) & -- (TO) & -- (TO) & -- (TO) & 4 (FEAS) & 6 (FEAS) & 8 (FEAS) & 5 (FEAS) \\
\texttt{can\_\_\_61} & 61 & 309 & 7* (37.33) & 7 (FEAS) & 9 (FEAS) & 7 (FEAS) & 4 (FEAS) & 4 (FEAS) & 7 (FEAS) & 4 (FEAS) \\
\texttt{bfw62a} & 62 & 450 & 7 (FEAS) & 7 (FEAS) & 8 (FEAS) & 7 (FEAS) & 3* (61.27) & 5 (FEAS) & 7 (FEAS) & 4 (FEAS) \\
\texttt{bfw62b} & 62 & 342 & 5* (434.41) & 5 (FEAS) & 7 (FEAS) & 5 (FEAS) & 3* (290.19) & 5 (FEAS) & 6 (FEAS) & 3* (1153.74) \\
\texttt{can\_\_\_62} & 62 & 140 & 3* (36.54) & 3* (1070.51) & 6 (FEAS) & 3 (FEAS) & 2* (11.60) & 4 (FEAS) & 5 (FEAS) & 2* (36.10) \\
\texttt{bcsstk02} & 66 & 2211 & 33* (0.06) & 33 (FEAS) & 33 (FEAS) & 33 (FEAS) & 11 (FEAS) & -- (TO) & -- (TO) & 11 (FEAS) \\
\texttt{dwt\_\_\_66} & 66 & 193 & 2* (17.44) & 2* (622.91) & 4 (FEAS) & 2* (147.26) & 2* (10.88) & 4 (FEAS) & 7 (FEAS) & 2* (19.53) \\
\texttt{west0067} & 67 & 294 & -- (TO) & -- (TO) & -- (TO) & -- (TO) & 5 (FEAS) & 7 (FEAS) & 10 (FEAS) & 5 (FEAS) \\
\texttt{dwt\_\_\_72} & 72 & 147 & 3* (108.82) & 3 (FEAS) & 7 (FEAS) & 3 (FEAS) & 1* (24.21) & 4 (FEAS) & 8 (FEAS) & 1* (6.20) \\
\texttt{can\_\_\_73} & 73 & 225 & 9 (FEAS) & 9 (FEAS) & 11 (FEAS) & 9 (FEAS) & 3 (FEAS) & 5 (FEAS) & 9 (FEAS) & 3 (FEAS) \\
\texttt{steam3} & 80 & 928 & 3* (30.95) & 3* (67.49) & 5 (FEAS) & 3 (FEAS) & 2* (19.36) & 7 (FEAS) & 16 (FEAS) & 2* (221.05) \\
\texttt{ash85} & 85 & 304 & 5 (FEAS) & 6 (FEAS) & 11 (FEAS) & 5 (FEAS) & 2* (37.20) & 6 (FEAS) & 12 (FEAS) & 2* (205.00) \\
\texttt{dwt\_\_\_87} & 87 & 314 & 5 (FEAS) & 7 (FEAS) & 9 (FEAS) & 6 (FEAS) & 3* (757.44) & 6 (FEAS) & 42 (FEAS) & 3 (FEAS) \\
\texttt{tols90} & 90 & 1746 & -- (TO) & -- (TO) & -- (TO) & -- (TO) & 9 (FEAS) & -- (TO) & -- (TO) & 9 (FEAS) \\
\texttt{can\_\_\_96} & 96 & 432 & 10 (FEAS) & 8 (FEAS) & -- (TO) & 9 (FEAS) & 3 (FEAS) & -- (TO) & 14 (FEAS) & 3 (FEAS) \\
\texttt{nos4} & 100 & 347 & 5 (FEAS) & 5 (FEAS) & 14 (FEAS) & 6 (FEAS) & 2* (82.79) & 10 (FEAS) & 11 (FEAS) & 2 (FEAS) \\
\texttt{olm100} & 100 & 396 & 2* (31.35) & 3 (FEAS) & 6 (FEAS) & 4 (FEAS) & 2* (54.22) & 5 (FEAS) & 3 (FEAS) & 2* (312.55) \\
\texttt{tub100} & 100 & 396 & 1* (26.47) & 1* (948.40) & 2 (FEAS) & 4 (FEAS) & 1* (10.36) & 6 (FEAS) & 2 (FEAS) & 1* (3.02) \\
\texttt{ck104} & 104 & 992 & 3* (809.80) & 9 (FEAS) & 7 (FEAS) & 7 (FEAS) & 3* (911.67) & -- (TO) & -- (TO) & 3 (FEAS) \\
\texttt{bcsstk03} & 112 & 376 & 2* (693.32) & 3 (FEAS) & 11 (FEAS) & 5 (FEAS) & 2* (96.65) & -- (TO) & 49 (FEAS) & 2* (244.31) \\
\texttt{gent113} & 113 & 655 & -- (TO) & -- (TO) & -- (TO) & 14 (FEAS) & 6 (FEAS) & -- (TO) & -- (TO) & 6 (FEAS) \\
\texttt{gre\_\_115} & 115 & 421 & 15 (FEAS) & 14 (FEAS) & -- (TO) & 14 (FEAS) & 3 (FEAS) & -- (TO) & 69 (FEAS) & 4 (FEAS) \\
\texttt{bcspwr03} & 118 & 297 & 6 (FEAS) & 7 (FEAS) & 11 (FEAS) & 7 (FEAS) & 2* (160.18) & -- (TO) & -- (TO) & 2 (FEAS) \\
\texttt{arc130} & 130 & 1282 & -- (TO) & 32 (FEAS) & -- (TO) & 32 (FEAS) & 9 (FEAS) & -- (TO) & -- (TO) & 8 (FEAS) \\
\texttt{lns\_\_131} & 131 & 536 & 11 (FEAS) & 11 (FEAS) & -- (TO) & 12 (FEAS) & 3* (2227.35) & -- (TO) & -- (TO) & 3 (FEAS) \\
\texttt{lnsp\_131} & 131 & 536 & 11 (FEAS) & 14 (FEAS) & -- (TO) & 12 (FEAS) & 3* (2280.70) & -- (TO) & -- (TO) & 3 (FEAS) \\
\texttt{bcsstk04} & 132 & 1890 & -- (TO) & -- (TO) & -- (TO) & -- (TO) & 8 (FEAS) & -- (TO) & -- (TO) & 9 (FEAS) \\
\texttt{west0132} & 132 & 414 & -- (TO) & -- (TO) & -- (TO) & -- (TO) & 6 (FEAS) & -- (TO) & -- (TO) & 7 (FEAS) \\
\texttt{rw136} & 136 & 479 & 9 (FEAS) & 9 (FEAS) & -- (TO) & 12 (FEAS) & 2* (349.05) & -- (TO) & -- (TO) & 2 (FEAS) \\
\texttt{impcol\_c} & 137 & 411 & -- (TO) & 16 (FEAS) & -- (TO) & -- (TO) & 5 (FEAS) & -- (TO) & -- (TO) & 5 (FEAS) \\
\texttt{bcsstk22} & 138 & 417 & 6 (FEAS) & 6 (FEAS) & -- (TO) & 8 (FEAS) & 2* (379.35) & -- (TO) & -- (TO) & 3 (FEAS) \\
\texttt{can\_\_144} & 144 & 720 & 10 (FEAS) & 11 (FEAS) & -- (TO) & 13 (FEAS) & 3* (1394.24) & -- (TO) & -- (TO) & 4 (FEAS) \\
\texttt{lund\_a} & 147 & 1298 & 12 (FEAS) & -- (TO) & -- (TO) & 16 (FEAS) & 4 (FEAS) & -- (TO) & -- (TO) & 6 (FEAS) \\
\texttt{lund\_b} & 147 & 1294 & 16 (FEAS) & -- (TO) & -- (TO) & 15 (FEAS) & 4 (FEAS) & -- (TO) & -- (TO) & 6 (FEAS) \\
\texttt{bcsstk05} & 153 & 1288 & 16 (FEAS) & -- (TO) & -- (TO) & 12 (FEAS) & -- & -- & -- & -- \\
\texttt{west0156} & 156 & 371 & -- (TO) & -- (TO) & -- (TO) & -- (TO) & -- & -- & -- & -- \\
\texttt{can\_\_161} & 161 & 769 & 15 (FEAS) & -- (TO) & -- (TO) & 15 (FEAS) & -- & -- & -- & -- \\
\texttt{dwt\_\_162} & 162 & 672 & 12 (FEAS) & 17 (FEAS) & -- (TO) & 12 (FEAS) & -- & -- & -- & -- \\
\texttt{lop163} & 163 & 935 & 8 (FEAS) & 9 (FEAS) & -- (TO) & 10 (FEAS) & -- & -- & -- & -- \\
\texttt{west0167} & 167 & 507 & -- (TO) & -- (TO) & -- (TO) & -- (TO) & -- & -- & -- & -- \\
\texttt{mcca} & 180 & 2659 & -- (TO) & -- (TO) & -- (TO) & -- (TO) & -- & -- & -- & -- \\
\texttt{fs\_183\_1} & 183 & 1069 & -- (TO) & -- (TO) & -- (TO) & -- (TO) & -- & -- & -- & -- \\
\texttt{gre\_\_185} & 185 & 1005 & 15 (FEAS) & -- (TO) & -- (TO) & 19 (FEAS) & -- & -- & -- & -- \\
\texttt{can\_\_187} & 187 & 839 & 14 (FEAS) & -- (TO) & -- (TO) & 12 (FEAS) & -- & -- & -- & -- \\
\texttt{dwt\_\_193} & 193 & 1843 & -- (TO) & -- (TO) & -- (TO) & -- (TO) & -- & -- & -- & -- \\
\texttt{dwt\_\_198} & 198 & 795 & 13 (FEAS) & 16 (FEAS) & -- (TO) & 13 (FEAS) & -- & -- & -- & -- \\
\texttt{will199} & 199 & 701 & -- (TO) & -- (TO) & -- (TO) & -- (TO) & -- & -- & -- & -- \\
\texttt{bwm200} & 200 & 796 & 13 (FEAS) & -- (TO) & 1* (392.73) & 13 (FEAS) & -- & -- & -- & -- \\
\texttt{rdb200} & 200 & 1120 & 12 (FEAS) & -- (TO) & -- (TO) & 16 (FEAS) & -- & -- & -- & -- \\
\texttt{impcol\_a} & 207 & 572 & -- (TO) & -- (TO) & -- (TO) & -- (TO) & -- & -- & -- & -- \\
\texttt{wm1} & 207 & 2909 & -- (TO) & -- (TO) & -- (TO) & -- (TO) & -- & -- & -- & -- \\
\texttt{wm2} & 207 & 2942 & -- (TO) & -- (TO) & -- (TO) & -- (TO) & -- & -- & -- & -- \\
\texttt{wm3} & 207 & 2948 & -- (TO) & -- (TO) & -- (TO) & -- (TO) & -- & -- & -- & -- \\
\texttt{dwt\_\_209} & 209 & 976 & 20 (FEAS) & -- (TO) & -- (TO) & 19 (FEAS) & -- & -- & -- & -- \\
\texttt{fidap001} & 216 & 4374 & -- (TO) & -- (TO) & -- (TO) & -- (TO) & -- & -- & -- & -- \\
\texttt{gre\_216a} & 216 & 876 & 18 (FEAS) & -- (TO) & -- (TO) & 19 (FEAS) & -- & -- & -- & -- \\
\texttt{ash219} & 219 & 438 & 19 (FEAS) & -- (TO) & -- (TO) & -- (TO) & -- & -- & -- & -- \\
\texttt{dwt\_\_221} & 221 & 925 & 14 (FEAS) & -- (TO) & -- (TO) & 19 (FEAS) & -- & -- & -- & -- \\
\texttt{impcol\_e} & 225 & 1308 & -- (TO) & -- (TO) & -- (TO) & -- (TO) & -- & -- & -- & -- \\
\texttt{pde225} & 225 & 1065 & 14 (FEAS) & -- (TO) & -- (TO) & 14 (FEAS) & -- & -- & -- & -- \\
\texttt{can\_\_229} & 229 & 1003 & 21 (FEAS) & -- (TO) & -- (TO) & 19 (FEAS) & -- & -- & -- & -- \\
\texttt{dwt\_\_234} & 234 & 534 & 13 (FEAS) & -- (TO) & -- (TO) & 14 (FEAS) & -- & -- & -- & -- \\
\texttt{e05r0000} & 236 & 5856 & -- (TO) & -- (TO) & -- (TO) & -- (TO) & -- & -- & -- & -- \\
\texttt{nos1} & 237 & 627 & 19 (FEAS) & 16 (FEAS) & -- (TO) & 13 (FEAS) & -- & -- & -- & -- \\
\texttt{steam1} & 240 & 3762 & -- (TO) & -- (TO) & -- (TO) & -- (TO) & -- & -- & -- & -- \\
\texttt{dwt\_\_245} & 245 & 853 & 18 (FEAS) & -- (TO) & -- (TO) & 18 (FEAS) & -- & -- & -- & -- \\
\texttt{can\_\_256} & 256 & 1586 & -- (TO) & -- (TO) & -- (TO) & -- (TO) & -- & -- & -- & -- \\
\texttt{nnc261} & 261 & 1500 & 22 (FEAS) & -- (TO) & -- (TO) & 19 (FEAS) & -- & -- & -- & -- \\
\texttt{lshp\_265} & 265 & 1009 & 19 (FEAS) & -- (TO) & -- (TO) & 21 (FEAS) & -- & -- & -- & -- \\
\texttt{can\_\_268} & 268 & 1675 & -- (TO) & -- (TO) & -- (TO) & -- (TO) & -- & -- & -- & -- \\
\texttt{bcspwr04} & 274 & 943 & 19 (FEAS) & -- (TO) & -- (TO) & -- (TO) & -- & -- & -- & -- \\
\texttt{ash292} & 292 & 1250 & 14 (FEAS) & -- (TO) & -- (TO) & -- (TO) & -- & -- & -- & -- \\
\texttt{can\_\_292} & 292 & 1416 & -- (TO) & -- (TO) & -- (TO) & -- (TO) & -- & -- & -- & -- \\
\texttt{dwt\_\_307} & 307 & 1415 & -- (TO) & -- (TO) & -- (TO) & -- (TO) & -- & -- & -- & -- \\
\texttt{dwt\_\_310} & 310 & 1379 & 24 (FEAS) & -- (TO) & -- (TO) & -- (TO) & -- & -- & -- & -- \\
\texttt{dwt\_\_346} & 346 & 1786 & -- (TO) & -- (TO) & -- (TO) & -- (TO) & -- & -- & -- & -- \\
\texttt{dwt\_\_361} & 361 & 1657 & -- (TO) & -- (TO) & -- (TO) & -- (TO) & -- & -- & -- & -- \\
\texttt{plat362} & 362 & 3074 & -- (TO) & -- (TO) & -- (TO) & -- (TO) & -- & -- & -- & -- \\
\end{longtable}
\endgroup

\section*{Declaration of generative AI and AI-assisted technologies in the manuscript preparation process}

During the preparation of this work, the authors used ChatGPT by OpenAI to assist with language revision. After using this tool, the authors reviewed and edited the content as needed and take full responsibility for the content of the published article.

\bibliographystyle{cas-model2-names}
\bibliography{cas-refs}

\end{document}